\PassOptionsToPackage{dvipsnames}{xcolor}
\documentclass[acmsmall,10pt,screen]{acmart} 
\setcopyright{rightsretained} 
\acmJournal{PACMPL}
\acmYear{2020} \acmVolume{4} \acmNumber{OOPSLA} \acmArticle{149} \acmMonth{11} \acmPrice{}\acmDOI{10.1145/3428217}

\usepackage{amsthm} \usepackage{proof}
\usepackage{thm-restate}
\usepackage{mathpartir}
\usepackage{booktabs}
\usepackage{amsmath}
\usepackage{balance}
\usepackage{nicefrac}
\usepackage{url}
\usepackage{xspace}
\usepackage{hyperref}
\usepackage{todonotes}
\usepackage[normalem]{ulem}  \usepackage{suffix}
\usepackage[inline]{enumitem}  \usepackage{afterpage}
\usepackage{stmaryrd} \usepackage{ifthen}

\usepackage[utf8]{inputenc}
\DeclareUnicodeCharacter{1438}{\pa}
\DeclareUnicodeCharacter{1433}{\po}
\DeclareUnicodeCharacter{1428}{\stroke}

\usepackage[final]{listings}
\usepackage{multicol}
\usepackage{golang}     

\newcommand{\lstCodeSize}{\footnotesize\ttfamily} 

\lstdefinestyle{mygo}{language=Go,
moredelim=[is][\color{gray}]{``}{``},
}

\usepackage{tikz}  \usetikzlibrary{arrows,shapes,calc,fit,patterns,backgrounds,trees,arrows,positioning,shapes.multipart}

\theoremstyle{acmplain}

\newtheorem{theorem}{Theorem}[section]
\newtheorem{lemma}[theorem]{Lemma}

\newif\iflong 

\longtrue

\newenvironment{Figure}{
  \begin{figure}
}{
\end{figure}
}
\newenvironment{ruled}{
  \hrule
  \vspace{1ex}
}{
  \vspace{1ex}
  \hrule
  \vspace{-2ex}
}
\newenvironment{ruledgo}{
  \hrule
  \vspace{-2.5ex}
}{
  \vspace{-1.25ex}
  \hrule
  \vspace{-2ex}
}
\newenvironment{ruledGo}{
  \hrule
  \vspace{-2.5ex}
}{
  \vspace{-3.5ex}
  \hrule
  \vspace{-2ex}
}

\newcommand{\ROUNDTWO}[1]{#1}

\newcommand{\squeeze}{\hspace{-1.5em}}
\newcommand{\gap}{\hspace{0.75em}}
\newcommand{\Strut}{\vphantom{\ov{f}}}
\newcommand{\calD}{\mathcal{D}}
\newcommand{\calS}{\mathcal{S}}
\newcommand{\ov}{\overline}

\newcommand{\kw}[1]{\texttt{\bf #1}}
\newcommand{\id}[1]{\texttt{#1}}
\newcommand{\meta}{\mathit}
\newcommand{\ok}{~\meta{ok}}

\newcommand{\fields}{\meta{fields}}
\newcommand{\methods}{\meta{methods}}
\newcommand{\vtype}{\meta{type}}
\newcommand{\mbody}{\meta{body}}
\newcommand{\tdecls}{\meta{tdecls}}
\newcommand{\mdecls}{\meta{mdecls}}
\newcommand{\bounds}{\meta{bounds}}

\newcommand{\unique}{\meta{unique}}
\newcommand{\distinct}{\meta{distinct}}
\newcommand{\type}{\kw{type}}
\newcommand{\struct}{\kw{struct}}
\newcommand{\interface}{\kw{interface}}
\newcommand{\func}{\kw{func}}
\newcommand{\return}{\kw{return}}
\newcommand{\package}{\kw{package}}
\newcommand{\main}{\kw{main}}

\newcommand{\un}{\id{\textunderscore}}
\newcommand{\prog}{\rhd}
\newcommand{\br}[1]{\id{\{}#1\id{\}}}
\newcommand{\lst}[1]{[#1]}
\newcommand{\set}[1]{\{#1\}}
\newcommand{\an}[1]{\langle #1 \rangle}
\newcommand{\imp}{\mathbin{\id{<:}}}
\newcommand{\notimp}{\mathbin{\not\!\!\imp}}
\newcommand{\becomes}{\longrightarrow}
\newcommand{\by}{\mathbin{:=}}
\newcommand{\gray}[1]{{\color{gray}#1}}
\newcommand{\black}[1]{{\color{black}#1}}
\newcommand{\comma}{,\,}
\newcommand{\stoup}{;\,}
\newcommand{\Hole}{\Box}

\newcommand{\PANIC}
{\ensuremath{\mathsf{panic}}}

\newcommand{\goinl}[1]{\lstinline[style=mygo]+#1+}
\lstnewenvironment{golst}{\lstset{style=mygo}}{}
\lstnewenvironment{golsttwo}{\lstset{style=mygo,multicols=2}}{}

\newcommand{\yields}{\blacktriangleright}

\newcommand{\extensionkw}{closure}

\newcommand{\Sclo}{\textit{S-\extensionkw}}
\newcommand{\Iclo}{\textit{I-\extensionkw}}
\newcommand{\Fclo}{\textit{F-\extensionkw}}
\newcommand{\Mclo}{\textit{M-\extensionkw}}

\newcommand{\SExtensionD}[2]{\Sclo_{#2}(#1)}
\newcommand{\IExtensionD}[2]{\Iclo_{#2}(#1)}
\newcommand{\FExtensionD}[2]{\Fclo(#1)}
\newcommand{\MExtensionD}[2]{\Mclo_{#2}(#1)}

\newcommand{\goodsub}{good}
\newcommand{\eo}{\rho}
\DeclareMathOperator{\sorder}{\gtrdot}

\newcommand{\notmonomorphisable}{~{\notmonomorphisable*}}
\WithSuffix\newcommand\notmonomorphisable*{\meta{nomono}}  \newcommand{\mkdummy}{\meta{hash}}
\newcommand{\naturals}{\mathbb{N}}
\newcommand{\dummytype}{\mathsf{Top}}

\newcommand{\Ga}{\Gamma}
\newcommand{\TENV}{\ensuremath{\Delta}}
\newcommand{\RULENAME}[1]{\textsc{#1}}
\newcommand{\tyrulename}[1]{\RULENAME{#1}}
\newcommand{\EQ}{\ensuremath{\mathbin{=}}}

\newcommand{\vvv}{\ensuremath{v}}  \newcommand{\SSS}{\ensuremath{S}}  \newcommand{\III}{\ensuremath{I}}  

\newcommand{\ttt}{\ensuremath{t}}  \newcommand{\tttS}{\ensuremath{\ttt_\SSS}}  \newcommand{\tttI}{\ensuremath{\ttt_\III}}  \newcommand{\fff}{\ensuremath{f}}  \newcommand{\mmm}{\ensuremath{m}}  \newcommand{\xxx}{\ensuremath{x}}  \newcommand{\yyy}{\ensuremath{y}}  

\newcommand{\eee}{\ensuremath{e}}  

\newcommand{\SUBS}[2]{\ensuremath{[#2 \by #1]}}

\newcommand{\PANICe}[1]{\ensuremath{#1\,\PANIC}}

\DeclareMathOperator{\IMPLOP}{<:}
\newcommand{\IMPLS}[2]{\ensuremath{#1 \IMPLOP #2}}
\newcommand{\IMPLSG}[3]{\ensuremath{#1 \vdash #2 \IMPLOP #3}}
\newcommand{\NIMPLS}[2]{\ensuremath{#1 \not\!\!\IMPLOP #2}}
\newcommand{\NIMPLSG}[3]{\ensuremath{#1 \vdash #2 \not\!\!\IMPLOP #3}}
\newcommand{\impls}[2]{\IMPLS{#1}{#2}}

\newcommand{\METHSOFt}[1]{\methods(#1)}
\newcommand{\FIELDSOFt}[1]{\fields(#1)}
\newcommand{\FIELDSOFtD}[1]{\FIELDSOFt{#1}}  \newcommand{\methodsof}[1]{\METHSOFt{#1}}
\newcommand{\methodsofD}[2]{\methods_{#1}(#2)}

\newcommand{\MBODYDmt}[3]{\mbody_{#1}(#3.#2)}
\newcommand{\MBODYmt}[2]{\mbody(#2.#1)}
\newcommand{\MBODYxxe}[3]{\ensuremath{(#1, #2) . #3}} 

\newcommand{\MBODYmtN}[4]{\MBODYDmt{#1}{#2(#3)}{#4}}

\newcommand{\RETe}[1]{\ensuremath{\kw{return}\,#1}}

\newcommand{\MDECLpmGe}[4]
{\ensuremath{\kw{func}\,(#1)\,#2\,#3\,\{\RETe{#4}\}}}

\newcommand{\JUDGEWF}[1]{\ensuremath{#1 \, \ok}}

\newcommand{\JUDGEEXPR}[3]{\ensuremath{#1 \vdash #2 : #3}}

\newcommand{\JUDGEWFG}[2]{\ensuremath{#1 \vdash #2 \, \ok}}
\newcommand{\JUDGEFORMCTXT}[3]{\ensuremath{#1 ; #2 \ok~#3}}
\newcommand{\JUDGEEXPRG}[4]{\ensuremath{#1; #2 \vdash #3 : #4}}

\newcommand{\kkk}{\ensuremath{\tvar}}
\newcommand{\tvar}{\ensuremath{\alpha}}

\newcommand{\BOUNDTENVt}[2]{\ensuremath{\bounds_{#1}(#2)}}

\newcommand{\SUBSTBUILDER}[3]{#1 = (#2 \by #3)}
\newcommand{\SUBSTBUILDERCTXT}[4]{#1 = (#3 \by_{#2} #4)}

\newcommand{\NTYPEtt}[2]{\ensuremath{#1(#2)}}
\newcommand{\TYPINSTtn}[2]{\NTYPEtt{#1}{#2}}

\newcommand{\GPARAMSk}[1]{\ensuremath{(\kw{type}\,#1)}}
\newcommand{\PDECLkN}[2]{\ensuremath{#1\,#2}}

\newcommand{\pDECLxt}[2]{\ensuremath{#1\,#2}}

\newcommand{\MSIGpt}[2]{\ensuremath{(#1)\,#2}}
\newcommand{\METHmG}[2]{\ensuremath{#1\,#2}}

\newcommand{\MSIGPpt}[3]{\ensuremath{#1\,(#2)\,#3}}

\newcommand{\METHmpt}[3]{\METHmG{#1}{\MSIGpt{#2}{#3}}}
\newcommand{\METHODmpt}[3]{\METHmpt{#1}{#2}{#3}}

\newcommand{\STRUCTte}[2]{\ensuremath{#1\br{#2}}}

\newcommand{\SELef}[2]{\ensuremath{#1.#2}}

\newcommand{\FDECLft}[2]{\ensuremath{#1\,#2}}

\newcommand{\MCALLeme}[3]{\ensuremath{#1.#2(#3)}}

\newcommand{\MCALLemne}[4]{\ensuremath{#1.#2(#3)(#4)}}

\newcommand{\ASSRTet}[2]{\ensuremath{#1.(#2)}}

\newcommand{\REDUCE}{\becomes}

\newcommand{\CS}{C\nolinebreak\hspace{-.05em}\raisebox{.5ex}{\footnotesize\bf \#}}
\newcommand{\Cpp}{C\nolinebreak\hspace{-.05em}\raisebox{.4ex}{\footnotesize\bf +}\nolinebreak\hspace{-.10em}\raisebox{.4ex}{\footnotesize\bf +}}

\newcommand{\FV}[1]{\ensuremath{\mathit{fv}(#1)}}

\newcommand{\monoid}[1]{\an{#1}}

\title{Featherweight Go}
\author{Robert Griesemer}
\orcid{}
\affiliation{\institution{Google}
\country{USA}
}
\author{Raymond Hu}
\orcid{0000-0003-4361-6772}
\affiliation{\institution{University of Hertfordshire}
  \department{School of Engineering and Computer Science}
\city{Hatfield}
\country{UK}
}
\email{r.z.h.hu@herts.ac.uk}

\author{Wen Kokke}
\orcid{0000-0002-1662-0381}
\affiliation{\institution{University of Edinburgh}
  \department{Laboratory for Foundations of Computer Science}
\city{Edinburgh}
\country{UK}
}
\email{wen.kokke@ed.ac.uk}

\author{Julien Lange}
\orcid{0000-0001-9697-1378}
\affiliation{\institution{Royal Holloway, University of London}
  \department{Department of Computer Science}
\city{Egham}
\country{UK}
}
\email{julien.lange@rhul.ac.uk}

\author{Ian Lance Taylor}
\orcid{}
\affiliation{\institution{Google}
\country{USA}
}

\author{Bernardo Toninho}
\orcid{0000-0002-0746-7514}
\affiliation{
  \department{Departamento de Inform\'{a}tica}
  \institution{NOVA-LINCS, FCT-NOVA, Universidade Nova de Lisboa}
\country{Portugal}}
\email{btoninho@fct.unl.pt}

\author{Philip Wadler}
\orcid{0000-0001-7619-6378}
\affiliation{\institution{University of Edinburgh}
  \department{School of Informatics}
\city{Edinburgh}
\country{UK}
}
\email{wadler@inf.ed.ac.uk}
\author{Nobuko Yoshida}
\orcid{0000−0002−3925−8557}
\affiliation{\institution{Imperial College London}
  \department{Computing}
\city{London}
\country{UK}
}
\email{n.yoshida@imperial.ac.uk}

\begin{abstract}

  We describe a design for generics in Go inspired by previous
  work on Featherweight Java by Igarashi, Pierce, and Wadler.  Whereas
  subtyping in Java is nominal, in Go it is structural, and whereas
  generics in Java are defined via erasure, in Go we use
  monomorphisation.  Although monomorphisation is widely used, we are
  one of the first to formalise it.  Our design also supports a solution to
  The Expression Problem.

\end{abstract}



\begin{document}

\maketitle
\renewcommand{\shortauthors}{Griesemer \emph{et al}.}

\section{Introduction}

Google introduced the Go programming language in 2009
\citep{Griesemer-et-al-2009,The-Go-Team-2020}.  Today it sits at
position 12 on the Tiobe Programming Language Index
\ROUNDTWO{and position 10 on the IEEE Spectrum Programming Language
  Ranking}
(Haskell sits at positions 41 \ROUNDTWO{and 29}, respectively).
Recently, the Go team mooted a design to extend Go with generics
\citep{Taylor-and-Griesemer-2019}, and
Rob Pike wrote Wadler to ask:
\begin{quote}
 Would you be interested in helping us get polymorphism right (and/or
 figuring out what ``right'' means) for some future version of Go?
\end{quote}
This paper is our response to that question.

Two decades ago, Igarashi, Pierce, and Wadler
[\citeyear{Igarashi-et-al-1999,Igarashi-et-al-2001}], introduced
Featherweight Java.  They considered a tiny model of Java (FJ),
extended that model with generics (FGJ), and translated FGJ to FJ
(via \emph{erasure}).  In their footsteps, we introduce Featherweight
Go.  We consider a tiny model of Go (FG), extend that model with
generics (FGG), and translate FGG to FG (via \emph{monomorphisation}).

Go differs in interesting ways from Java.
Subtyping in Java is nominal, whereas in Go it is structural.
\ROUNDTWO{Casts in Java correspond to type assertions in Go:
casts in Java with generics are restricted to support erasure},
whereas type assertions in Go with generics are unrestricted
thanks to monomorphisation.  Monomorphisation is widely used,
but we are among the first to formalise it.  The Expression Problem
was first formulated by \citet{Wadler-1998} in the context of Java,
though Java never supported a solution; our design does.

We provide a full formal development:
for FG and FGG, type and reduction rules, and preservation and progress;
for monomorphisation, a formal translation from FGG to FG
that preserves types and is a bisimulation.
{\iflong
  The appendices contain complete proofs\else
  The complete proofs are available in~\cite{techreport}\fi}.

\paragraph{Structural subtyping}
Go is based on structures and interface types.
Whereas most programming languages use
\emph{nominal subtyping}, Go is unique among mainstream
typed programming languages in using
\emph{structural subtyping}. A structure implements an interface
if it defines all the methods specified by that interface,
and an interface implements another if the methods of the first
are a superset of the methods specified by the second.

In Java, the superclasses of a class are fixed by the declaration.
If lists are defined before collections, then one
cannot retrofit collections as a superclass of lists---save by
rewriting and recompiling the entire library.
In Haskell the superclasses of a type class are fixed by the
declaration. If monads are defined before functors, then one
cannot retrofit functors as a superclass of monads---save by
rewriting and recompiling the entire library.
In contrast, in Go one might define lists or monads first, and
later introduce collections or functors as an interface that
the former implements---without rewriting or
recompiling the earlier code.

\paragraph{The Expression Problem}
The Expression Problem was formulated by \citet{Wadler-1998}.
It gave a name to issues described by
\citet{Cook-1990}, \citet{Reynolds-1994}, and
\citet{Krishnamurthi-et-al-1998},
and became the basis of subsequent work by
\citet{Torgersen-2004}, \citet{Zenger-and-Odersky-2004},
\citet{Swierstra-2008}, and many others.
Wadler defines The Expression Problem this way:
\begin{quote}
  The goal is to define a data type by cases, where one can add new
  cases to the data type and new functions over the data type, without
  recompiling existing code, and while retaining static type safety.
\end{quote}
And motivates its interest as follows:
\begin{quote}
  Whether a language can solve the Expression Problem is a salient
  indicator of its capacity for expression.  One can think of cases as
  rows and functions as columns in a table.  In a functional language,
  the rows are fixed (cases in a datatype declaration) but it is easy
  to add new columns (functions).  In an object-oriented language, the
  columns are fixed (methods in a class declaration) but it is easy to
  add new rows (subclasses).  We want to make it easy to add either
  rows or columns.
\end{quote}
One can come close to solving The Expression Problem in Go as it
exists now, using dynamic checking via type assertions. We show how to
provide a fully static solution with generics.
\ROUNDTWO{We had to adjust our design: our first design
for generics used nonvariant matching on bounds in receivers of
methods, but to solve The Expression Problem we had to relax this
to covariant matching.}

\paragraph{Monomorphisation}
FGJ translates to FJ via \emph{erasure},
whereas FGG translates to FG via \emph{monomorphisation}.
Two instances \goinl{List<int>} and \goinl{List<bool>} in FGJ
both translate to \goinl{List} in FJ
(where \goinl{<>} are punctuation),
whereas two instances \goinl{List(int)} and \goinl{List(bool)} in FGG
translate to separate types \goinl{List<int>} and \goinl{List<bool>} in FG
(where \goinl{()} are punctuation, but \goinl{<>} are taken as part of the name).

Erasure is more restrictive than monomorphisation.  In Java with
generics, a cast \goinl{(a)x} is illegal if \goinl{a} is a type
variable, whereas in Go with generics, the equivalent type assertion
\goinl{x.(a)} is permitted.  Erasure is often less efficient
than monomorphisation.  In Java with generics, all type
variables are boxed, whereas in Go with generics 
type variables may instantiate to be unboxed. 
However, erasure is linear in the size
of the code, whereas monomorphisation can suffer an exponential
blowup; and erasure is suited to separate compilation, whereas
monomorphisation requires the whole program.
\ROUNDTWO{We choose to look at monomorphisation in the first instance,
  because it is simple, efficient, and the first design looked at by the Go team.
Other designs offer other tradeoffs, e.g.,
the restriction on a type assertion \goinl{x.(a)} could also
be avoided by passing runtime representations of types. This
solution avoids exponential blowup and offers better support
for separate compilation, but at a cost in efficiency
and complexity. We expect the final solution will involve a mix
of both monomorphisation and runtime representations of types,
see Section~\ref{sec:conclusion} for more details.}

Template expansion in \Cpp{}~\cite[Chapter~26]{Stroustrup-2013}
corresponds to monomorphisation.
Generics in .NET are implemented by a mixture of erasure and monomorphisation
\cite{Kennedy-and-Syme-2001}.
The MLton compiler for Standard ML
\citep{Cejtin-et-al-2000} and the Rust programming language
\citep{Rust-Team-2017} both apply techniques closely related to
monomorphisation, as described by \citet{Fluet-2015} and
\citet{Turon-2015} on web pages and blog posts.  We say more about
related work in Section~\ref{sec:related}, but we have found
only a handful of peer-reviewed publications that touch on
formalisation of monomorphisation.
Monomorphisation is possible only when it requires a finite set
of instances of types and methods.  We believe we are
the first to formalise computation of instance sets
and determination of whether they are finite.

The bookkeeping required to formalise monomorphisation of instances
and methods is not trivial.  Monomorphising an interface with type
parameters that contains a method with type parameters may require
different instances of the interfaces to contain different instances
of the methods. It took us several tries over many months to formalise it
correctly. While the method for monomorphisation described here is
specialised to Go, we expect it to be of wider interest, since similar
issues arise for other languages and compilers such as \Cpp{}, .Net,
MLton, or Rust.

\paragraph{Featherweight vs complete}
A reviewer of an earlier revision of this paper wrote:
\begin{quotation}
It is also quite common for semantics to strive for ``completeness'',
instead of being ``featherweight''. There is a lot of value in having
featherweight semantics, but the argument for completeness is that
it helps language designers understand bad interactions between
features.
(For example, \citet{Amin-and-Tate-2016} recently showed that Java
generics are unsound, but the bug is beyond the scope of
Featherweight Generic Java.)
\end{quotation}
We agree with these words. Since the review was a reject, we
deduce an implicit claim that it is better to be complete.
Here, with respect, we disagree.
We argue both ``featherweight'' and ``complete''
descriptions have value.  As evidence, compare citations
counts for the paper on Featherweight Java,
\citet{Igarashi-et-al-2001}, with the four most-cited papers
on more complete models, 
\citet{Flatt-et-al-1998,Nipkow-and-von-Oheimb-1998,Drossopoulou-and-Eisenbach-1997,Syme-1999}:
1070 as compared with 549, 248, 174, 158, respectively
(Google Scholar, April 2020).

\paragraph{Impact}
The original proposal for generics in Go
\citep{Taylor-and-Griesemer-2019} was based on \emph{contracts}, which
are syntactically convenient but lack a clear semantics.  One result
of our work is that the new proposal \citep{Taylor-and-Griesemer-2020}
is based on \emph{interfaces}, which are already well defined in Go.
After we submitted the draft of this paper, Griesemer wrote to Wadler:
\begin{quote}
I want to thank you and your team for all the type theory work on Go
so far---it really helped clarify our understanding to a massive
degree. So thanks!
\end{quote}
Another result is the proposal for covariant receiver typing,
a feature required by The Expression Problem. It is not part of the
Go team's current design, but they have noted it is backward
compatible and are considering adding it in the future.

In this paper we adopt the syntax originally proposed by the Go team
in July 2019~\cite{Taylor-and-Griesemer-2019}. In September
2020~\cite{Taylor-and-Griesemer-2020}, they proposed a revised syntax
where type parameters are declared within square brackets and the
$\type$ keyword is omitted.
We prefer the new syntax, but retained the
old since our artifact uses the old syntax and artifact evaluation was
already complete.

\paragraph{Outline}
Section~\ref{sec:examples} introduces FG and FGG,
presents a solution to The Expression Problem,
and introduces monomorphisation.
Sections~\ref{sec:fg} and~\ref{sec:fgg} present FG and FGG;
we give formal rules for types and reductions, and prove preservation and progress.
Section~\ref{sec:mono} presents monomorphisation, which translates FGG back to FG;
we prove the translation preserves types and is a bisimulation.
Section~\ref{sec:implementation} describes our prototype implementation.
Section~\ref{sec:related} describes related work.
Section~\ref{sec:conclusion} concludes.
{\iflong
  Appendices provide extra examples and details of all proofs\else
  Additional examples and details of all proofs are available
  in~\cite{techreport}\fi}.

 \section{Featherweight Go by example}
\label{sec:examples}

Formally, FG and FGG are tiny languages, containing only structures,
interfaces, and methods. Our examples use features of
Go missing in FG and FGG, including booleans, integers, strings, and
variable bindings.
We show how to declare booleans in FG and FGG in {\iflong
  Appendices~\ref{sec:fg-bool} and~\ref{sec:fgg-bool}\else
  \cite{techreport}\fi}.

\subsection{FG by example}
\label{sec:fg-ex}

\paragraph{Functions in FG}
Figure~\ref{fig:fg-function} shows higher-order functions in FG.

Interface \goinl{Any} has no methods, and so is implemented
by any type.  Interface \goinl{Function} has a single method,
\goinl{Apply(x Any) Any}, which has an argument and result of
type \goinl{Any}.  It is implemented by any structure
that defines a method with the same name and same signature.
In Go structures and methods are declared separately, as compared to
Java where they are grouped together in a class declaration.
We give three examples.

Structure \goinl{incr} has a single field, \goinl{n}, of type
\goinl{int}.  Its \goinl{Apply} method has receiver \goinl{this} of
type \goinl{incr}, argument \goinl{x} of type \goinl{Any}, and result
type \goinl{Any}, and increments its argument by \goinl{n}.  You might
expect the argument and result to instead have type \goinl{int}, but
then the declared method would not implement the \goinl{Function}
interface, because the method name and signature must match exactly.
In the method's body, \goinl{x.(int)} is a \emph{type assertion} that
checks its argument is an integer; otherwise it \emph{panics}, which is
Go jargon for raising a runtime error.
A structure is created by a \emph{literal},
consisting of the structure name and its field values in braces.
For instance, \goinl{incr\{-5\}.Apply(3)} returns \goinl{-2}.
A field of a structure is accessed in the usual way, \goinl{this.n}.
Here \goinl{this} is a variable bound to the receiver, not a keyword.

Structure \goinl{pos} has no fields, and its apply method returns \goinl{true}
if given a positive integer. Structure \goinl{compose} has two fields,
each of which is a function, and its apply method applies the
two in succession.  The top-level \goinl{main}
method composes \goinl{incr\{-5\}} with \goinl{pos\{\}} and
applies it to \goinl{3}, yielding false.  One cannot pass a value of
type \goinl{Any} where a boolean is expected, so the type assertion
\goinl{.(bool)} is required.

Bound variable names are irrelevant when comparing method
signatures, but method names and type names must match exactly.
For example, the following signatures are considered equivalent:
\begin{center}
\text{\goinl{Apply(x Any) Any}}
\qquad
\qquad
\qquad
\text{\goinl{Apply(arg Any) Any}}
\end{center}

\begin{Figure}
\begin{ruledGo}
\begin{golsttwo}
type Any interface {}
type Function interface {
  Apply(x Any) Any
}
type incr struct { n int }
func (this incr) Apply(x Any) Any {
  return x.(int) + this.n
}
type pos struct {}
func (this pos) Apply(x Any) Any {
  return x.(int) > 0
}
(* \newpage *)
type compose struct {
  f Function
  g Function
}
func (this compose) Apply(x Any) Any {
  return this.g.Apply(this.f.Apply(x))
}
func main() {
  var f Function = compose{incr{-5},pos{}}
  var _ bool = f.Apply(3).(bool) // false
}
\end{golsttwo}
\end{ruledGo}
\caption{Functions in FG}
\label{fig:fg-function}

\vspace{2ex}
\begin{ruledGo}
\begin{golsttwo}
type Eq interface {
  Equal(that Eq) bool
}
type Int int
func (this Int) Equal(that Eq) bool {
  return this == that.(Int)
}
type Pair struct {
  left Eq
  right Eq
}
(* \newpage *)
func (this Pair) Equal(that Eq) bool {
  return this.left.Equal(that.(Pair).left) &&
    this.right.Equal(that.(Pair).right)
}
func main() {
  var i, j Int = 1, 2
  var p Pair = Pair{i, j}
  var _ bool = p.Equal(p) // true
}
\end{golsttwo}
\end{ruledGo}
\caption{\ROUNDTWO{Equality in FG}}
\label{fig:fg-eq}

\vspace{2ex}
\begin{ruledGo}
\begin{golsttwo}
type List interface {
  Map(f Function) List
}
type Nil struct {}
type Cons struct {
  head Any
  tail List
}
func (xs Nil) Map(f Function) List {
  return Nil{}
}
(* \newpage *)
func (xs Cons) Map(f Function) List {
  return Cons{f.Apply(xs.head), xs.tail.Map(f)}
}
func main() {
  var xs List = Cons{3, Cons{6, Nil{}}}
  var ys List = xs.Map(incr{-5})
    // Cons{-2, Cons{1, Nil{}}}
  var _ List = ys.Map(pos{})
    // Cons{false, Cons{true, Nil{}}}
}
\end{golsttwo}
\end{ruledGo}
\caption{Lists in FG}
\label{fig:fg-list}
\end{Figure}

\ROUNDTWO{
\paragraph{Equality in FG}
Figure~\ref{fig:fg-eq} shows equality in FG.

Interface \goinl{Eq} has one method with signature
\goinl{Equal(that Eq) bool}.
If a type implements this interface we say it supports equality.

A type declaration introduces \goinl{Int} as a synonym for integers,
and a method declaration ensures that type supports
equality. Since signatures must match exactly, in the method the
argument has type \goinl{Eq} and the body uses a type assertion to
convert it to \goinl{Int} as required.

A second type declaration introduces the structure \goinl{Pair}
with two fields \goinl{left} and \goinl{right} which may be of
any type that supports equality, and the method declaration
ensures that pairs themselves support equality.  Again, the
argument has type \goinl{Eq} and the body uses a type assertion
to convert it to a \goinl{Pair} as required.
The top level \goinl{main} method builds a pair of integers
and compares it to itself for equality, yielding \goinl{true}.

Since pairs are to support equality, their components are
also required to support equality. In general, if a structure is
to satisfy some interface we may need to require that each field
of that structure satisfies the same interface---a property we
refer to as \emph{type pollution}.
An alternative design would
give fields the type \goinl{Any}, and to replace \goinl{this.left}
by \goinl{this.left.(Eq)}, and similarly for the other component.
The alternative design is more flexible---it permits fields of
the pair to have any type---but is less efficient
(the type assertions must be checked at runtime)
and less reliable (the type assertions may fail). 
As we will see, FGG will let us avoid type pollution,
providing flexibility, efficiency, and reliability all at the same time.
}

\paragraph{Lists in FG}
Figure~\ref{fig:fg-list} shows lists in FG.

Interface \goinl{List} has a single method:
\goinl{Map(f Function) List},
which applies its argument to each element of its receiver.
We define two structures that implement the list interface. Structure
\goinl{Nil} has no fields, while structure \goinl{Cons} has two
fields, a head of any type and a tail which is a list.
The methods to define \goinl{Map} are straightforward, and the main
method shows an example of its use.

Go is designed to enable efficient implementation.  Structures are
laid out in memory as a sequence of fields, while an interface is a
pair of a pointer to an underlying structure and a
pointer to a dictionary of methods. To ensure the
layout of a structure is finite, a structure that recurses on itself
is forbidden. Thus, the declaration
\begin{center}
\goinl{type Bad struct \{ oops Bad \}}
\end{center}
is not allowed, and similarly for mutual recursion.
\ROUNDTWO{However,
structures that recurse through an interface are permitted, such as}
\begin{center}
\goinl{type Cons struct \{ head Any; tail List \}}
\end{center}
\ROUNDTWO{where the \goinl{tail} field of type \goinl{List}
may itself contain a \goinl{Cons}, since \goinl{Cons}
implements interface \goinl{List}.}

\subsection{FGG by Example}
\label{sec:fgg-ex}

We now adapt the examples of the previous section to generics.

\paragraph{Functions in FGG}
Figure~\ref{fig:fgg-func} shows higher-order functions in FGG.

The interface for functions now takes two type parameters,
\goinl{Function(type a Any, b Any)}.  Each type parameter is followed
by an interface it must implement, called its \emph{bound}.
Here the bounds indicate that the argument and result may be
of any type.  The signature for the method is now
\goinl{Apply(x a) b}, where the first type parameter is the argument
type and the second the result type.

Structures \goinl{incr} and \goinl{pos} are as before.  However, they
now have more natural signatures for their apply methods, where all
occurrences of \goinl{Any} are replaced by \goinl{int} or
\goinl{bool} as appropriate.  Type
assertions in the method bodies are no longer needed, and
the types ensure a panic never occurs.

The structure for composition now takes three type parameters.
In the \ROUNDTWO{\goinl{main} method}, type parameters are added and the type
assertion at the end is no longer required.

\begin{Figure}
\begin{ruledgo}
\begin{golsttwo}
type Function(type a Any, b Any) interface {
  Apply(x a) b
}
type incr struct { n int }
func (this incr) Apply(x int) int {
  return x + this.n
}
type pos struct {}
func (this pos) Apply(x int) bool {
  return x > 0
}
(* \newpage *)
type compose(type a Any, b Any, c Any) struct {
  f Function(a, b)
  g Function(b, c)
}
func (this compose(type a Any, b Any, c Any))
    Apply(x a) c {
  return this.g.Apply(this.f.Apply(x))
}
func main() {
  var f Function(int,bool) =
    compose(int,int,bool){incr{-5},pos{}}
  var _ bool = f.Apply(3) // false
}
\end{golsttwo}
\end{ruledgo}
\caption{Functions in FGG}
\label{fig:fgg-func}

\vspace{2ex}
\begin{ruledGo}
\begin{golsttwo}
type Eq(type a Eq(a)) interface {
  Equal(that a) bool
}
type Int int
func (this Int) Equal(that Int) bool {
  return this == that
}
type Pair(type a Any, b Any) struct {
  left a
  right b
}
(* \newpage *)
func (this Pair(type a Eq(a), b Eq(b)))
    Equal(that Pair(a,b)) bool {
  return this.left.Equal(that.left) &&
    this.right.Equal(that.right)
}
func main() {
  var i, j Int = 1, 2
  var p Pair(Int, Int) = Pair(Int, Int){i, j}
  var _ bool = p.Equal(p) // true
}
\end{golsttwo}
\end{ruledGo}
\caption{\ROUNDTWO{Equality in FGG}}
\label{fig:fgg-eq}

\vspace{2ex}
\begin{ruledGo}
\begin{golsttwo}
type List(type a Any) interface {
  Map(type b Any)(f Function(a, b)) List(b)
}
type Nil(type a Any) struct {}
type Cons(type a Any) struct {
  head a
  tail List(a)
}
func (xs Nil(type a Any))
    Map(type b Any)(f Function(a,b)) List(b) {
  return Nil(b){}
}
(* \newpage *) 
func (xs Cons(type a Any))
    Map(type b Any)(f Function(a,b)) List(b) {
  return Cons(b)
    {f.Apply(xs.head), xs.tail.Map(b)(f)}
}
func main() {
  var xs List(int) =
    Cons(int){3, Cons(int){6, Nil(int){}}}
  var ys List(int) = xs.Map(int)(incr{-5})
  var _ List(bool) = ys.Map(bool)(pos{})
}
\end{golsttwo}
\end{ruledGo}
\caption{Lists in FGG}
\label{fig:fgg-list}

\vspace{2ex}
\begin{ruledGo}
\begin{golsttwo}
type Edge(type e Edge(e,v),
               v Vertex(e,v)) interface {
   Source() v
   Target() v
}
(* \newpage *)
type Vertex(e Edge(e,v),
            v Vertex(e,v)) interface {
   Edges() List(e)
}
\end{golsttwo}
\end{ruledGo}
\caption{Mutually recursive bounds}
\label{fig:mutual}
\end{Figure}

\ROUNDTWO{
\paragraph{Equality in FGG}
Figure~\ref{fig:fgg-eq} shows equality in FGG.

The interface for equality is now written \goinl{Eq(type a Eq(a))}.
It accepts a type parameter \goinl{a} where the bound
is itself \goinl{Eq(a)}. The method has signature
\goinl{Equal(that a) bool}.  
The situation where a type parameter appears in
its own bound is known as \emph{F-bounded polymorphism}
\citep{Canning-et-al-1989}, and a similar idiom is common in Java with
generics \citep{Bracha-et-al-1998,Naftalin-and-Wadler-2006}.

\ROUNDTWO{Since we use a type parameter for the argument to \goinl{Equal},
in the method declaration for \goinl{Int}
the argument must now have type \goinl{Int} instead of type \goinl{Eq}.
A type assertion in the method body is no longer required,
increasing efficiency and reliability.}

The type declaration for \goinl{Pair} now take two type parameters,
\goinl{a} and \goinl{b}, which are both bounded by \goinl{Any}.
The method declaration for equality on pairs also uses two type
parameters, \goinl{a} and \goinl{b}, bounded by
\goinl{Eq(a)} and \goinl{Eq(b)} respectively,
so the call to \goinl{Equality} on the components of the pair
is permitted.

Crucially, FGG permits the
bounds on the type parameter in a receiver to \emph{implement} the
bound on the type parameter in the corresponding structure declaration
(receiver type parameters are \emph{covariant}).  This is in
contrast to method signatures, which must exactly match the signature
in the interface (signatures are \emph{nonvariant}).
In this case, the bounds on \goinl{a} and \goinl{b} in the type
declaration for pairs are both \goinl{Any}, while the bounds on \goinl{a}
and \goinl{b} in the receiver are \goinl{Eq(a)} and \goinl{Eq(b)}.
By covariance, since \goinl{Eq(a)} and \goinl{Eq(b)} implement
\goinl{Any}, the method declaration is allowed.

The Go team's current design does not support covariant receiver
typing.  Instead, receivers are nonvariant, just like method
signatures.  With that design, the bounds on \goinl{a} and \goinl{b}
in the declaration for pairs must exactly match those in the method
receiver.  Either the type and method declarations must both use bounds \goinl{Eq(a)} and \goinl{Eq(b)} (in which case one
cannot have pairs where the components do not support equality, even
if don't need that pair to itself support equality, reintroducing type
pollution and reducing flexibility), or they must both use bounds
\goinl{Any} (in which case the method body will need to add type
assertions to \goinl{Eq(a)} and \goinl{Eq(b)}, reducing efficiency and
reliability).}

\paragraph{Lists in FGG}
Figure~\ref{fig:fgg-list} shows lists in FGG.

Interface \goinl{List} now takes as a parameter the type of the
elements of the list, bounded by \goinl{Any}.
The signature for the map method is now
\goinl{Map(type b Any)(f Function(a, b)) List(b)}.
The list interface takes a parameter \goinl{a}
for the type of elements of the receiver list, while the map method itself
takes an additional parameter \goinl{b} for the type of elements of
the result list.

The two structures that implement lists now also take a type
parameter, again bounded by \goinl{Any}.  It may seem odd that
\goinl{Nil} requires a parameter, since it represents a list with no
elements.  However, without this parameter we could not declare that
\goinl{Nil} has method \goinl{Map}, whose signature mentions the
type of the list elements.

The main method simply adds type parameters.
In Go, no name can be bound to both a value and a type in a given
scope, so it is always unambiguous as to whether one is parsing a type
or an expression.  In practice, writing out all type parameters in
full can be tedious, and generic Go permits such
parameters to be omitted when they can be inferred.  Here we always
require type parameters, leaving inference for future work.

Type parameter names and variable names are irrelevant
when comparing method signatures, but method names,
bounds on type parameters, and type names must match exactly.
For example, the following two signatures are considered equivalent:
\begin{center}
\goinl{Map(type b Any)(f Function(a, b)) List(b)}
\quad
\goinl{Map(type bob Any)(fred Function(a, bob)) List(bob)}
\end{center}

If we wanted to define equality on lists without type
assertions, we would need to bound the elements of the list
so that they support equality, changing every occurrence
of \goinl{(type a Any)}, in the code to \goinl{(type a Eq(a))},
and similarly for \goinl{b} in the signature of \goinl{Map}.
\ROUNDTWO{This is a form of type pollution.
An alternative design that avoids pollution,
based on our solution to The Expression Problem,
can be found in {\iflong Appendix~\ref{sec:expression-list}\else
\cite{techreport}\fi}.}

\paragraph{Mutual recursion in type bounds}
In a declaration that introduces a list of type parameters,
the bounds of each may refer to any of the others.
Figure~\ref{fig:mutual} shows two mutually-recursive interface
declarations that may be useful in representing graphs.
It is parameterised over types for edges and vertexes.
Each edge has a source and target vertex, while each vertex
has a list of edges.

\subsection{The Expression Problem}
\label{sec:expression-fgg}

Following \citet{Wadler-1998}, we present The Expression Problem pared
to a minimum.  Our solution appears in Figure~\ref{fig:expression-fgg}.
There are just two structures that construct expressions,
\goinl{Num} and \goinl{Plus}, which denote numbers and the sum of two
expressions, respectively; and two methods that operate on
expressions, \goinl{Eval} and \goinl{String}, which evaluate an
expression and convert it to a string, respectively.
We show that each constructor and operation can be added
independently, proceeding in four steps:
\begin{center}
\begin{tabular}{l@{\qquad\qquad\qquad\qquad}l}
(1) define \goinl{Eval} on \goinl{Num}  & (3) define \goinl{String} on \goinl{Num} \\
(2) define \goinl{Eval} on \goinl{Plus} & (4) define \goinl{String} on \goinl{Plus}
\end{tabular}
\end{center}
The order of steps 2 and 3 can be reversed:
we may extend either by adding a new constructor or a new operation.
We assume availability of the library function
\goinl{fmt.Sprintf} to format strings.

Structure \goinl{Num} contains an integer value.
Structure \goinl{Plus} contains two fields, left and right,
which are themselves expressions.  Typically, we might expect
the type of these fields to be an interface specifying all operations we
wish to perform on expressions. But the whole point of the
expression problem is that we may add other operations later!
Thus, plus takes a type parameter for the type of these fields.
We bound the type parameter with \goinl{Any}, permitting it to be
instantiated by any type.

For each operation, \goinl{Eval} and \goinl{String}, we define a
corresponding interface to specify that operation, \goinl{Evaler} and
\goinl{Stringer}.
(The naming convention is typical of Go.)
When defining \goinl{Eval} on \goinl{Plus}, the receiver's type parameter
is bounded by interface \goinl{Evaler}, allowing \goinl{Eval} to be
recursively invoked on the left and right expressions.  Similarly,
when defining \goinl{String} on \goinl{Plus}, the receiver's type
parameter is bounded by interface \goinl{Stringer}, allowing
\goinl{String} to be recursively invoked.
\ROUNDTWO{Note this depends crucially on FGG's support for covariant receivers
in method declarations. Since the bounds or the receivers in the
method declarations, \goinl{Evaler} and \goinl{Stringer}, implement
the bounds in the type declarations, \goinl{Any}, 
the method declarations are allowed.}

\begin{Figure}
\begin{ruledgo}
\begin{golsttwo}
// Eval on Num
type Evaler interface {
  Eval() int
}
type Num struct {
  value int
}
func (e Num) Eval() int {
  return e.value
}
// Eval on Plus
type Plus(type a Any) struct {
  left a
  right a
}
func (e Plus(type a Evaler)) Eval() int {
  return e.left.Eval() + e.right.Eval()
}
(* \newpage *)
// String on Num
type Stringer interface {
  String() string
}
func (e Num) String() string {
  return fmt.Sprintf("%d", e.value)
}
// String on Plus
func (e Plus(type a Stringer)) String() string {
  return fmt.Sprintf("(%s+%s)",
    e.left.String(), e.right.String())
}
// tie it all together
type Expr interface {
  Evaler
  Stringer
}
func main() {
  var e Expr = Plus(Expr){Num{1}, Num{2}}
  var _ Int = e.Eval() // 3
  var _ string = e.String() // "(1+2)"
}
\end{golsttwo}
\end{ruledgo}
\caption{Expression problem in FGG}
\label{fig:expression-fgg}
\end{Figure}

A last step shows how to tie it all together. We define
an interface \goinl{Expr} embedding \goinl{Evaler}
and \goinl{Stringer}, and show how to build an \goinl{Expr} value 
which we can both evaluate and convert to a string.

How close could we get without generics? If we know all operations
in advance, then in place of the type parameter in \goinl{Plus} we can
use interface \goinl{Expr}, defining all required operations;
but that violates the requirement that we can add operations later.
Alternatively, in place of the type parameter in \goinl{Plus} we can
use interface \goinl{Any}, with type assertions to \goinl{Evaler} or
\goinl{Stringer} before the recursive calls; that allows us to add
operations later, but violates the requirement that all types be
checked statically.

\begin{Figure}
\begin{ruledgo}
\begin{golsttwo}
type Top struct {}
type Function<int,int> interface {
  Apply<0> Top
  Apply(x int) int
}
type incr struct { n int }
func (this incr) Apply<0> Top {
  return Top{}
}
func (this incr) Apply(x int) int {
  return x + this.n
}
type Function<int,bool> interface {
  Apply<1> Top
  Apply(x int) bool
}
type pos struct {}
func (this pos) Apply<1> Top {
  return Top{}
}
func (this pos) Apply(x int) bool {
  return x > 0
}
type List<int> interface {
  Map<2>() Top
  Map<int>(f Function<int,int>) List<int>
  Map<bool>(f Function<int,bool>) List<bool>
}
type Nil<int> struct {}
type Cons<int> struct {
  head int
  tail List<int>
}
func (xs Nil<int>) Map<2>() Top {
  return Top{}
}
func (xs Cons<int>) Map<2>() Top {
  return Top{}
}
func (xs Nil<int>)
    Map<int>(f Function<int,int>) List<int> {
  return Nil<int>{}
}
func (xs Cons<int>)
   Map<int>(f Function<int,int>) List<int> {
 return Cons<int>
   {f.Apply(xs.head), xs.tail.Map<int>(f)}
}
func (xs Nil<int>)
   Map<bool>(f Function<int,bool>) List<bool> {
 return Nil<bool>{}
}
func (xs Cons<int>)
   Map<bool>(f Function<int,bool>) List<bool> {
 return Cons<bool>
   {f.Apply(xs.head), xs.tail.Map<bool>(f)}
}
type List<bool> interface {
  Map<3>() Top
}
type Nil<bool> struct {}
type Cons<bool> struct {
 head bool
 tail List<bool>
}
func (xs Nil<bool>) Map<3>() Top {
 return Top{}
}
func (xs Cons<bool>) Map<3>() Top {
 return Top{}
}
func main() {
 var xs List<int> =
   Cons<int>{3, Cons<int>{6, Nil<int>}}
 var ys List<int> = xs.Map<int>(incr{-5})
 var _ List<bool> = ys.Map<bool>(pos{})
}
\end{golsttwo}
\end{ruledgo}
\caption{Monomorphisation: example of FG translation}
\label{fig:mono-ex-fg}

\vspace{2ex}
\begin{ruledgo}
\begin{golsttwo}
type Box(type a Any) struct {
  value a
}
(* \newpage *)
func (this Box(type a Any)) Nest(n int) Any {
  if (n == 0) { return this }
  else { return Box(Box(a)){this}.Nest(n-1) }
}
\end{golsttwo}
\end{ruledgo}
\caption{FGG code that cannot be monomorphised}
\label{fig:nomono}
\end{Figure}

\subsection{Monomorphisation by example}

We translate FGG into FG via \emph{monomorphisation}.
As an example, consider the FGG code in Figures~\ref{fig:fgg-func}
and~\ref{fig:fgg-list}.  We only include the code relevant to the main
method, so omit composition and equality.  The given code
monomorphises to the FG program shown in Figure~\ref{fig:mono-ex-fg}.

Each parametric type and method in FGG is translated to a family of
types and methods in FG, one for each possible instantiation of the
type parameters. FGG type \goinl{List(a)} is instantiated at types
\goinl{int} and \goinl{bool}, so it translates to the two FG types
\goinl{List<int>} and \goinl{List<bool>}.  For convenience, we assume
that angle brackets and commas ``\goinl{<,>}'' may appear in FG
identifiers, although that is not allowed in Go.  In our prototype, we
use Unicode letters that resemble angle brackets and a dash: Canadian
Syllabics Pa (U+1438), Po (U+1433), and Final Short Horizontal Stroke
(U+1428).

Monomorphisation tracks for each method
the possible types of its receiver and type parameters. In this
particular program, we need two instances of \goinl{Map} over lists of
integers, one that yields a list of integers and
one that yields a list of booleans,
and none for \goinl{Map} over lists of booleans.

Each interface also contains an instance of a
\emph{dummy} version of \goinl{Apply} or \goinl{Map},
here called, e.g., \goinl{Map<2>}, where the number in
brackets stands for a hash computed from the method signature.
A dummy method is provided for every source FGG method; these
dummy methods are needed to ensure a correct implementation relation
between structures and interfaces is maintained at runtime.  For
instance, if \goinl{f} is bound to
\goinl{incr\{1\}} then the type assertion
\goinl{f.(List<bool>)} should fail; but without the dummy,
interface \goinl{List<bool>} would have no methods
and hence any structure or interface would implement it.

Monomorphisation yields specialised type declarations for structures
and interfaces, and specialised method declarations, plus the required
dummy methods.
The source FGG and its translation to FG are both well-typed,
and both evaluate to corresponding terms: we will show the
translation preserves typing and is a bisimulation.

Not all typable FGG source can be monomorphised.
Figure~\ref{fig:nomono} shows a program that exhibits
\emph{polymorphic recursion}, where a method called at one type
recursively calls itself at a different type.
Here, calling method \goinl{Nest}
on a receiver of type \goinl{Box(a)}
leads to a recursive call
on a receiver of type \goinl{Box(Box(a))}.  Monomorphisation is
impossible because we cannot determine in advance to what depth the
types will nest.
We will present a theorem stating that if source code does not exhibit
problematic polymorphic recursion then it can be monomorphised.

 \section{Featherweight Go}
\label{sec:fg}

\subsection{FG syntax}
\label{sub:fg-syntax}

\begin{Figure}
\begin{ruled}
\begin{minipage}[t]{\textwidth}
\begin{tabular}[t]{ll}
Field name                      & $f$ \\
Method name                     & $m$ \\
Variable name                   & $x$ \\
Structure type name             & $t_S, u_S$ \\
Interface type name             & $t_I, u_I$ \\
Type name                       & $t, u$ ::= $t_S \mid t_I$ \\
Method signature                & $M$ ::= $(\ov{x~t})~t$ \\
Method specification            & $S$ ::= $mM$ \\
Type Literal                    & $T$ ::= \\
\quad Structure                 & \quad $\struct~\br{\ov{f~t}}$ \\
\quad Interface                 & \quad $\interface~\br{\ov{S}}$ \\
Declaration                     & $D$ ::= \\
\quad Type declaration          & \quad $\type~t~T$ \\
\quad Method declaration        & \quad $\func~(x~t_S)~mM~\br{\return~e}$ \\
Program                         & $P$ ::= $\package~\main;~\ov{D}~\func~\main()~\br{\un=e}$
\end{tabular}
\end{minipage}
\hspace{-0.5\textwidth}
\begin{minipage}[t]{0.4\textwidth}
\begin{tabular}[t]{ll}
Expression                      & $d, e$ ::= \\
\quad Variable                  & \quad $x$ \\
\quad Method call               & \quad $e.m(\ov{e})$ \\
\quad Structure literal         & \quad $t_S\br{\ov{e}}$ \\
\quad Select                    & \quad $e.f$ \\
\quad Type assertion            & \quad $e.(t)$
\end{tabular}
\end{minipage}
\end{ruled}
\caption{FG syntax}
\label{fig:fg-syntax}
\end{Figure}

Figure~\ref{fig:fg-syntax} presents FG syntax.
We let $f$ range over field names, $m$ range over
method names, $x$ range over variable names,
$t_S, u_S$ range over structure names, and
$t_I, u_I$ range over interface names.

We let $t, u$ range over type names, which are
either structure or interface type names.
We let $d$ and $e$ range over expressions,
which have five forms:
variable $x$,
method call $e.m(\ov{e})$, 
structure literal $t_S\br{\ov{e}}$,
selection $e.f$,
and type assertion $e.(t)$.
By convention, $\ov{e}$ stands for the sequence $e_1,\ldots,e_n$.
We consider $e$ and $\ov{e}$ to be distinct metavariables.

A method signature $M$ has the form $(\ov{x~t})~t$.
Here $\ov{x}$ stands for $x_1,\ldots,x_n$
and $\ov{t}$ stands for $t_1,\ldots,t_n$,
and hence $\ov{x~t}$ stands for $x_1~t_1,\ldots,x_n~t_n$.
We use similar conventions extensively.

A method specification $S$ is a method name followed
by a method signature.
A type literal $T$ is either
a structure $\struct~\br{\ov{f~t}}$
or an interface $\interface~\br{\ov{S}}$.
A declaration $D$ is either a
type declaration $\type~t~T$
or a method declaration $\func~(x~t_S)~mM~\br{\return~e}$.
In our examples, interface declarations may 
contain interface embeddings, i.e., a reference to
another interface;
for our formalism, we assume these are always expanded out to the
corresponding method specifications.

A program $P$ consists of a
sequence of declarations $\ov{D}$
and a top-level expression $e$,
written in the stylised form shown in the figure
to make it legal Go.
We sometimes abbreviate it as $\ov{D} \prog e$.

\subsection{Auxiliary functions}
\label{sub:fg-aux}

Figure~\ref{fig:fg-aux} presents several auxiliary
definitions.  All definitions assume a given program
with a fixed sequence of declarations $\ov{D}$.

Function $\fields(t_S)$ looks up the structure
declaration for $t_S$ and returns a sequence $(\ov{f~t})$
of field names and their types.
Write $t_S.m$ to refer to the method declaration
with receiver type $t_S$ and name $m$.
Function $\mbody(t_S.m)$ returns $(x:t_S,\ov{x:t}).e$,
where $x:t_S$ is the receiver parameter and its type,
$\ov{x:t}$ the argument parameters and their types,
and $e$ the body from the declaration 
of a method with receiver of type $t_S$ and name $m$.
In the phrase $\func~(x~t_S)~m(\ov{x~t})~t$, each of
$t_S$, $\ov{t}$, and $t$ is considered a distinct metavariable,
and similarly for $x$ and $\ov{x}$.

Function $\vtype(v)$ is explained in Section~\ref{sub:fg-reduction}.
Predicate $\unique(\ov{S})$ holds if for every
method specification $mM$ in $\ov{S}$ the method name $m$
uniquely determines the method signature $M$.

Function $\tdecls(\ov{D})$ returns a sequence with
the name of every type declared in $\ov{D}$.
Function $\mdecls(\ov{D})$ returns a sequence with
a pair $t_S.m$ for every method declared in $\ov{D}$.
\ROUNDTWO{Predicate $\distinct$, not defined in the figure, takes a sequence,
and holds if no item in the sequence is duplicated.  We are careful to
distinguish between sets and sequences.  Predicate $\distinct$ must
take a sequence rather than a set, because items in a set are distinct
by definition.  Sequences may implicitly coerce to sets, but not
vice-versa.  When comparing method signatures, the names of formal
parameters are ignored; signatures are considered equal if they contain
the same types in the same sequence.}

\begin{Figure}
\begin{ruled}
\begin{mathpar}
\inferrule
  {(\type~t_S~\struct\br{\ov{f~t}}) \in \ov{D}}
  {\fields(t_S) = \ov{f~t}}

\inferrule
  {(\func~(x~t_S)~m(\ov{x~t})~t~\br{\return~e}) \in \ov{D}}
  {\mbody(t_S.m) = (x:t_S,\ov{x:t}).e}
\\
\inferrule
  {}
  {\vtype(t_S\br{\ov{v}}) = t_S}

\inferrule
  {\text{$mM_1, mM_2 \in \ov{S}$ implies $M_1 = M_2$}}
  {\unique(\ov{S})}
\\
\inferrule
  {}
  {\tdecls(\ov{D}) = \lst{t \mid (\type~t~T) \in \ov{D}}}

\inferrule
  {}
  {\mdecls(\ov{D}) = \lst{t_S.m \mid
      (\func~(x~t_S)~mM~\br{\return~e}) \in \ov{D}}}
\\
\inferrule
  {}
  {\methods(t_S) = \set{ mM \mid (\func~(x~t_S)~mM~\br{\return~e}) \in \ov{D}}}

\inferrule
  {\type~t_I~\interface\br{\ov{S}} \in \ov{D}}
  {\methods(t_I) = \ov{S}}
\end{mathpar}
\end{ruled}
\caption{FG auxiliary functions}
\label{fig:fg-aux}
\end{Figure}

\begin{Figure}
\begin{ruled}
Implements, well-formed type
  \hfill \fbox{$t \imp u$} \qquad \fbox{$t \ok$}
\begin{mathpar}
\inferrule[<:$_S$]
  { ~ }
  { t_S \imp t_S }

\inferrule[<:$_I$]
  {
    \methods(t) \supseteq \methods(t_I)
  }
  { t \imp t_I }

\inferrule[t-named]
  {
    (\type~t~T) \in \ov{D}
  }
  { t \ok }
\end{mathpar}

Well-formed method specifications and type literals
  \hfill \fbox {$S \ok$} \qquad \fbox{$T \ok$}
\begin{mathpar}
\inferrule[t-specification]
  {
    \distinct(\ov{x}) \\
    \ov{t \ok} \\
    t \ok
  }
  { m(\ov{x~t})~t \ok }

\inferrule[t-struct]
  {
    \distinct(\ov{f}) \\
    \ov{t \ok}
  }
  { \struct~\br{\ov{f~t}} \ok }

\inferrule[t-interface]
  {
     \unique(\ov{S}) \\
     \ov{S \ok}
  }
  { \interface~\br{\ov{S}} \ok }
\end{mathpar}

Well-formed declarations \hfill \fbox{$D \ok$}
\vspace{-2ex}
\begin{mathpar}
\inferrule[t-type]
  {
    T \ok
  }
  { \type~t~T \ok }

\inferrule[t-func]
  {
    \distinct(x, \ov{x}) \\\\
    t_S \ok \\
    \ov{t \ok} \\
    u \ok \\
    x : t_S \comma \ov{x : t} \vdash e : t \\
    t \imp u
  }
  { \func~(x~t_S)~m(\ov{x~t})~u~\br{\return~e} \ok }
\end{mathpar}

Expressions \hfill \fbox{$\Gamma \vdash e : t$}
\begin{mathpar}
\inferrule[t-var]
  { (x : t) \in \Gamma }
  { \Gamma \vdash x : t }

\inferrule[t-call]
  {
    \Gamma \vdash e : t \\
    \Gamma \vdash \ov{e : t} \\
    (m(\ov{x~u})~u) \in \methods(t) \\
    \ov{t \imp u}
  }
  { \Gamma \vdash e.m(\ov{e}) : u }
\\
\inferrule[t-literal]
  {
    t_S \ok \\
    \Gamma \vdash \ov{e : t} \\
    (\ov{f~u}) = \fields(t_S) \\
    \ov{t \imp u}
  }
  { \Gamma \vdash t_S\br{\ov{e}} : t_S }

\inferrule[t-field]
  {
    \Gamma \vdash e : t_S \\
    (\ov{f~u}) = \fields(t_S)
  }
  { \Gamma \vdash e.f_i : u_i }
\\
\inferrule[t-assert$_I$]
  {
    t_I \ok \\
    \Gamma \vdash e : u_I
  }
  { \Gamma \vdash e.(t_I) : t_I }

\inferrule[t-assert$_S$]
  {
    t_S \ok \\
    \Gamma \vdash e : u_I \\
    t_S \imp u_I
  }
  { \Gamma \vdash e.(t_S) : t_S }

\fbox{
\inferrule[\textsc{t-stupid}]
  { 
    t \ok \\
    \Gamma \vdash e : u_S
  }
  { \Gamma \vdash e.(t) : t }
}
\end{mathpar}

Programs \hfill \fbox{$P \ok$}
\begin{mathpar}
  \inferrule[t-prog]
  {   
     \distinct(\tdecls(\ov{D})) \\
     \distinct(\mdecls(\ov{D})) \\
     \ov{D \ok} \\
     \emptyset \vdash e : t
  }
  { \package~\main;~\ov{D}~\func~\main()~\br{\un=e} \ok }
\end{mathpar}
\end{ruled}
\caption{FG typing}
\label{fig:fg-typing}
\vspace{3ex}
\end{Figure}

Function $\methods(t)$ returns the set of all
method specifications belonging to type $t$.
If $t$ is a structure type the method specifications
are those from the method declarations with a
receiver of the given type.  It $t$ is an interface
type, the method specifications are those given in
the interface.

Figure~\ref{fig:fg-typing} presents the FG typing rules.
Let $\Gamma$ range over environments, which are sequences
of variables paired with type names, $\ov{x : t}$.
We write $\emptyset$ for the empty environment.

Judgement $t \imp u$ holds if type $t$ \emph{implements}
type $u$.  A structure type $t_S$ is only implemented by itself,
while type $t$ implements interface type $t_I$ if the
methods defined on $t$ are a superset of those defined on $t_I$.
It follows from the definition that $\imp$
is reflexive and transitive.

\subsection{FG Typing}
\label{sub:fg-typing}

\begin{Figure}
\begin{ruled}
\begin{center}
Value \qquad $v$ ::= $t_S\br{\ov{v}}$
\end{center}

\begin{center}
\begin{minipage}{0.5\textwidth}
\begin{tabular}{ll}
  Evaluation context	        & $E$ ::= \\
  \quad Hole                    & \quad $\Hole$ \\
  \quad Method call receiver    & \quad $E.m(\ov{e})$ \\
  \quad Method call arguments	& \quad $v.m(\ov{v},E,\ov{e})$
\end{tabular}
\end{minipage}
\begin{minipage}{0.4\textwidth}
\begin{tabular}{ll}
  ~\\
  \quad Structure	        & \quad $t_S\br{\ov{v},E,\ov{e}}$ \\
  \quad Select		        & \quad $E.f$ \\
  \quad Type assertion	        & \quad $E.(t)$
\end{tabular}
\end{minipage}
\end{center}

Reduction \hfill \fbox{$d \becomes e$}
\begin{mathpar}
\inferrule[r-field]
  { (\ov{f~t}) = \fields(t_S) }
  { t_S\br{\ov{v}}.f_i \becomes v_i }
\qquad
\inferrule[r-call]
  { (x : t_S,\ov{x: t}).e = \mbody(\vtype(v).m) }
  { v.m(\ov{v}) \becomes e[x \by v, \ov{x \by v}] }
\qquad
\inferrule[r-assert]
  { \vtype(v) \imp t }
  { v.(t) \becomes v }
\qquad
\inferrule[r-context]
  { d \becomes e }
  { E[d] \becomes E[e] }
\end{mathpar}
\end{ruled}
\caption{FG reduction}
\label{fig:fg-reduction}
\end{Figure}

We write $\meta{ok}$ to indicate a construct is well-formed.
Judgement $t \ok$ holds if type $t$ is declared.
Judgement $S \ok$ holds if method specification $S$ is well formed:
all formal parameters $\ov{x}$ in it are distinct,
and all the types $\ov{t}, t$ in it are declared.
Judgement $T \ok$ holds if type literal $T$ is well formed:
for a structure, all field names must be distinct and
all types declared; for an interface, all its method
specifications must be well formed.
Judgement $D \ok$ holds if declaration $D$ is well formed:
for a type declaration,
its type literal must be well formed;
for a method declaration,
its receiver and formal parameters must be distinct,
all types must be declared,
the method body must be well typed in the appropriate environment,
and the expression type must implement the declared return type.

Judgement $\Gamma \vdash e:t$ holds if in environment $\Gamma$
expression $e$ has type $t$.  Rules for variable,
method call, structure literal, and field selection are straightforward.
For instance, the four hypotheses for a method
call check the type of the receiver, check the types of the arguments,
look up the signature of the method, and confirm the types of the
arguments implement the types of the parameters.

\ROUNDTWO{Type assertions are more interesting.  A type assertion $e.(t)$ always
returns a value of type $t$ (if it doesn't panic).  Let $e$ have type
$u$.  There are three cases.  If $u$ and $t$ are both interface types
(\tyrulename{t-assert$_I$}) then the assertion is always allowed, since
$e$ could always conceivably evaluate to a structure that implements
$t$.  If $u$ is an interface but $t$ is a structure
(\tyrulename{t-assert$_S$}) then the assertion is allowed only if $t$
implements $u$, since otherwise $e$ could not possibly contain a
structure of type $t$.  If $u$ is a structure type
(\tyrulename{t-stupid}) then it is stupid to write the assertion in
source code, since the assertion could be checked at compile time,
making it pointless.  Nonetheless, during reduction a variable of
interface type will be replaced by a value of structure type, so
without such stupid type assertions an expression would become
ill-typed during reduction.  We write a box around this rule to
indicate that it doesn't apply to source terms, but may apply to terms
that result from reducing the source.  Stupid type assertions are
similar to stupid casts as found in Featherweight Java.}

Judgement $P \ok$ holds if program $P$ is well formed:
all its type declarations are distinct,
all its method declarations are distinct
(each pair of a receiver type with a method name is distinct),
all its declarations are well formed,
and its body is well typed in the empty environment.

\subsection{FG Reduction}
\label{sub:fg-reduction}

Figure~\ref{fig:fg-reduction} presents the FG reduction rules.
A value $v$ is a structure literal $t_S\br{\ov{v}}$
where each field is itself filled with a value.
The auxiliary function $\vtype(v)$ returns $t_S$
when $v = t_S\br{\ov{v}}$.  Evaluation contexts $E$
are defined in the usual way.
Judgement $d \becomes e$ holds if expression $d$
steps to expression $e$.  There are four rules, for
field selection, method call, type assertion, and
closure under evaluation contexts.  All are straightforward.

\subsection{FG properties}
\label{sub:fg-prop}

We have the usual results relating typing and reduction.

\begin{lemma}[Well formed]
  If\/ $\Gamma \vdash e : t$ then\/ $t \ok$.
\end{lemma}

The substitution lemma is straightforward. It is sufficient
to consider empty environments for the substituted terms, since
FG has no binding constructs (such as lambda) in expressions.

\begin{lemma}[Substitution]
  If\/ $\emptyset \vdash \ov{v:t}$
  and $\ov{x:u} \vdash e:u$
  and $\ov{t \imp u}$
  then $\emptyset \vdash e[\ov{x \by v}]:t$
  for some type $t$ with $t \imp u$.
\end{lemma}

The following are straightforward adaptions of the usual
results. We say expression $e$ \emph{panics} if
there exist evaluation context $E$, value $v$, and type $t$
such that $e = E[v.(t)]$ and $\vtype(v) \notimp t$.

\begin{theorem}[Preservation]
  If\/ $\emptyset \vdash d : u$ and $d \becomes e$
  then\/ $\emptyset \vdash e : t$ for some $t$
  with\/ $t \imp u$.
\end{theorem}

\begin{theorem}[Progress]
  If\/ $\emptyset \vdash d:u$ then
  either\/ $d$ is a value,
  $d \becomes e$ for some $e$,
  or\/ $d$ panics.
\end{theorem}

 \section{Featherweight Generic Go}
\label{sec:fgg}

\subsection{FGG Syntax}
\label{sub:fgg-syntax}

Figure~\ref{fig:fgg-syntax} presents FGG syntax, with the
differences from FG syntax highlighted.
We let $\alpha$ range over type parameters
and let $\tau, \sigma$
range over types. A type is either a type parameter
$\alpha$ or a named type $t(\ov{\tau})$.
We also let $\tau_S, \sigma_S$
range over structure types of the form $t_S(\ov{\tau})$;
$\tau_I, \sigma_I$
range over interface types of the form $t_I(\ov{\tau})$;
and, $\tau_J, \sigma_J$ range over
types that are either type parameters or interfaces $\tau_I$.

Expressions and declarations are updated to replace type names by
types, and method calls are updated to include type parameters:
a structure declaration is now $\struct~\br{\ov{f~\tau}}$
and a method call is now $e.m(\ov{\tau})(\ov{e})$,
a structure literal is now $\tau_S\br{\ov{e}}$,
and a type assertion is now $e.(\tau)$.

We let $\Phi, \Psi$ range over type formals,
which have the form $\type~\ov{\alpha~\tau_I}$,
pairing type parameters with their bounds, which
are interface types.  The bounds in type formals
are mutually recursive, i.e., each interface in $\ov{\tau_I}$
may refer to any parameter in $\ov{\alpha}$.
Type declarations $\type~t(\Phi)~T$, and
signatures $(\Psi)(\ov{x~\tau})~\tau$,
and method declarations $\func~(x~t_S(\Phi))~mM~\br{\return~e}$
now include type formals.

We let $\phi, \psi$ range over type actuals,
which are sequences of types.

\begin{Figure}
\begin{ruled}
\gray{
\begin{minipage}[t]{\textwidth}
\begin{tabular}[t]{ll}
Field name                      & $f$ \\
Method name                     & $m$ \\
Variable name                   & $x$ \\
Structure type name             & $t_S, u_S$ \\
Interface type name             & $t_I, u_I$ \\
Type name                       & $t, u$ ::= $t_S \mid t_I$ \\
\black{Type parameter}          & \black{$\alpha$} \\
Method signature                & $M$ ::= $\black{(\Psi)}(\ov{x~\black{\tau}})~\black{\tau}$ \\
Method specification            & $S$ ::= $mM$ \\
Type Literal                    & $T$ ::= \\
\quad Structure                 & \quad $\struct~\br{\ov{f~\black{\tau}}}$ \\
\quad Interface                 & \quad $\interface~\br{\ov{S}}$ \\
Declaration                     & $D$ ::= \\
\quad Type declaration          & \quad $\type~t\black{(\Phi)}~T$ \\
\quad Method declaration        & \quad $\func~(x~t_S\black{(\Phi)})~mM~\br{\return~e}$ \\
Program                         & $P$ ::= $\package~\main;~\ov{D}~\func~\main()~\br{\un=e}$
\end{tabular}
\end{minipage}
\hspace{-0.5\textwidth}
\begin{minipage}[t]{0.4\textwidth}
\begin{tabular}[t]{ll}
\black{Type}                    & \black{$\tau, \sigma$ ::=} \\
\quad \black{Type parameter}    & \quad \black{$\alpha$} \\
\quad \black{Named type}        & \quad \black{$t(\ov{\tau})$} \\
\black{Structure type}          & \black{$\tau_S,\sigma_S$ ::= $t_S(\ov{\tau})$} \\
\black{Interface type}          & \black{$\tau_I,\sigma_I$ ::= $t_I(\ov{\tau})$} \\
\black{Interface-like type}     & \black{$\tau_J,\sigma_J$ ::= $\alpha \mid \tau_I$} \\
\black{Type formal}             & \black{$\Phi$, $\Psi$ ::= $\type~\ov{\alpha~\tau_I}$} \\
\black{Type actual}             & \black{$\phi$, $\psi$ ::= $\ov{\tau}$} \\
Expression                      & $e$ ::= \\
\quad Variable                  & \quad $x$ \\
\quad Method call               & \quad $e.m\black{(\ov{\tau})}(\ov{e})$ \\
\quad Structure literal         & \quad $\black{\tau_S}\br{\ov{e}}$ \\
\quad Select                    & \quad $e.f$ \\
\quad Type assertion            & \quad $e.(\black{\tau})$
\end{tabular}
\end{minipage}
}
\end{ruled}
\caption{FGG syntax}
\label{fig:fgg-syntax}
\vspace{3ex}
\end{Figure}

\subsection{Auxiliary functions}
\label{sub:fgg-aux}

Figure~\ref{fig:fgg-aux} presents several auxiliary definitions.
As before, $\Gamma$ ranges over environments, which are now
sequences that pair variables with types, $\ov{x : \tau}$.
In addition, $\Delta$ ranges over type environments, which
are sequences that pair type parameters with bounds,
$\ov{\alpha : \tau_I}$. Type formals $\Phi, \Psi$ may
implicitly coerce to type environments.

We write $\eta = (\Phi \by \phi)$ for the
substitution of formals $\Phi$ by actuals $\phi$,
and $\eta = (\Phi \by_\Delta \phi)$ for the
partial function that also checks that $\phi$
respects the bounds imposed by $\Phi$.
If a partial function that is undefined appears
in the hypothesis of a rule, then the corresponding premise does not hold.
We write $\hat\Phi$ for the type parameters of $\Phi$.

Functions $\fields(\tau_S)$
and $\mbody(\tau_S.m(\psi))$
are updated to replace type names by types,
and for the latter to include method type arguments.
The definitions are adjusted to include type
formals which are instantiated appropriately.
Functions $\vtype$, $\tdecls$, and $\mdecls$
and predicate $\unique$
are updated to replace type names by types and to include
type formals.
Function $\bounds_\Delta(\tau)$
takes a type parameter to its bound,
and leaves its argument unchanged otherwise.

Function $\methods_\Delta(\tau)$ is updated
to accept a type environment and
to replace type names by types.
The definition is adjusted to include type
formals which are instantiated appropriately.
If $\methods$ is applied to a type parameter,
that parameter behaves the same as its bounding interface,
so $\methods_\Delta(\tau) = \methods_\Delta(\bounds_\Delta(\tau))$ for
all $\tau$.

\begin{Figure}
\begin{ruled}
\begin{mathpar}
\inferrule
  {(\type~\ov{\alpha~\tau_I}) = \Phi \\  
   \eta = (\ov{\alpha \by \tau})}
  {(\Phi \by \ov{\tau}) = \eta}   

\inferrule
  {(\type~\ov{\alpha~\tau_I}) = \Phi \\
   \eta = (\Phi \by \phi) \\
   \Delta \vdash (\ov{\alpha \imp \tau_I})[\eta]}
  {(\Phi \by_\Delta \phi) = \eta}  
\\
\inferrule
  {
    (\type~t_S(\Phi)~\struct~\br{\ov{f~\tau}}) \in \ov{D} \\
    \eta = (\Phi \by \phi)
  }
  {\fields(t_S(\phi)) = (\ov{f~\tau})[\eta]}
\\
\inferrule
  {
    (\func~(x~t_S(\Phi))~m(\Psi)(\ov{x~\tau})~\tau~\br{\return~e}) \in \ov{D} \\
    \theta = (\Phi, \Psi \by \phi, \psi)
  }
  {\mbody(t_S(\phi).m(\psi)) = (x:t_S(\phi),\ov{x:\tau}).e[\theta]}
\\
\inferrule
  {(\type~\ov{\alpha~\tau_I}) = \Phi}
  {\hat\Phi = \ov{\alpha}}

\inferrule
  {}
  {\vtype(\tau_S\br{\ov{v}}) = \tau_S}

\inferrule
  {\text{$mM_1, mM_2 \in \ov{S}$ implies $M_1 = M_2$}}
  {\unique(\ov{S})}
\\
\inferrule
  {}
  {\tdecls(\ov{D}) = \lst{t \mid (\type~t(\Phi)~T) \in \ov{D}}}

\inferrule
  {}
  {\mdecls(\ov{D}) = \lst{t_S.m \mid
      (\func~(x~t_S(\Phi))~mM~\br{\return~e}) \in \ov{D}}}
\\
\inferrule
 {(\alpha : \tau_I) \in \Delta}
 {\bounds_\Delta(\alpha) = \tau_I}

\inferrule
 {}
 {\bounds_\Delta(\tau_S) = \tau_S}

\inferrule
 {}
 {\bounds_\Delta(\tau_I) = \tau_I}
\\
\inferrule  
  {}
  {\methods_\Delta(t_S(\phi)) =
     \set{(mM)[\eta]
       \mid (\func~(x~t_S(\Phi))~mM~\br{\return~e}) \in \ov{D} \comma
         \eta = (\Phi \by_\Delta \phi)}}
\\
\inferrule
  {\type~t_I(\Phi)~\interface~\br{\ov{S}} \in \ov{D} \\
   \eta = (\Phi \by \phi)}
  {\methods_\Delta(t_I(\phi)) = \ov{S}[\eta]}

\inferrule
  {(\alpha : \tau_I) \in \Delta}
  {\methods_\Delta(\alpha) = \methods_\Delta(\tau_I)}
\end{mathpar}
\end{ruled}
\caption{FGG auxiliary functions}
\label{fig:fgg-aux}
\end{Figure}

\subsection{FGG Typing}
\label{sub:fgg-typing}

\begin{Figure}
\begin{ruled}
Implements
  \hfill \fbox{$\Delta \vdash \tau \imp \sigma$}
  \qquad \fbox{$\Phi \imp \Psi$}
\begin{mathpar}
\inferrule[<:-param]
  { ~ }
  { \Delta \vdash \alpha \imp \alpha }

\inferrule[<:$_S$]
  { ~ }
  { \Delta \vdash \tau_S \imp \tau_S }

\inferrule[<:$_I$]
  { \methods_\Delta(\tau) \supseteq \methods_\Delta(\tau_I) }
  { \Delta \vdash \tau \imp \tau_I }

\inferrule[<:-formals]
  {  
    \emptyset \vdash \ov{\tau_I \imp \sigma_I}
  }
  {(\type~\ov{\alpha~\tau_I}) \imp (\type~\ov{\alpha~\sigma_I})}
\end{mathpar}

Well-formed type and actuals
  \hfill \fbox{$\Delta \vdash \tau \ok$}
  \qquad \fbox{$\Delta \vdash \phi \ok$}
\begin{mathpar}
\inferrule[t-param]
  { (\alpha : \tau_I) \in \Delta }
  { \Delta \vdash \alpha \ok }

\inferrule[t-named]
  {
    \Delta \vdash \phi \ok
    \and
    (\type~t(\Phi)~T) \in \ov{D}
    \and
    \eta = (\Phi \by_\Delta \phi)
  }
  { \Delta \vdash t(\phi) \ok }

\inferrule[t-actual]
  {
    \ov{\tau} = \phi \\
    \Delta \vdash \ov{\tau \ok}
  }
  { \Delta \vdash \phi \ok }
\end{mathpar}

Well-formed type formals and nested formals
  \hfill \fbox{$\Phi \vdash \Psi \ok$}
  \qquad \fbox{$\Phi \stoup \Psi \ok~\Delta$}
\begin{mathpar}
\inferrule[t-formal]
  {
    (\type~\ov{\alpha~\tau_I}) = \Psi \\
    \distinct(\hat\Phi, \ov{\alpha}) \\
    \Phi, \Psi \vdash \ov{\tau_I \ok}
  }
  {\Phi \vdash \Psi \ok}

\inferrule[t-nested]
  {
    \emptyset \vdash \Phi \ok \\
    \Phi \vdash \Psi \ok \\
    \Delta = \Phi, \Psi
  }
  {\Phi \stoup \Psi \ok~\Delta}

\end{mathpar}

Well-formed method specifications and type literals
  \hfill \fbox{$\Phi \vdash S \ok$} \qquad \fbox{$\Phi \vdash T \ok$}
\begin{mathpar}
\inferrule[t-specification]
  {
    \Phi \stoup \Psi \ok~\Delta \\\\
    \distinct(\ov{x}) \\
    \Delta \vdash \ov{\tau \ok} \\
    \Delta \vdash \tau \ok
  }
  { \Phi \vdash m(\Psi)(\ov{x~\tau})~\tau \ok }

\inferrule[t-struct]
  {
    \distinct(\ov{f}) \\
    \Phi \vdash \ov{\tau \ok}
  }
  { \Phi \vdash \struct~\br{\ov{f~\tau}} \ok }

\inferrule[t-interface]
  {
    \unique(\ov{S}) \\
    \Phi \vdash \ov{S \ok}
  }
  { \Phi \vdash \interface~\br{\ov{S}} }
\end{mathpar}

Well-formed declarations \hfill \fbox{$D \ok$}
\begin{mathpar}
\inferrule[t-type]
 {
   \emptyset \vdash \Phi \ok \\
   \Phi \vdash T \ok
 }
 { \type~t(\Phi)~T \ok }

\inferrule[t-func]
 {
   \distinct(x, \ov{x}) \\
   (\type~t_S(\Phi')~T) \in \ov{D} \\
   \Phi \imp \Phi' \\
   \Phi \stoup \Psi \ok~\Delta \\\\
   \Delta \vdash \ov{\tau \ok} \\
   \Delta \vdash \sigma \ok \\
   \Delta \stoup
     x : t_S(\hat{\Phi}) \comma \ov{x : \tau} \vdash e : \tau \\
   \Delta \vdash \tau \imp \sigma
 }
 { \func~(x~t_S(\Phi))~m(\Psi)(\ov{x~\tau})~\sigma~\br{\return~e} \ok }
\end{mathpar}

Expressions \hfill \fbox{$\Delta \stoup \Gamma \vdash e : \tau<$}
\begin{mathpar}
\inferrule[t-var]
 {
   (x : \tau) \in \Gamma
 }
 { \Delta \stoup \Gamma \vdash x : \tau }

\inferrule[t-call]
 {
   (m(\Psi)(\ov{x~\sigma})~\sigma) \in \methods_\Delta(\tau)  \\\\
   \Delta \stoup \Gamma \vdash e : \tau \\
   \Delta \stoup \Gamma \vdash \ov{e : \tau} \\
   \eta = (\Psi \by_\Delta \psi) \\
   \Delta \vdash (\ov{\tau \imp \sigma})[\eta]
 }
 { \Delta \stoup \Gamma \vdash e.m(\psi)(\ov{e}) : \sigma[\eta] }
\\
\inferrule[t-literal]
 {
   \Delta \vdash \tau_S \ok
   \quad
   \Delta \stoup \Gamma \vdash \ov{e : \tau}
   \quad
   (\ov{f~\sigma}) = \fields(\tau_S)
   \quad
   \Delta \vdash \ov{\tau \imp \sigma}
 }
 { \Delta \stoup \Gamma \vdash \tau_S\br{\ov{e}} : \tau_S }

\inferrule[t-field]
 {
   \Delta \stoup \Gamma \vdash e : \tau_S 
   \quad
   (\ov{f~\tau}) = \fields(\tau_S)
 }
 { \Delta \stoup \Gamma \vdash e.f_i : \tau_i }
\\
\inferrule[t-assert$_I$]
 {
   \Delta \vdash \tau_J \ok \gap
   \Delta \stoup \Gamma \vdash e : \sigma_J
 }
 { \Delta \stoup \Gamma \vdash e.(\tau_J) : \tau_J }

\squeeze
\inferrule[t-assert$_S$]
 {
   \Delta \vdash \tau_S \ok \gap
   \Delta \stoup \Gamma \vdash e : \sigma_J \gap
   \tau_S \imp \bounds_\Delta(\sigma_J)
 }
 { \Delta \stoup \Gamma \vdash e.(\tau_S) : \tau_S }
\squeeze
 
\fbox{
\inferrule[t-stupid]
 {
  \Delta \vdash \tau  \ok \gap
  \Delta \stoup \Gamma \vdash e : \sigma_S
 }
 { \Delta \stoup \Gamma \vdash e.(\tau) : \tau }
}
\end{mathpar}

Programs  \hfill \fbox{$P \ok$}
\begin{mathpar}
  \inferrule[t-prog]
  {
    \distinct(\tdecls(\ov{D})) \\
    \distinct(\mdecls(\ov{D})) \\
    \ov{D \ok} \\
    \emptyset \stoup \emptyset \vdash e : \tau
  }
  { \package~\main;~\ov{D}~\func~\main()~\br{\un=e} \ok }
\end{mathpar}
\end{ruled}
\caption{FGG typing}
\label{fig:fgg-typing}
\end{Figure}

Figure~\ref{fig:fgg-typing} presents the FGG typing rules.
Judgement $\Delta \vdash \tau \imp \sigma$ now depends
on a type environment and relates types rather than type names.
The definition is adjusted so that a type
parameter implements its bound.  It still follows from
the definition that $\imp$ is reflexive and transitive.
Judgement $\Phi \imp \Psi$ compares the
corresponding bounds of two type formals under
an empty type environment.

Judgement $\Delta \vdash \tau \ok$ holds if a type is
well formed: all type parameters in it must be
declared in $\Delta$ and all named types must be
instantiated with type \ROUNDTWO{arguments} that satisfy the bounds of
the corresponding type parameters.
Judgement $\Delta \vdash \phi \ok$ holds if under
environment $\Delta$ all types in
$\phi$ are well formed.

Judgement $\Phi \vdash \Psi \ok$ holds if 
under type environment $\Phi$ the type formal $\Psi$
is well formed: all type parameters
bound by $\Phi$ and $\Psi$ are distinct, and all
the bounds in $\Psi$ are well formed in the type environment
that results from combining $\Phi$ and $\Psi$.  Note
this permits mutually recursive bounds in a type formal.
Judgement $\Phi \stoup \Psi \ok~\Delta$ holds if
a method declaration with receiver formals $\Phi$
and method formals $\Psi$ is well formed,
yielding type environment $\Delta$:
it requires that
$\Phi$ is well formed under the empty environment,
$\Psi$ is well formed under $\Phi$,
and $\Delta$ is the concatenation of $\Phi$ and $\Psi$.
Hence, the type formals of the receiver are in scope when
declaring the type formals of the method, but not vice
versa, and both are in scope for declaring the types of
the arguments and result.

Judgement $\Phi \vdash S \ok$ holds if under type environment $\Phi$
method specification $S$ is well formed: it requires
$\Phi, \Psi \ok~\Delta$ where $\Psi$ is the type formals of the method
specification, and the rest is similar to before but now under type
environment $\Delta$.
Judgement $\Phi \vdash T \ok$ holds
if under type formals $\Phi$ type literal $T$ is well formed,
and again is a straightforward adjustment of its earlier definition.

Judgement $D \ok$ holds if declaration $D$ is well formed.
The definitions are similar to previous definition.
For a type declaration,
its type formals must be well formed under the empty type environment,
and its type literal must be well formed under the environment given by
the type formals.
For a method declaration,
we require $\Phi \stoup \Psi \ok~\Delta$, 
where $\Phi$ are the type formals of the receiver
and $\Psi$ are the type formals of the method.
The receiver type must be declared with formals $\Phi'$,
where $\Phi \imp \Phi'$.
An alternative, simpler design would require $\Phi$ and $\Phi'$
to be identical; but that would rule out the solution to
the expression problem given in Section~\ref{sec:expression-fgg}.
The rest is a straightforward adjustment of its earlier definition.

Judgement $\Delta \stoup \Gamma \vdash e : \tau$ holds
if under type environment $\Delta$ and environment $\Gamma$
expression $e$ has type $\tau$.  The adjustments are
straightforward.  Method calls are
adjusted so that the type of the arguments and result
are instantiated by the method type arguments.
Type assertions are adjusted to take into account
that type names are replaced by types.  In method calls
and type assertions, type parameters are treated
as equivalent to the parameters' bound.

\subsection{FGG Reduction}
\label{sub:fgg-reduction}

\begin{Figure}
\begin{ruled}
\gray{
\begin{center}
Value \qquad $v$ ::= $\black{\tau_S}\br{\ov{v}}$
\end{center}

\hspace{-1em}
\begin{minipage}{0.5\textwidth}
\begin{tabular}{ll}
 Evaluation context             & $E$ ::= \\
 \quad Hole                     & \quad $\Hole$ \\
 \quad Method call receiver     & \quad $E.m\black{(\ov{\tau})}(\ov{e})$ \\
 \quad Method call arguments    & \quad $v.m\black{(\ov{\tau})}(\ov{v},E,\ov{e})$
\end{tabular}
\end{minipage}
\begin{minipage}{0.4\textwidth}
\begin{tabular}{ll}
  ~\\
 \quad Structure                & \quad $\black{\tau_S}\br{\ov{v},E,\ov{e}}$ \\
 \quad Select                   & \quad $E.f$ \\
 \quad Type assertion           & \quad $E.(\black{\tau})$
\end{tabular}
\end{minipage}
}

Reduction \hfill \fbox{$d \becomes e$}
\begin{gather*}
\inferrule[r-field]
 { (\ov{f~\tau}) = \fields(\tau_S) }
 { \tau_S\br{\ov{v}}.f_i \becomes v_i }
\quad
\inferrule[r-call]
 { (x : \tau_S ,\ov{x : \tau}).e = \mbody(\vtype(v).m(\psi)) }
 { v.m(\psi)(\ov{v}) \becomes e[x \by v, \ov{x \by v}] }
\quad
\inferrule[r-assert]
 {\emptyset\vdash  \vtype(v) \imp \tau }
 { v.(\tau) \becomes v }
\quad
\inferrule[r-context]
 { d \becomes e }
 { E[d] \becomes E[e] }
\end{gather*}
\end{ruled}
\caption{FGG reduction}
\label{fig:fgg-reduction}
\end{Figure}

Figure~\ref{fig:fgg-reduction} presents the FGG reduction rules.

The adjustment to values, the auxiliary
function $\vtype$, and to evaluation contexts are
simple, replacing type names by types and adding
type arguments as appropriate.
Judgement $d \becomes e$ holds if expression $d$
steps to expression $e$.  Again, the adjustments
are all simple.

\subsection{FGG properties}
\label{sub:fgg-prop}

The results of the previous section adapt straightforwardly.

\begin{restatable}[Well formed]{lemma}{fggwftype}
  If\/ $\Delta \stoup \Gamma \vdash e : \tau$
  then\/ $\Delta \vdash \tau \ok$.
\end{restatable}

The substitution lemma is adapted to take into account that
in a method declaration the type parameters of the receiver
are substituted before the type parameters of the method.

\begin{lemma}[Substitution] Let $\eta = (\ov{\alpha \by \tau})$ be a substitution.
  \begin{itemize}
  \item
    If\/ $\emptyset \vdash \ov{\alpha \imp \tau_I}[\eta]$
    and\/ $\ov{\alpha : \tau_I} \comma \Delta \vdash \tau \ok$
    then\/ $\Delta[\eta] \vdash \tau[\eta] \ok$.
  \item
    If\/ $\emptyset \vdash \ov{\alpha \imp \tau_I}[\eta]$
    and\/ $\ov{\alpha : \tau_I} \comma \Delta \stoup \Gamma \vdash e:\tau$
    then\/ $\Delta[\eta] \stoup \Gamma[\eta] \vdash e[\eta] : \tau[\eta]$.
  \item
    If\/ $\emptyset \stoup \emptyset \vdash \ov{v:\tau}$
    and\/ $\emptyset \stoup \ov{x:\sigma} \vdash e:\sigma$
    and\/ $\emptyset \vdash \ov{\tau \imp \sigma}$
    then\/ $\emptyset \stoup \emptyset \vdash e[\ov{x \by v}]:\tau$
    for some type $\tau$ with $\emptyset \vdash \tau \imp \sigma$.
  \end{itemize}
\end{lemma}

The remaining results are easy to adjust.

\begin{theorem}[Preservation]
  If\/ $\emptyset \stoup \emptyset \vdash d : \sigma$ and $d \becomes e$
  then\/ $\emptyset \stoup \emptyset \vdash e : \tau$ for some $\tau$
  with $\emptyset \vdash \tau \imp \sigma$.
\end{theorem}

Expression $d$ \emph{panics} if
there exist evaluation context $E$, value $v$, and type $\tau$
such that $d = E[v.(\tau)]$ and $\emptyset \vdash \vtype(v) \notimp \tau$.

\begin{theorem}[Progress]
  If\/ $\emptyset \stoup \emptyset \vdash d:\sigma$ then
  either\/ $d$ is a value,
  $d \becomes e$ for some $e$,
  or\/ $d$ panics.
\end{theorem}

 \section{Monomorphisation}
\label{sec:mono}
The monomorphisation process consists of two phases.
In the first phase, a set of types and method instantiations are
collected from an FGG program.
In the second phase, an FGG program is translated to its FG equivalent
following the instance set computed in the first phase.

\begin{Figure}
\begin{ruledgo}
\begin{golsttwo}
type Any interface {}
type Int struct {}

type Event interface {
  Process(type b Any)(y b) Int
}
type UIEvent struct {}
func (x UIEvent) Process(type b Any)(y b) Int {
  return Int{} 
}

type Dispatcher struct {}
func (x Dispatcher) Dispatch(y Event) Int {
  return y.Process(Int)(Int{})
}

func main() {
  _ = Dispatcher{}.Dispatch(UIEvent{}) 
}(* \newpage *)
type Top struct {}
type Int struct {}
type Event interface {
  Process<Int>(y Int) Int
  Process<1>() Top 
}
type UIEvent struct {}
func (x UIEvent) Process<Int>(y Int) Int {
  return Int{} }
func (x UIEvent) Process<1>() Top {
  return Top{} }
type Dispatcher struct {}
func (x Dispatcher) Dispatch(y Event) Int {
  return y.Process<Int>(Int{}) }
func (x Dispatcher) Dispatch<1>() Top {
  return Top{} }
func main() {
_ = Dispatcher{}.Dispatch(UIEvent{}) 
}
\end{golsttwo}
\end{ruledgo}
\caption{Dispatcher example: FGG source (left) and FG translation (right)}
\label{fig:dispatch-fgg}
\end{Figure}

\ROUNDTWO{
Throughout this section, we illustrate the monomorphisation process
with the FGG program in Figure~\ref{fig:dispatch-fgg} (left).
This program contains a \goinl{Dispatcher} structure which processes
abstract \goinl{Event}s.
A  \goinl{Dispatcher} processes events, and events are objects that
can be processed.
For the sake of space, the program in Figure~\ref{fig:dispatch-fgg}
includes only one implementation of \goinl{Events}, i.e.,
\goinl{UIEvents}, but other implementations may be easily added
following the same pattern.  }

\subsection{Collecting type and method instances}
Let $\omega, \Omega$ range over instance sets, which contain elements
of type $\tau$ or pairs of a type with a method and its type
arguments, $\tau.m(\psi)$.
In  Figure~\ref{fig:fgg-omega-new}
we define a judgement $P \yields \Omega$ which computes the set of
instances of types and methods required to correctly monomorphise an
FGG program.

Judgement $\Delta \stoup \Gamma \vdash e \yields \omega$ holds if
$\omega$ is the instance set for expression $e$, given environments
$\Delta$ and $\Gamma$.
In the rules for variables, structure literals, field selections and type
assertions, we simply collect the occurrences of type instances and
proceed inductively. 
The rule for method calls additionally records the instantiation of the
method $\tau.m(\psi)$ where $\tau$ is the type of its receiver.
In the rules for structure literals and method calls, we assume that
sequences of instance sets, e.g., $\ov{\omega}$, coerce to a set
consisting of the union of the elements of the sequence.

The instance set of a program $P$ is the limit of function $G$ applied
to the instance set of its body.
$G_\Delta(\omega)$ is defined via four auxiliary functions that compute the
type and method instances required by $\omega$. 
It returns all type and method instantiations that are required
to monomorphise declarations ($\Fclo$, $\Mclo$) and to preserve the
$\imp$ relation ($\Iclo$, $\Sclo$).

$\Fclo$ finds all the type instances that occur in the declarations of
structures, while $\Mclo$ finds all the type instances that occur in
the declarations of method instances.

$\Iclo$ finds all the method signature instances that are required to
preserve the $\imp$ relation over interfaces.
$\Sclo$ finds all the type and method sets required by method calls,
inter-procedurally.
For each method instance $\tau.m(\psi)$ in $\omega$ it finds all
instance sets of all implementations of method $m$, following the
$\imp$ relation.
Intuitively, $\Iclo$ and $\Sclo$ are used to guarantee that
if $\tau \imp \sigma$, then the monomorphised version of $\tau$ also
implements the monomorphised version of $\sigma$.

\begin{Figure}
  \begin{ruled}
  Instance sets \hfill \fbox{$\omega, \Omega$}
  \begin{mathpar}
    \omega, \Omega
    \text{ range over sets containing elements of the form }
    \tau \text{ or } \tau.m(\psi).
  \end{mathpar}
  
  Expressions and programs \hfill \fbox{$ \Delta \stoup \Gamma \vdash e \yields \omega$} \quad \fbox{$P \yields \Omega$}
  \begin{mathpar}
    \inferrule[I-var]
    { 
    }
    {
     \Delta \stoup \Gamma \vdash x \yields \emptyset
    }

    \inferrule[I-literal]
    {
     \Delta \stoup \Gamma \vdash \ov{e \yields \omega}
    }
    {
      \Delta \stoup\Gamma \vdash \tau_S\br{\ov{e}} \yields \set{\tau_S} \cup \ov{\omega}
    }

    \inferrule[I-field]
    {
      \Delta \stoup\Gamma \vdash e \yields \omega
    }
    {
    \Delta \stoup  \Gamma \vdash e.f_i \yields \omega
    }

    \inferrule[I-assert]
    {
    \Delta \stoup  \Gamma \vdash e \yields \omega
    }
    {
      \Delta \stoup\Gamma \vdash e.(\tau) \yields \set{\tau} \cup \omega
    }

    \inferrule[I-call]
    {
      \Delta  \stoup \Gamma \vdash e : \tau \quad
      \Delta \stoup\Gamma \vdash e \yields \omega \quad
     \Delta \stoup \Gamma \vdash \ov{e \yields \omega}
    } 
    {
      \Delta \stoup \Gamma \vdash e.m(\psi)(\ov{e}) \yields
      \set{\tau \comma \tau.m(\psi)} \cup \omega \cup \ov{\omega}
    }

    \inferrule[I-prog]
    {
      \emptyset \stoup \emptyset \vdash e \yields \omega \\
      \Omega = \lim_{n \rightarrow \infty}  G_{\emptyset}^{n}(\omega)      
    }    
    {
      \package~\main;~\ov{D}~\func~\main()~\br{\un=e} \yields \Omega
    }
  \end{mathpar}

  Auxiliary functions  \hfill \ \vphantom{\fbox{$\Gamma$}} \begin{mathpar}
    G_\Delta(\omega)  =  \omega \cup                  
    \FExtensionD{\omega}{\Delta} \cup
    \MExtensionD{\omega}{\Delta} \cup
    \IExtensionD{\omega}{\Delta} \cup
    \SExtensionD{\omega}{\Delta}
    \\
    \FExtensionD{\omega}{\Delta} = \bigcup\left\{\Strut
      \ov{\tau}
      \;\middle|\;
      \tau_S \in \omega \comma
      (\ov{f~\tau}) = \fields(\tau_S)                  
    \right\}
    \\
    \MExtensionD{\omega}{\Delta}  = \bigcup\left\{\Strut
      \ov{\sigma}[\eta]  \cup
      \set{\sigma[\eta] }
      \;\middle|\;
      \tau.m(\psi) \in \omega \comma
      (m(\Psi)(\ov{x~\sigma})~\sigma) \in \methods_\Delta (\tau) \comma
      \eta = (\Psi \by \psi)
    \right\}
    \\
    \IExtensionD{\omega}{\Delta} = \left\{\Strut
      \tau'_I.m(\psi)
      \;\middle|\;
      \tau_I.m(\psi) \in \omega \comma
      \tau'_I \in \omega \comma
      \Delta \vdash \tau'_I \imp \tau_I  \right\}
    \\
    \SExtensionD{\omega}{\Delta} = \bigcup\left\{\Strut
      \set{\tau_S.m(\psi)} \cup \Omega
      \middle|
      \begin{array}{l}
        \tau.m(\psi) \in \omega \comma
        \tau_S \in \omega \comma
        \Delta \vdash \tau_S \imp \tau 
        \comma
\Delta \stoup x : \tau_S,\ov{x : \sigma} \vdash e \yields \Omega
        \\
        (x: \tau_S,\ov{x: \sigma}).e = \mbody(\tau_S.m(\psi))
      \end{array}
    \right\}
\end{mathpar}
  \end{ruled}
  \caption{Computing instance sets}
  \label{fig:fgg-omega-new}
\end{Figure}

\begin{Figure}
  \begin{ruled}
  Types and methods \hfill
  \fbox{$\eta \vdash \tau \mapsto t^\dagger$} \quad
  \fbox{$\eta \vdash t(\Phi) \mapsto t^\dagger$} \quad
  \fbox{$\eta \vdash m(\psi) \mapsto m^\dagger$} \quad
  \fbox{$\eta \vdash m(\Psi) \mapsto m^\dagger$}
  \begin{mathpar}
    \inferrule[m-type]
    { t^\dagger = \monoid{\tau[\eta]} }
    { \eta \vdash \tau \mapsto  t^\dagger }

    \inferrule[m-tformal]
    { t^\dagger = \monoid{t(\ov{\alpha}[\eta])} }
    { \eta \vdash t(\type~\ov{\alpha~\tau_I}) \mapsto t^\dagger }

    \inferrule[m-method]
    { m^\dagger = \monoid{m(\psi[\eta])} } 
    { \eta \vdash m(\psi) \mapsto m^\dagger }

    \inferrule[m-mformal]
    { m^\dagger = \monoid{m(\ov{\alpha}[\eta])} }
    { \eta \vdash m(\type~\ov{\alpha~\tau_I}) \mapsto m^\dagger }
  \end{mathpar}

  Expression \hfill
  \fbox{$\eta \vdash e \mapsto e^\dagger$}
  \vspace{-2ex}
  \begin{mathpar}
    \inferrule[m-var]
    { ~ }
    {  \eta \vdash x \mapsto x }
    \qquad
    \inferrule[m-value]
    {
      \eta \vdash \tau_S \mapsto t_S^\dagger
      \\
      \eta \vdash \ov{e \mapsto e^\dagger}
    }
    {
      \eta \vdash
      \tau_S\{\ov{e}\} \mapsto t^\dagger_S\{\ov{e^\dagger}\}
    }
    \qquad
    \inferrule[m-select]
    { \eta \vdash e \mapsto e^\dagger }
    { \eta \vdash e.f \mapsto e^\dagger.f }
    \\
    \inferrule[m-call]
    {
      \eta \vdash e \mapsto e^\dagger
      \\
      \eta \vdash m(\psi) \mapsto m^\dagger
      \\
      \eta \vdash \ov{e \mapsto e^\dagger}
    }
    {
      \eta \vdash
      e.m(\psi)(\ov{e}) \mapsto e^\dagger.m^\dagger(\ov{e^\dagger})
    }
    \qquad
    \inferrule[m-assert]
    {
      \eta \vdash e \mapsto e^\dagger
      \\
      \eta \vdash \tau \mapsto t^\dagger
    }
    { \eta \vdash e.(\tau) \mapsto e^\dagger.(t^\dagger) }
  \end{mathpar}

  Method signature \hfill
  \fbox{$\eta \vdash M \mapsto M^\dagger$}
  \quad
    \fbox{$\eta \vdash S \mapsto S^\dagger$}  
  \vspace{-2ex}
  \begin{mathpar}
    \inferrule[m-sig]
    {
      \eta \vdash \ov{\tau \mapsto t^\dagger}
      \\
      \eta \vdash \tau \mapsto u^\dagger
    }
    {
      \eta \vdash
      (\ov{x~\tau})~\tau \mapsto (\ov{x~t^\dagger})~u^\dagger
    } 

      \inferrule[m-id]
      {
        \mkdummy(mM[\eta]) = m^\ast
      }
      {
        \eta \vdash mM  \mapsto m^\ast()~{\dummytype}
      }
  \end{mathpar}

  Type literal \hfill
  \fbox{$\eta \stoup \mu \vdash T \mapsto T^\dagger$}
  \begin{mathpar}
    \inferrule[m-struct]
    {
      \eta \vdash \ov{\tau \mapsto t^\dagger}
    }
    {
      \eta \stoup \mu \vdash
      \struct\br{\ov{f~\tau}} \mapsto
      \struct\br{\ov{f~t^\dagger}}
    }
    \qquad
    \inferrule[m-interface] 
    {
      \eta \stoup \mu \vdash \ov{S \mapsto \calS}
    }
    { \eta \stoup \mu \vdash
      \interface\br{\ov{S}} \mapsto
      \interface\br{\bigcup \ov{\calS}}
    }
  \end{mathpar}
  \end{ruled}
  \caption{Monomorphisation of FGG into FG --- name mapping}
  \label{fig:rcver-mono-subs}
\end{Figure}

\begin{Figure}
  \begin{ruled}
  Interface specification \hfill
  \fbox{$\eta \stoup \mu \vdash S \mapsto \calS$}
  \begin{mathpar}
    \inferrule[m-spec]
    {
      \calS = \left\{\; \Strut
        m^\dagger N^\dagger
        \;\middle|\; 
        m(\psi) \in \mu \comma
        \theta = (\eta, \Psi \by \psi) \comma
        \theta \vdash m(\Psi) \mapsto m^\dagger \comma
        \theta \vdash N \mapsto N^\dagger
        \;\right\}
      \and
      \eta \vdash m(\Psi)N \mapsto S^\dagger
    }
    {
      \eta \stoup \mu \vdash
      m(\Psi)N\mapsto \calS~
      \cup~
      \{  S^\dagger \}
    }
  \end{mathpar}

  Declaration \hfill \fbox{$\Omega \vdash D \mapsto \mathcal{D}$}
  \begin{mathpar}
    \inferrule[m-type]
    {
      \calD = \left\{\; \Strut
        \type~t^\dagger~T^\dagger
        \;\middle|\;
        t(\phi) \in \Omega \comma
        \eta = (\Phi \by \phi) \comma
        \mu = \set{ m(\psi) \mid t(\phi).m(\psi) \in \Omega } \comma
        \eta \stoup \mu \vdash T \mapsto T^\dagger
        \;\right\}
    }
    { 
      \Omega \vdash \type~t(\Phi)~T \mapsto \calD
    }
    \\
    \inferrule[m-func]
    {
      \calD = \left\{\;
        {
            \func~(x~t_S^\dagger)~m^\dagger N^\dagger~\br{\return~e^\dagger}
        }
        \;\middle|
        { \begin{array}{l}
            t_S(\phi).m(\psi) \in \Omega \comma
            \theta = (\Phi \by \phi, \Psi \by \psi) \comma
            \theta \vdash t_S(\Phi) \mapsto t_S^\dagger
            \\
            \theta \vdash m(\Psi) \mapsto m^\dagger \comma
            \theta \vdash N \mapsto N^\dagger \comma
            \theta \vdash  e \mapsto e^\dagger
          \end{array}
        }
      \right\}
      \\\\
        \calD' = \left\{\;
          {
            \func~(x~t_S^\dagger)~S^\dagger~\br{\return~\dummytype\br{}}
          }
          \;\middle| \;
          {
            t_S(\phi) \in \Omega \comma
            \eta = (\Phi \by \phi) \comma
            \eta \vdash t_S(\Phi) \mapsto t_S^\dagger \comma
            \eta \vdash m(\Psi)N \mapsto S^\dagger
          }
        \right\}
    }
    {
      \Omega \vdash
      \func~(x~t_S(\Phi))~m(\Psi)N~\br{\return~e} \mapsto \calD \cup\calD'
    }
  \end{mathpar}

  Program \hfill \fbox{$\vdash P \mapsto P^\dagger$}
  \begin{mathpar}
    \inferrule[m-program]
    {
      \package~\main;~\ov{D}~\func~\main()~\br{\un=e} \yields \Omega
      \\\\
      \Omega \vdash \ov{D \mapsto \calD}
      \and
      \ov{D^\dagger} = \{ \type~\dummytype~\struct~\br{} \}  \cup  \bigcup\ov{\calD}~
      \and
      \emptyset \vdash e \mapsto e^\dagger
    }
    {
      \vdash
      \package~\main;~\ov{D}~\func~\main()~\br{\un=e} \mapsto
      \package~\main;~\ov{D^\dagger}~
     \func~\main()~\br{\un=e^\dagger}
    }
  \end{mathpar}
  \end{ruled}
  \caption{Monomorphisation of FGG into FG --- instance generation, where $N$
    ranges over $(\ov{x~\tau})~\tau$}
  \label{fig:rcver-mono-new}
\end{Figure}

\ROUNDTWO{
Consider the example in Figure~\ref{fig:dispatch-fgg} (left). For the
top-level method, we have
\begin{center}
  $
\emptyset \stoup \emptyset \vdash
\text{\goinl{Dispatcher\{\}.Dispatch(UIEvent\{\})}}
\yields
\{
\text{\goinl{Dispatcher}}, \ \text{\goinl{Dispatcher.Dispatch()}}
\}
$
\end{center}
Posing
$\omega_0 = \{ \text{\goinl{Dispatcher}}, \
\text{\goinl{Dispatcher.Dispatch()}}
\} $, we compute the limit of $G$ applied to $\omega_0$.
We have $G_\emptyset( \omega_0 ) = \omega_0 \cup \{
\text{\goinl{Event}}, \
\text{\goinl{Int}}, \
\text{\goinl{Dispatcher.Dispatch()}}, \
\text{\goinl{Event}}, \
\text{\goinl{Event.Process(Int)}}
\}$
where
\goinl{Event} and \goinl{Int} are obtained from
$\MExtensionD{\omega_0}{\emptyset}$, while
\goinl{Dispatcher.Dispatch()},
\goinl{Event}, and
\goinl{Event.Process(Int)} are obtained from
$\SExtensionD{\omega_0}{\emptyset}$. 

The limit of function $G$ is reached after two iterations, i.e.,
$ \lim_{n \rightarrow \infty}  G_{\emptyset}^{n}(\omega_0)  =
G_\emptyset( G_\emptyset (\omega_0)) = G_\emptyset (\omega_0) \cup
\{ \text{\goinl{UIEvent.Process(Int)}} \}$.
Note that \goinl{UIEvent.Process(Int)} is obtained via $\Sclo$ using the
fact that \goinl{UIEvent} $\imp$ \goinl{Event} holds.
}

\subsection{Monomorphisation judgement}

We now define a judgement $\vdash P \mapsto P^\dagger$ where $P$ is a
program in FGG, and $P^\dagger$ is a corresponding monomorphised
program in FG.
This judgement is in turn defined by judgements for each of the
syntactic categories in FGG.  Some of the judgements are also
parameterised by instance sets (ranged over by $\Omega$),
substitutions that map type parameters to ground types (ranged over by
$\eta$), or method instance sets (ranged over by $\mu$).

Figure~\ref{fig:rcver-mono-subs} formalises how we recursively apply a
consistent renaming to generate FG code.
To monomorphise types, type formals, method names, and method formals,
given a substitution $\eta$, we assume a map from closed types to
identifiers.
For instance, if \id{f} is a type with two \ROUNDTWO{arguments}, and \id{g} and
\id{h} are types with no \ROUNDTWO{arguments}, then closed type
$\id{f}(\id{g}(),\id{h}())$ might correspond to the identifier
``$\id{f<g<>,h<>}\id{>}$'' assuming ``$\id{<,>}$'' are allowed as
letters in identifiers.  We write $t^\dagger = \monoid{\tau}$ to
compute the identifier $t^\dagger$ that corresponds to closed type
$\tau$. Similarly, we write $m^\dagger = \monoid{m(\psi)}$ to compute
the identifier $m^\dagger$ that corresponds to closed method
instantiation $m(\psi)$.

To monomorphise an expression given a substitution $\eta$, we
recursively monomorphise all the types and expressions contained
within this expression.
We proceed similarly to monomorphise method signatures in Rule
\textsc{m-sig}.

Rule \textsc{m-id} is used to generate a dummy method signature that
represents uniquely the alpha-equivalence class of its FGG
counterpart. The signature specifies no parameters and the return type
$\dummytype$.
It is necessary to generate such methods to ensure that if a type does
not implement another in an FGG program, then this is also the case in its
monomorphised counterpart.
We assume that $\mkdummy(mM_1) = \mkdummy(mM_2) $ for all $M_1 = M_2$,
using the same notion of equality as in $\unique$ and $\imp$.

To monomorphise a structure given a substitution, we recursively
monomorphise all the types contained within its field declarations.
To monomorphise an interface, we recursively monomorphise each of its
signatures and flatten the result in a single sequence of declarations.

Figure~\ref{fig:rcver-mono-new} formalises how declarations are
generated from $\Omega$. Here we let $N$ range over
$(\ov{x~\tau})~\tau$.

To monomorphise an interface specification we pass \emph{two}
environments. 
One is a substitution from type parameters to ground types $\eta$ and
the other is a set of method instances $\mu$, i.e., a set of entries
of type $m(\psi)$.
For each entry in $\mu$, we compute a new substitution $\theta$ which
extends $\eta$ and is used to generate a monomorphised instance of the
corresponding signature.
In addition, we generate dummy signature $S^\dagger$ that uniquely
identifies the FGG signature.
Each parameterised method may produce zero or more
monomorphised instances, plus a dummy method signature.

To monomorphise the declaration of a type $t$ given an instance set,
for each instance of $t$, we generate a substitution $\eta$ and
a method instance set $\mu$.
Then we recursively produce a monomorphised declaration for each
generated pair of $\eta$ and $\mu$. Note that each type declaration
may produce zero or more monomorphised declarations.

To monomorphise a method declaration given an instance set, we compute
a substitution $\theta$ for each method instance in $\Omega$.
Then we produce a monomorphised version of a method for each of its
instantiations. Note that each method declaration may produce zero or
more monomorphised declarations.
In addition, for each type instance and each of its methods, we
generate a dummy method that returns an instance of $\dummytype$.

To monomorphise a program $P$, we compute its instance set $\Omega$
then monomorphise its declaration and body given with respect to
$\Omega$. We additionally add the declaration of the empty structure
$\dummytype$.

\ROUNDTWO{
The FGG program in Figure~\ref{fig:dispatch-fgg} (left) is translated to
the FG program on the righthand side of the figure.
The translation starts with rule \tyrulename{m-program} where $\Omega
= \{$\goinl{Int},
\goinl{Event}, 
\goinl{Event.Process(Int)}, 
\goinl{UIEvent}, 
\goinl{UIEvent.Process(Int)},
\goinl{Dispatcher},
\goinl{Dispatcher.Dispatch()}$\}$.
Rule \tyrulename{m-type} is used to generate the instances of types
\goinl{Int}, \goinl{Event}, \goinl{UIEvent}, and
\goinl{Dispatcher}.
Rule \tyrulename{m-func} is used to generate the methods
\goinl{UIEvent.Process(Int)} and \goinl{Dispatcher.Dispatch()}, as
well as their dummy counterparts \goinl{Process<1>()} and
\goinl{Dispatch<1>()}.
Rule \tyrulename{m-spec} is used to generate the method signatures of
interface \goinl{Event}, i.e., \goinl{Event.Process(Int)}  and its dummy
counterpart \goinl{Process<1>()}.
}

\subsection{Ruling out non-monomorphisable programs}
\label{sub:polyrec}

It is not possible to monomorphise all FGG programs since programs
that contain polymorphic recursive methods may produce infinitely many
type instances.
To address this issue, we propose a predicate
$P \notmonomorphisable$ which conservatively identifies programs that
can produce infinitely many type instances.
Programs for which this predicate does \emph{not} hold are guaranteed to be
monomorphisable. \ROUNDTWO{Note that there exist
programs which produce finitely many type instances but for which
the predicate does hold, e.g., programs containing a polymorphic
recursive method that is never called.}

The predicate is formally defined in
Figure~\ref{fig:fgg-polyrec-check}, notably reusing our instance
generation procedures.
$P \notmonomorphisable$ holds if $D \notmonomorphisable$ holds for at
least one of its method declarations.
$D \notmonomorphisable$ holds if it is possible to find an element in
its instance set, inductively constructed using function $\Sclo$, such
that the occurs check is satisfied.
The occurs check holds when a type variable appears under a type
constructor in a type instance or method call (in the same position it
occupies in the type or method formal).
We write $\FV{\tau}$ for the set of type parameters occurring in $\tau$.

\begin{Figure}
  \begin{ruled}
  Not monomorphisable\hfill 
  \fbox{$P \notmonomorphisable$} \qquad \fbox{$D\notmonomorphisable$} 
  \begin{mathpar}
    \inferrule[mn-program]
    {
      D_i \notmonomorphisable
    }
    { 
      \package~\main;~\ov{D}~\func~\main()~\br{\un=e} \notmonomorphisable 
    }

    \inferrule[mn-func]
    {
      \Delta = 
      {\Phi, \Psi} 
      \and
      \text{exists}~n \in \naturals~\text{and}~
      (t_S(\phi).m(\psi)) \in G_{\Delta}^{n}(\{ 
      t_S(\hat\Phi),
      t_S(\hat\Phi).m(\hat{\Psi})
      \}) ~\text{s.t.}~
      \Phi \prec  \phi  ~\text{or}~ \Psi \prec \psi
}
    {  
      \func~(x~t_S(\Phi))~m(\Psi)(\ov{y~\tau})~\tau~\br{\return~e}  \notmonomorphisable 
    }
  \end{mathpar}

  Occurs check \hfill
  \fbox{$\Phi \prec \phi $} \qquad \fbox{$\alpha \prec \tau$}
  \begin{mathpar}
    \inferrule
    {
      \phi = \ov{\tau}
      \and
      \alpha_i \prec \tau_i 
    }
    { (\type~\ov{\alpha~\tau_I}) \prec \phi}

    \inferrule
    {
      \tau \neq \alpha
      \and
      \alpha \in \FV{\tau}
    }
    { \alpha \prec \tau } 
  \end{mathpar}
\end{ruled}
  \caption{Monomorphisability check}
  \label{fig:fgg-polyrec-check}
\end{Figure}

\subsection{Monomorphisation properties}
\label{sub:mono-prop}

\ROUNDTWO{
Not all programs are monomorphisable, however we can decide whether a
program is monomorphisable with the $P \notmonomorphisable$
  predicate}.

\begin{theorem}[Decidability]\label{thm:polyrec-decide}
  If\/ $P \ok$ then it is decidable whether or not\/
  $P \notmonomorphisable$ holds.
\end{theorem}

\ROUNDTWO{
For all well-typed programs $P$ such that $P \notmonomorphisable$ does
not hold, their instance sets are finite.}
\begin{theorem}[Monomorphisable]\label{thm:nopolyrec-mono}
  If\/ $P \ok$
  and\/ $P \notmonomorphisable$ doesn't hold
  then\/ $P \yields \Omega$
  with\/ $\Omega$ finite.
\end{theorem}

Monomorphisation preserves typing, \ROUNDTWO{i.e., the translation of
  a well-typed FGG program is a well-typed FG program.
}
\begin{theorem}[Sound]
  \label{thm:main:monotsound}
  If \/ $P \ok$
  and\/ $\vdash P \mapsto P^\dagger$
  then $P^\dagger \ok$.
\end{theorem}

\ROUNDTWO{
The reduction behaviour of well-typed FGG programs is preserved and
reflected by their monomorphised counterpart (see
Figure~\ref{fig:bisim}).}
Write $\vdash d \mapsto d^\dagger$ to abbreviate
$\emptyset \stoup \emptyset \vdash d \mapsto d^\dagger$.

\begin{theorem}[Bisimulation]
  \label{thm:main:correspondence}
  Let \/ $P \ok$
  and\/ $\vdash P \mapsto P^\dagger$
  with\/ $P = \ov{D} \prog d$
  and\/ $P^\dagger = \ov{D^\dagger} \prog d^\dagger$.
  Then:
  \begin{itemize}
  \item[(a)]
  if\/ $d \becomes e$ 
  then there exists\/ $e^\dagger$
  such that\/ $d^\dagger \becomes e^\dagger$
  and\/ $\vdash e \mapsto e^\dagger$;
  \item[(b)]
  if\/ $d^\dagger \becomes e^\dagger$
  then there exists\/ $e$ 
  such that\/ $d \becomes e$ 
  and\/ $\vdash  e \mapsto e^\dagger$.
  \end{itemize}
\end{theorem}

 \section{Implementation}
\label{sec:implementation}

We have made available a prototype implementation,\footnote{\url{https://github.com/rhu1/fgg/}}
which contains FG and FGG type checkers and interpreters,
and a monomorphiser from FGG to FG
(including the $\notmonomorphisable*$ check).
We wrote the implementation in Go to facilitate interactions with the
Go designers and community.
Our interface includes some extensions to our tiny syntax for FG and FGG,
such as direct support for interface embedding, some primitive data
types, and minimal I/O.
Versions of all the presented examples have been tested using the implementation.

We took advantage of the implementation to perform extensive testing.
FG evaluation results are compared to those using the official Go
compiler, and the FG and FGG interpreters support dynamic checking of
preservation and progress.
To test monomorphisation, we added the
test of bisimulation depicted in Figure~\ref{fig:bisim}:
given a well-typed FGG program
we generate its FG monomorphisation;
we step the FGG and FG terms and
confirm that the new FGG term monomorphises to the new FG term;
and repeat until termination.

Besides handcrafted examples and stress tests, we 
used NEAT~\citep{Claessen15,Duregard16} to lazily enumerate all well-typed
programs (up to some size relative to the total number of occurrences
of method and type symbols) from a subset of FGG (similar to
\textsc{SmallCheck}~\citep{Runciman08}).
The subset consists of programs which have:
\begin{enumerate*}
\item at least one method and one field;
\item at most one empty interface; and
\item at most two empty structs.
\end{enumerate*}
And where:
\begin{enumerate*}
\item each method has at most two arguments;
\item each struct has at most two fields;
\item each interface has at most two members and two parents; and
\item each method or type constructor has at most two type parameters.
\end{enumerate*}
Moreover, we disallow mutually recursive type definitions.
These measures are taken to truncate the space of possible programs.
We generate all FGG programs in our subset up to size 20,
and confirm they pass the bisimulation test described above.

\begin{Figure}
\begin{ruled}
\begin{tikzpicture}
\node (p0) {$d$};
\node[anchor=center,below=of p0.center] (m0) {$d^\dagger$};
\node[anchor=center,right=of p0.center] (p1) {$e_1$};
\node[anchor=center,right=of m0.center] (m1) {$e_1^\dagger$};
\node[anchor=center,right=of p1.center] (pdots) {$...$};
\node[anchor=center,right=of m1.center] (mdots) {$...$};
\node[anchor=center,right=of pdots.center] (pi) {$e_n$};
\node[anchor=center,right=of mdots.center] (mi) {$e_n^\dagger$};
\draw[->,densely dashed] (p0) --
  node[left]{$\vdash d \mapsto d^\dagger$} (m0);
\draw[->] (p0) -- (p1);
\draw[->] (m0) -- (m1);
\draw[->,densely dashed] (p1) --
node[right]{$\vdash e_1 \mapsto e_1^\dagger$} (m1);
\draw[->] (p1) -- (pdots);
\draw[->] (m1) -- (mdots);
\draw[->] (pdots) -- (pi);
\draw[->] (mdots) -- (mi);
\draw[->,densely dashed] (pi) --
  node[right]{$\vdash e_n \mapsto e_n^\dagger$} (mi);
\node[anchor=south east,left=20mm of m0.south west,align=left,yshift=10mm,text width=35mm]
{Assume $P \ok$ \\
  and $\vdash P \mapsto P^\dagger$ \\
  with $P = \ov{D} \prog d$ \\
  and $P^\dagger = \ov{D^\dagger} \prog d^\dagger$};
\end{tikzpicture}
\end{ruled}
\caption{Bisimulation}
\label{fig:bisim}
\end{Figure}

 \section{Related Work}
\label{sec:related}

This paper is the first to present a core formalism
for the Go language.
Our presentation is styled after that of Featherweight
Java by \citet{Igarashi-et-al-2001}.
Like them, we focus on a tiny, functional subset of the language;
we define versions with and without generics;
and we consider translation from one to the other.
We also mark as ``stupid'' casts/type assertions
that are disallowed in source but are required for reduction
to preserve types.

Our work resembles the development of
generics for Java by \citet{Bracha-et-al-1998} and for
.NET by \cite{Kennedy-and-Syme-2001} in that we build on
a well-established base language. \ROUNDTWO{We note that in Featherweight Go,
Featherweight Generic Go (and in the Go language), since method
signatures are nonvariant, there are no fundamental decidability issues
related to F-bounded polymorphism and
variance~\cite{DBLP:conf/popl/Pierce92}, and so there is no need to
consider more sophisticated techniques such as those of
\citet{DBLP:conf/pldi/GreenmanMT14} to ensure decidability of type
checking.}

\ROUNDTWO{In terms of formalisations of generics}, prior work on generics adopts one, or a combination, of three main approaches,
\emph{erasure}, \emph{runtime representation of types as values}, and
\emph{monomorphisation}.

\paragraph{Erasure}
\citet{Bracha-et-al-1998} present a translation from Java with
generics to Java without generics that erases all information
about type parameters.  The translation relies on \emph{bridge} methods,
which in turn rely on method overloading, which is not supported in Go.
\citet{Igarashi-et-al-1999} formalised the translation for the FJ
subset of Java (avoiding bridge methods) and
proved it preserves typing and reductions.
Downsides of erasure are that casts to generic types must be
restricted and creation of generic arrays becomes tricky (see
\cite{Naftalin-and-Wadler-2006}).
Moreover, erased code is often less efficient than monomorphised code.
An upside is that erasure
is linear in the size of the source, whereas monomorphisation
can lead to an explosion in code size.

\paragraph{Runtime representation}
In contrast to Java erasure, \citet{Kennedy-and-Syme-2001} developed
an extension of the .NET Common Language Runtime (CLR) and \CS{}
with direct support for generics.
They generate a mixture of specialised and shared code:
the former is compiled separately for each primitive type
and is similar to monomorphisation;
the latter is compiled once for every object type
and is similar to erasure.
JIT compilation is exploited to perform
specialisation \emph{dynamically}, avoiding potential code bloat.
Code sharing is implemented by storing runtime 
\emph{type-reps} \cite{DBLP:conf/icfp/CraryWM98} for type parameters.

The overhead of runtime assembly of type-reps can be optimised by
pre-computing and caching maps from open types to their reps when a generic
type or method is instantiated
\cite{Kennedy-and-Syme-2001,DBLP:conf/oopsla/ViroliN00}.
In future work, we will also look to techniques of optimising the coexistence
of uniform (boxed) and non-uniform representations in polymorphic
code \cite{DBLP:conf/popl/Leroy92} for dealing with the analogous mixture of
struct and interface values in generic Go code.

\paragraph{Monomorphisation}
Although monomorphisation has been employed for languages such as
\Cpp{}, Standard ML and Rust~\cite{Turon-2015},
we found a relative lack of peer-reviewed literature on this topic.
This section discusses works related to monomorphisation that
do not state or prove any correctness results.

\citet[Chapter~26]{Stroustrup-2013} describes template instantiation in
\Cpp{}. It is widely used, and infamous for code bloat.

\citet{DBLP:conf/icfp/BentonKR98} describes a whole-program compiler
from SML'97 to Java bytecode, where polymorphism is fully
monomorphised. Monomorphisation is alway possible, since
Standard ML rules out polymorphic recursion (unlike FGG).
\citet{Fluet-2015} sketches a similar approach used in the SML
optimising compiler MLton.

\citet[Section 6]{Tolmach-and-Oliva-1998} develop a typed translation
from an ML-like language to Ada, based on monomorphisation and
presented in detail.  Unlike us, they do not address subtyping
(structural or otherwise) and they presume the absence of polymorphic
recursion.

\citet{Jones-1995} describes the use of specialisation to efficiently
compile type classes for Haskell, which bears some resemblance to
monomorphisation.

\paragraph{Formalisation}
We now consider works that formalise some aspect of monomorphisation.

\citet{DBLP:conf/popl/YuKS04} formalise the mixed specialisation and
sharing mechanisms of the .NET JIT
compiler~\cite{Kennedy-and-Syme-2001}.  The work describes a type and
semantics preserving translation to a \emph{polymorphic} .NET
Intermediate Language (IL), where polymorphic behaviours are driven by
\emph{type-reps} \cite{DBLP:conf/icfp/CraryWM98}, codifying runtime
type data that can be used in e.g. dynamic casts. Their approach only
generates code where type variables are instantiated with basic data
types, using a uniform (i.e.~boxed) representation for all other
types.  This sidesteps the key challenges of monomorphising code with
polymorphic recursion and parameterised methods. Notably,
\citet{Kennedy-and-Syme-2001} state that ``some polymorphism simply
cannot be specialized statically (polymorphic recursion, first-class
polymorphism)''. In contrast, we present an algorithm that can
determine whether monomorphisation is possible in the presence of
polymorphic recursion.

\citet{Siek-and-Taha-2006} formalise the \Cpp{} template instantiation
mechanism.  They model partial specialization, template
parameterisation and instantiation, and prove type soundness of
template expansion.  Unlike us, they do not state or prove
bisimulation or preservation of reductions.  Since
\Cpp{} templates are Turing-complete, their soundness results are
modulo termination, whereas our algorithms are guaranteed to
terminate (see Theorem~\ref{thm:polyrec-decide} \ROUNDTWO{\iflong
  and Appendix~\ref{app:nomono}\else
  and~\cite{techreport}\fi}).

\citet[Section 2]{TanakaAG18} report on a monomorphisation algorithm
for Coq (Gallina) used in generation of low-level C code. Unlike us,
they do not test for polymorphic recursion.

\paragraph{Monomorphisation and logic}
In a related area, \citet{BlanchetteB0S16,BobotP11} study
monomorphisation for polymorphic first-order logic formulas,
targeting the untyped or multi-sorted logics found in
automated theorem provers.

 \section{Conclusion}
\label{sec:conclusion}

\ROUNDTWO{In this
  work we studied generics for a minimal subset of Go and their
  compilation via monomorphisation. We chose monomorphisation since it
is a simple first way of concretely explaining the semantics of
generics using (simplified) Go and
it avoids restrictions required by the erasure-based approach used in
FGJ (e.g.~type assertions on type variables).
Another key benefit of monomorphisation
is that of enabling 0-cost abstractions -- programs that do not use
generics incur no runtime penalty and generic code is translated
into code that is equivalent to hand coded instantiations. The cost is
that of requiring a whole program analysis and
disallowing programs that would result in infinite instantiations
(Section~\ref{sub:polyrec}). Clearly, this is the beginning of the
story, not the end.}

In future
work, we plan to look at other methods of implementation
beside monomorphisation, and in particular \ROUNDTWO{we plan} to consider an
implementation based on passing runtime representations of
types, similar to that used for .NET generics
\citep{Kennedy-and-Syme-2001}. \ROUNDTWO{The idea is to automatically
  equip methods and
  structs with data that codifies type arguments used in generic
  code at runtime. For instance, a \goinl{Cons} struct would carry at
  runtime an additional field that specifies the type of the element
 contained in the cell, and a method that constructs a tree from a
 \goinl{List} would thread the runtime type information of the list
 cells into the tree. This approach requires translating all structs
 and methods of a program to account for runtime type passing and
 construction and thus the resulting programs can incur some
 runtime penalty.}
A mixed approach that uses monomorphisation sometimes
and passing runtime representations sometimes might be best,
again similar to that used for .NET generics. \ROUNDTWO{We will study the
trade-offs and performance impact of this spectrum of approaches in
future work.}

Featherweight
Go is restricted to a tiny subset of Go. We plan a model of
other important features such as assignments, arrays,
slices, and packages, which we will dub Bantamweight Go; and
a model of Go's innovative concurrency mechanism based on
``goroutines'' and message passing, which we will dub
Cruiserweight Go.

\begin{acks}
  We thank Nicholas Ng, Sam Lindley, \ROUNDTWO{and our referees}
  for comments and
  suggestions. This work was funded under EPSRC EP/K034413/1,
  EP/T006544/1, EP/K011715/1, EP/L00058X/1, EP/N027833/1,
  EP/N028201/1, EP/T006544/1, EP/T014709/1 and EP/V000462/1,
  NCSS/EPSRC VeTSS, NOVA LINCS (UIDB/04516/2020) with the financial
  support of FCT- Fundação para a Ciência e a Tecnologia, and EU
  MSCA-RISE BehAPI (ID:778233).
\end{acks}

\bibliography{main}

\iflong

\clearpage
\newpage
\appendix
\section{Further examples}
\label{sec:examples-more}

\subsection{Booleans in FG}
\label{sec:fg-bool}

\begin{Figure}
\begin{ruled}
\begin{golst}
    type Any interface {}
    type Eq interface {
        Equal(that Eq) Bool
    }
    type Bool interface {
        Not() Bool
        Equal(that Eq) Bool
        Cond(br Branches) Any
    }
    type Branches interface {
        IfTT() Any
        IfFF() Any
    }

    type TT struct {}
    type FF struct {}

    func (this TT) Not() Bool { return FF{} }
    func (this FF) Not() Bool { return TT{} }

    func (this TT) Equal(that Eq) Bool { return that.(Bool) }
    func (this FF) Equal(that Eq) Bool { return that.(Bool).Not() }

    func (this TT) Cond(br Branches) Any { return br.IfTT() }
    func (this FF) Cond(br Branches) Any { return br.IfFF() }
\end{golst}
\end{ruled}
\caption{Booleans in FG}
\label{fig:fg-bool}
\end{Figure}

Figure~\ref{fig:fg-bool} shows how to implement booleans in FG. We
begin by declaring two general-purpose interfaces. Interface
\goinl{Any} has no methods, and so is implemented by any
type. Interface \goinl{Eq} has one method with signature
\goinl{Equal(that Eq) Bool}, an equality test that
accepts an argument which is itself of type \goinl{Eq}
and returns a boolean.

Interface \goinl{Bool} has three methods. Method \goinl{Not() Bool}
computes the logical negation of its receiver. Method
\goinl{Equal(that Eq) Bool} checks whether its receiver is equal to
its argument. Method \goinl{Cond(br Branches) Any} executes one of two
branches depending on whether the receiver is true or false.  It
refers to interface \goinl{Branches}, which has two methods with
signatures \goinl{IfTT() Any} and \goinl{IfFF() Any}, one to be
invoked for a true conditional and one to be invoked for a false
conditional.

Go also supports \emph{interface embedding}.  Mentioning one interface
in another stands for all the methods declared in the first.  For
example, the \goinl{Bool} interface in the figure is equivalent to:
\begin{golst}
    type Bool interface {
      Not() Bool
      Eq
      Cond(br Branches) Any
    }
\end{golst}
where the method specification for \goinl{Equal} has been replaced by
the interface \goinl{Eq}.  We use interface embedding in examples, but
our formalism for FG assumes each interface embedding is expanded out
to the corresponding method specifications.

We then define two structure types, \goinl{TT} and \goinl{FF}, each
with no fields, and methods are defined with receivers of both types
for each of the three methods \goinl{Not}, \goinl{Equal}, and
\goinl{Cond}.  Since they have methods whose signatures match those in
interface \goinl{Bool}, we say that structure types \goinl{TT} and
\goinl{FF} both \emph{implement} the interface \goinl{Bool}.

Method \goinl{Not} returns false on receiver true, and true on receiver
false. For instance, \goinl{TT\{\}.Not()} returns \goinl{FF\{\}}.

Method \goinl{Equal} returns its argument on true and the negation of its
argument on false. The type assertion \goinl{that.(Bool)}
returns the argument \goinl{that} if it implements the interface
\goinl{Bool} and panics otherwise, where \emph{panic} is Go jargon for
indicating a runtime error has occurred. For instance,
\goinl{b.Equal(b)} returns true if \goinl{b} is a boolean,
while \goinl{b.Equal(x)} panics if \goinl{b} is a boolean but \goinl{x} is not.

Finally, method \goinl{Cond} invokes either \goinl{IfTT()} or
\goinl{IfFF()} on its argument, depending on whether the receiver is
true or false. For instance, assume \goinl{x} has type \goinl{Eq} and
\goinl{xs} has type \goinl{Cons}, and that \goinl{Equal} and
\goinl{Contains} are both methods that return booleans.
The conditional
\begin{golst}
    if x.Equal(xs.head) {
        return true
    } else {
        return xs.tail.Contains(x)
    }
\end{golst}
can be emulated in FG by introducing the declarations 
\begin{golst}
    type containsBr struct {
        xs List
        x Eq
    }
    func (this containsBr) IfTT() Bool { return TT{} }
    func (this containsBr) IfFF() Bool { return this.xs.tail.Contains(this.x) }
\end{golst}
and writing
\begin{golst}
    return x.Equal(xs.head).Cond(containsBr{x,xs}).(Bool)
\end{golst}
The type assertion \goinl{.(Bool)} is required since otherwise the
type of the expression would be \goinl{Any}, whereas \goinl{Bool} is
expected.  As one would expect, method \goinl{Contains} is called only
when the condition is true.

\subsection{Booleans in FGG}
\label{sec:fgg-bool}

\begin{Figure}
\begin{ruled}
\begin{golst}
    type Any interface {}
    type Eq(type a Eq(a)) interface {
        Equal(that a) Bool
    }
    type Bool interface {
        Not() Bool
        Equal(that Bool) Bool
        Cond(type a Any)(br Branches(a)) a
    }
    type Branches(type a Any) interface {
        IfTT() a
        IfFF() a
    }

    type TT struct {}
    type FF struct {}

    func (this TT) Not() Bool { return FF{} }
    func (this FF) Not() Bool { return TT{} }

    func (this TT) Equal(that Bool) Bool { return that }
    func (this FF) Equal(that Bool) Bool { return that.Not() }

    func (this TT) Cond(type a Any)(br Branches(a)) a { return br.IfTT() }
    func (this FF) Cond(type a Any)(br Branches(a)) a { return br.IfFF() }
\end{golst}
\end{ruled}
\caption{Booleans in FGG}
\label{fig:fgg-bool}
\end{Figure}

Booleans are adapted to generics in Figure~\ref{fig:fgg-bool}.
Interface \goinl{Any} is unchanged. The interface for equality is now
written \goinl{Eq(type a Eq(a))}, indicating that it
accepts a type parameter.  A type parameter is always followed by an
interface that it must implement, which is called its \emph{bound}.
Often the bound is \goinl{Any}, but in this case, the bound on
\goinl{a} is itself \goinl{Eq(a)}.  For instance, we will see
that \goinl{Bool} implements \goinl{Eq(Bool)}.  But
\goinl{Any} does not implement \goinl{Eq(Any)}, so the latter
is not a valid type.  The situation where a type parameter appears in
its own bound is known as \emph{F-bounded polymorphism}
\citep{Canning-et-al-1989}, and a similar idiom is common in Java with
generics \citep{Bracha-et-al-1998,Naftalin-and-Wadler-2006}.

In interface \goinl{Bool}, the signature for negation is
unchanged. The signature for equality is now
\goinl{Equal(that Bool) Bool}, where the argument type has now been
refined from \goinl{Eq} to \goinl{Bool}. The signature for
conditionals is now
\begin{golst}
    Cond(type a Any)(br Branches(a)) a
\end{golst}
where method \goinl{Cond} now accepts a type parameter, with bound
\goinl{Any}. Its argument type has been refined from type
\goinl{Branches} to type \goinl{Branches(a)}, and its result has been
refined from type \goinl{Any} to type \goinl{a}.  Using interface
embedding, we could equivalently write the interface as
\begin{golst}
    type Bool interface {
      Not() Bool
      Eq(Bool)
      Cond(type a Any)(br Branches(a)) a
    }
\end{golst}
which includes all members of the \goinl{Eq(Bool)} interface
in the \goinl{Bool} interface.

The interface for branches now also takes a type parameter,
\goinl{Branch(type a Any)}, and the signatures of its methods are now
\goinl{IfTT() a} and \goinl{IfFF() a}, where the result types have
been refined from \goinl{Any} to \goinl{a}.

The two structure types, \goinl{TT} and \goinl{FF} are just as before,
as is the method for logical negation. The method for equality are as
before, save that the argument type has changed from
\goinl{Eq} to \goinl{Bool}. The type assertions previously
required to convert the argument to a boolean are no longer
needed. Indeed, typing is now strong enough to assure a panic never
occurs when evaluating equality.

Finally, the method for conditionals is as before, save for
the refinement to its signature. Whereas in the previous example
we wrote
\begin{golst}
    return x.Equal(xs.head).Cond(containsBr{x,xs}).(Bool)
\end{golst}
now we write
\begin{golst}
    return x.Equal(xs.head).Cond(Bool)(containsBr{x,xs})
\end{golst}
Method \goinl{Cond} now takes a type parameter \goinl{Bool} specifying
its result type, so the type assertion \goinl{.(Bool)} can be removed.

\subsection{Lists in the style of The Expression Problem}
\label{sec:expression-list}

\begin{Figure}
\begin{ruled}
\begin{golsttwo}
// Nil, Cons
type Nil(type a Any, c Any) struct {}
type Cons(type a Any, c Any) struct {
  head a
  tail c
}
// Map on Nil, Cons
type Mapper(type a Any) interface {
  Map(type b Any, d Any)(f Function(a,b)) d
}
func (xs Nil(type a Any, c Mapper(a)))
    Map(type b Any, d Any)(f Function(a,b)) d {
  return Nil(b,d){}
}
func (xs Cons(type a Any, c Mapper(a)))
    Map(type b Any, d Any)(f Function(a,b)) d {
  return Cons(b,d)
    {f.apply(xs.head), xs.tail.Map(b,d)(f)}
}
// Equality on Nil, Cons
type Eq(type a Eq(a)) interface {
  Equal(that a) bool
}
func (xs Nil(type a Eq(a), c Eq(c))) Equal(ys c) bool {
  ys, ok := ys.(Nil(a,c))
  return ok
}
func (xs Cons(type a Eq(a), c Eq(c))) Equal(ys c) bool {
  ys, ok := ys.(Cons(a,c))
  return ok && xs.head.Equal(ys.head)
            && xs.tail.Equal(ys.tail)
}
// tie it all together
type List(type a Eq(a)) interface {
  Mapper(a)
  Eq(List(a))
}
\end{golsttwo}
\end{ruled}
\caption{Lists in the style of the expression problem}
\label{fig:expression-list}
\end{Figure}

We now present an alternative
design that permits a list with elements of any type when \goinl{Map}
is the only operation applied to the list, but requires elements to
support equality when \goinl{Equal} is applied.

Our alternative solution for lists appears in
Figure~\ref{fig:expression-list}. Structures \goinl{Nil} and
\goinl{Cons} now take not one but two type parameters, both bounded by
interface \goinl{Any}.  As before, parameter \goinl{a} is the type of
the elements of the list (the \goinl{head} field of \goinl{Cons}),
while the new parameter \goinl{c} is the type of the lists themselves
(the \goinl{tail} field of \goinl{Cons}).  The type parameters only
appear in the definition of \goinl{Cons}, but we also add them to
\goinl{Nil} because we may need to refer to them in the signatures of
operations. (As it happens, \goinl{c} doesn't appear in the signatures
here, but it would be needed if we wanted to define, for instance, a
method to append two lists.)

To avoid pollution, for each operation, \goinl{Map} and
\goinl{Equal} we define a corresponding interface to specify that
operation, \goinl{Mapper} and \goinl{Eq}.

Method \goinl{Map} now takes not one but two type
parameters, both bounded by interface \goinl{Any}.  As before,
parameter \goinl{b} is the type of the elements of the result list,
while the new parameter \goinl{d} is the type of the result list
itself.  When defining \goinl{Map} on \goinl{Cons}, the receiver's
first parameter \goinl{a} is bounded by \goinl{Any}
and its second type parameter \goinl{c} is bounded by \goinl{Mapper(a)},
allowing \goinl{Map} to be recursively invoked on the tail.

When defining \goinl{Equal} on \goinl{Cons}, the receiver's first
parameter \goinl{a} is bounded by \goinl{Eq(a)}, allowing list
elements to be tested for equality, and its second type parameter
\goinl{c} is bounded by \goinl{Eq(c)}, allowing \goinl{Equal} to be
recursively invoked on the tail.

The bodies of the methods for equality use a second form of type
assertion which returns both a value of the asserted type (\goinl{xs},
rebound) and a boolean saying whether the assertion succeeded
(\goinl{ok}). Unlike the other form of type assertions, these can
never panic, and so are not an issue with regard to static type
checking. They are not included in our formalisations of FG and FGG,
but would be easy to add.

As in the expression problem, for both \goinl{Map} and \goinl{Equal}
it is crucial that bounds on the type receiver are covariant.  Since
the bounds on the parameters in \goinl{Nil} and \goinl{Cons} are
\goinl{Any}, it is fine for the receivers to use tighter
bounds, such as \goinl{Mapper} or \goinl{Eq}.

A last step shows how to tie it all together. We define
an interface \goinl{List(a)} which embeds interfaces \goinl{Mapper(a)}
and \goinl{Eq(List(a))}.  At this point, the pollution occurs,
and we bound type parameter \goinl{a} by interface \goinl{Eq(a)},
since we must be able to test list elements for equality.

Not all applications will require this level of flexibility, or desire
the associated complexity. But it is good that this design pattern is
supported, and it is easy to see it may be valuable in some
situations.

 \newpage
\newpage

\section{FG Type Soundness}
\label{app:fgtsound}
This section develops FG type soundness in the form of Type
Preservation by evaluation (Appendix~\ref{app:fg_tpres} and Progress
(Appendix~\ref{app:fg_prog}).

The various substitution properties are often stated more generally in
this appendix for convenience (i.e., those in the main matter are
special cases of the results developed here). We make use of
a reduction relation that replaces the contextual rule
$\tyrulename{r-context}$ with the equivalent set of rules that
identify each possible reduction explicitly, implementing a
left-to-right call-by-value semantics. Uses of these congruence rules
are labelled with the prefix {\sc rc}. The development of progress
makes use of an inductively defined $e\PANIC$ predicate, which holds
iff expression $e$ causes a runtime panic (i.e., contains an invalid
type assertion). This predicate is equivalent to the definition of {\it
  panics} from the main matter.

\subsection{Type Preservation}
\label{app:fg_tpres}

\begin{lemma}[Weakening]\label{lem:weak}
   If $\Gamma \vdash e : t$ then,
   for all $\Gamma' \supseteq \Gamma$, $\Gamma' \vdash e : t$.
 \end{lemma}
 \begin{proof}
Straightforward induction on $\Gamma \vdash e : t$.
\end{proof}

\begin{lemma}\label{lem:methstruct}
Let $t \ok$. If $t \imp t_S$ then $t = t_S$.
\end{lemma}
\begin{proof}
By definition of $\imp$.
\end{proof}

\begin{lemma}\label{lem:methsub}
If $mM \in \methods(t)$ then, for all $t' \imp t$, $mM \in \methods(t')$.
\end{lemma}
\begin{proof}
 ~
  \begin{description}
    \item
      \begin{tabbing}
        $\methods(t)\subseteq \methods(t')$ \` by definition
        of $\imp$\\
        $mM \in \methods(t')$ \` by $\subseteq$
      \end{tabbing}
    \end{description}
  \end{proof}

  \begin{lemma}
The $\imp$ relation is reflexive and transitive.
\end{lemma}
\begin{proof}
Straightforward from the definition of $\imp$.
\end{proof}

 \begin{restatable}[Substitution]{lemma}{lemsubst}\label{lem:subst}
  If $\emptyset \vdash \ov{v} : \ov{t}$ and
  $\Gamma , \ov{x : t'} \vdash e' : t_b$, where
  $\ov{t} \imp \ov{t'}$, then
  $\Gamma \vdash e'[\ov{x} \by \ov{v}] : t_c$, for some
  $t_c \imp t_b$.
\end{restatable}
\begin{proof}
By induction on the derivation of  $\Gamma , \ov{x : t'} \vdash e' :
t_b$.

\begin{description}
\item[Case:] Rule {\sc t-var}

  \begin{tabbing}
  $e' = y$ and $t_b = \Gamma(y)$ \\
  {\bf Sub}\={\bf case:} $y \not\in \ov{x}$ \\
  \>$y[\ov{x}\by \ov{v}] = y$ \` by definition\\
  \>$\Gamma \vdash y : t_b$ \` by {\sc t-var}\\
  \>$t_b \imp t_b$ \` by reflexivity\\
  {\bf Subcase:} $y \in \ov{x}$ \\
  \>$x[\ov{x}\by \ov{v}] = v$ and $\Gamma(x) = t'$ \\
 \>$\Gamma \vdash v : t$ and $t \imp t'$ \`
 assumption and Lemma~\ref{lem:weak}
\end{tabbing}

\item[Case:] Rule {\sc t-literal}
   \begin{tabbing}
     $e' =\STRUCTte{\tttS}{\ov{\eee_f}}$,$\ttt_b = \tttS$ and
     $\eee'\SUBS{\ov{\vvv}}{\ov{\xxx}} = \STRUCTte{\tttS}{ \eee_f\SUBS{\ov{\vvv}}{\ov{\xxx}}}$ \\
  $\JUDGEEXPR{\Ga,\ov{\xxx : \ttt'}}{\ov{\eee_f}}{\ov{\ttt_a}}$ \` by inversion\\
  $\FIELDSOFtD{\tttS} = \ov{\FDECLft{\fff}{\ttt_f}}$ \` by inversion\\
  $\impls{\ov{\ttt_a}}{\ov{\ttt_f}}$ \` by inversion\\
  $\JUDGEEXPR{\Ga}{\ov{\eee_f\SUBS{\ov{\vvv}}{\ov{\xxx}}}}{\ov{\ttt_c}}$, for
  some $\IMPLS{\ov{\ttt_c}}{\ov{\ttt_a}}$ \` by i.h.\\
  $\IMPLS{\ov{\ttt_c}}{\ov{\ttt_f}}$ \` by transitivity\\
  $\JUDGEEXPR{\Ga}{\STRUCTte{\tttS}{ \eee_f\SUBS{\ov{\vvv}}{\ov{\xxx}}}
  }{\tttS}$ \` by $\tyrulename{t-literal}$\\
 $\IMPLS{\tttS}{\tttS}$ \` by reflexivity
 \end{tabbing}

\item[Case:] Rule $\tyrulename{t-field}$

\begin{tabbing}
  $\eee' = \SELef{\eee_s}{\fff}$, $\ttt_b = \ttt_f$ and
$\eee'\SUBS{\ov{\vvv}}{\ov{\xxx}} = \SELef{\eee_s\SUBS{\ov{\vvv}}{\ov{\xxx}}}{\fff}$\\

$\JUDGEEXPR{\Ga, \ov{\xxx : \ttt'}}{\eee_s}{\tttS}$ \` by inversion\\
$\FDECLft{\fff}{\ttt_f} \in \FIELDSOFtD{\tttS}$ \` by inversion\\

$\JUDGEEXPR{\Ga}{\eee_s\SUBS{\ov{\vvv}}{\ov{\xxx}}}{\ttt_b}$, for some
$\IMPLS{\ttt_b}{\tttS}$ \` by i.h.\\
$\ttt_b = \tttS$ \` by Lemma~\ref{lem:methstruct}\\
$\JUDGEEXPR{\Ga}{\eee_s\SUBS{\ov{\vvv}}{\ov{\xxx}}}{\tttS}$ \` substituting for
equals\\
$\JUDGEEXPR{\Ga}{\SELef{\eee_s}{\fff}}{\ttt_f}$ \` by
$\tyrulename{t-field}$\\
$\IMPLS{\ttt_f}{\ttt_f}$ \` by reflexivity
  \end{tabbing}

  \item[Case:] Rule $\tyrulename{t-call}$

  \begin{tabbing}
    $\eee' = \MCALLeme{\eee_c}{\mmm}{\ov{\eee_a}}$, $\ttt_b = \ttt_r$ and
    $\eee'\SUBS{\ov{\vvv}}{\ov{\xxx}} =
    \MCALLeme{\eee_c\SUBS{\ov{\vvv}}{\ov{\xxx}}}{\mmm}{\ov{\eee_a\SUBS{\ov{\vvv}}{\ov{\xxx}}}}$\\

    $\JUDGEEXPR{\Ga, \ov{\xxx : \ttt'}}{\eee_c}{\ttt_c}$ \` by inversion\\
    $\mmm\MSIGpt{\ov{\pDECLxt{\xxx}{\ttt_p}}}{\ttt_r} \in
  \methodsof{\ttt_c}$ \` by inversion\\
    $\JUDGEEXPR{\Ga, \ov{\xxx : \ttt'}}{\ov{\eee_a}}{\ov{\ttt_a}}$ \` by
    inversion\\
    $\IMPLS{\ov{\ttt_a}}{\ov{\ttt_p}}$ \` by inversion\\

    $\JUDGEEXPR{\Ga}{\eee_c\SUBS{\ov{\vvv}}{\ov{\xxx}}}{\ttt_d}$, for some
    $\IMPLS{\ttt_d}{\ttt_c}$ \` by i.h.\\
    $\mmm\MSIGpt{\ov{\pDECLxt{\xxx}{\ttt_p}}}{\ttt_r} \in
  \methodsof{\ttt_d}$
    \` by Lemma~\ref{lem:methsub}\\
    $\JUDGEEXPR{\Ga}{\ov{\eee_a\SUBS{\ov{\vvv}}{\ov{\xxx}}}}{\ov{\ttt_e}}$, for
    some
    $\IMPLS{\ov{\ttt_e}}{\ov{\ttt_a}}$ \` by i.h.\\
    $\IMPLS{\ov{\ttt_e}}{\ov{\ttt_p}}$ \` by transitivity\\
    $\JUDGEEXPR{\Ga}{\MCALLeme{\eee_c\SUBS{\ov{\vvv}}{\ov{\xxx}}}{\mmm}{\ov{\eee_a\SUBS{\ov{\vvv}}{\ov{\xxx}}}}}{\ttt_r}$
    \` by $\tyrulename{t-call}$\\
    $\IMPLS{\ttt_r}{\ttt_r}$ \` by reflexivity
  \end{tabbing}

\item[Case:] Rule \tyrulename{t-assert$_\III$}

  \begin{tabbing}
    $\eee' = \ASSRTet{\eee_a}{\tttI}$, $\ttt_b = \tttI$ and
    $\eee'\SUBS{\ov{\vvv}}{\ov{\xxx}} =
    \ASSRTet{\eee_a\SUBS{\ov{\vvv}}{\ov{\xxx}}}{\tttI}$ \\
    $\JUDGEWF{\tttI}$ \` by inversion\\
    $\JUDGEEXPR{\Ga, \ov{\xxx : \ttt'}}{\eee_a}{\tttI'}$ \` by inversion\\
    $\JUDGEEXPR{\Ga}{\eee_a\SUBS{\ov{\vvv}}{\ov{\xxx}}}{\ttt_a}$, for some
    $\IMPLS{\ttt_a}{\tttI'}$ \` by i.h.\\
    {\bf Sub}\={\bf case:} $\ttt_a = \tttS$, for some struct. type $\tttS$\\
    \>$\ASSRTet{\eee_a\SUBS{\ov{\vvv}}{\ov{\xxx}}}{\tttI}$ \` by
    \tyrulename{t-stupid}\\
    \>$\IMPLS{\tttI}{\tttI}$ \` by reflexivity\\
    {\bf Subcase:} $\ttt_a = \tttI''$, for some interface type
    $\tttI''$\\
    \>$\JUDGEEXPR{\Ga}{\ASSRTet{\eee_a\SUBS{\ov{\vvv}}{\ov{\xxx}}}{\tttI}}{\tttI}$ \` by
    \tyrulename{t-assert$_I$}\\
   \>$\IMPLS{\tttI}{\tttI}$ \` by reflexivity
  \end{tabbing}

\item[Case:] Rule \tyrulename{t-assert$_\SSS$}

  \begin{tabbing}
    $\eee' = \ASSRTet{\eee_a}{\tttS}$, $\ttt_b = \tttS$ and
    $\eee'\SUBS{\ov{\vvv}}{\ov{\xxx}} = \ASSRTet{\eee_a\SUBS{\ov{\vvv}}{\ov{\xxx}}
    }{\tttS}$\\
    $\JUDGEWF{\tttS}$ \` by inversion\\
    $\JUDGEEXPR{\Ga, \ov{\xxx : \ttt'}}{\eee_a}{\tttI}$ \` by inversion\\
    $\IMPLS{\ttt_\SSS}{\ttt_\III}$ \` by inversion\\
    $\JUDGEEXPR{\Ga}{\eee_a\SUBS{\ov{\vvv}}{\ov{\xxx}}}{\ttt_a}$, for some
    $\IMPLS{\ttt_a}{\tttI}$ \` by i.h.\\
    {\bf Sub}\={\bf case:} $\ttt_a = \tttS'$, for some struct. type
    $\tttS'$\\
      \>$\JUDGEEXPR{\Ga}{\ASSRTet{\eee_a\SUBS{\ov{\vvv}}{\ov{\xxx}}}{\tttS}}{\tttS}$ \` by
    \tyrulename{t-stupid}\\
    {\bf Subcase:} $\ttt_a = \tttI'$, for some interface type
    $\tttI'$\\
    $\tttI' = \tttI$ \` by value restriction\\

    $\JUDGEEXPR{\Ga}{\ASSRTet{\eee_a}{\tttS}}{\tttS}$ \` by
    \tyrulename{t-assert$_\SSS$}\\
    $\IMPLS{\tttS}{\tttS}$ \` by reflexivity\\

  \end{tabbing}

\item[Case:] Rule \tyrulename{t-stupid}
  \begin{tabbing}
    $\eee' = \ASSRTet{\eee_a}{\ttt}$, $\ttt_b = \ttt$ and
    $\eee'\SUBS{\ov{\vvv}}{\ov{\xxx}} = \ASSRTet{\eee_a\SUBS{\ov{\vvv}}{\ov{\xxx}}}{\ttt}$
    \\
    $\JUDGEWF{\ttt}$ \` by inversion\\
    $\JUDGEEXPR{\Ga, \ov{\xxx : \ttt'}}{\eee_a}{\tttS}$ \` by inversion\\
    $\JUDGEEXPR{\Ga }{\eee_a\SUBS{\ov{\vvv}}{\ov{\xxx}}}{\ttt_c}$, for some
    $\IMPLS{\ttt_c}{\tttS}$ \`by i.h.\\
    $\ttt_c = \tttS$ \` by Lemma~\ref{lem:methstruct}\\
    $\JUDGEEXPR{\Ga}{\ASSRTet{\eee_a\SUBS{\ov{\vvv}}{\ov{\xxx}}}{\ttt}}{\ttt}$
    \` by \tyrulename{t-stupid}\\
    $\IMPLS{\ttt}{\ttt}$ \` by reflexivity
   \end{tabbing}

\end{description}

\end{proof}

  \begin{restatable}[Type Preservation]{theorem}{thmpres}\label{thm:pres}
If $\JUDGEEXPR{\emptyset}{\eee}{\ttt}$ and $\eee \REDUCE{} \eee'$ then
  $\JUDGEEXPR{\emptyset}{\eee}{\ttt'}$, for some $\IMPLS{\ttt'}{\ttt}$.
\end{restatable}

\begin{proof}
  By induction on the derivation of $\eee \REDUCE{} \eee'$.

  \begin{description}
\item[Case:] Rule \RULENAME{r-field}

    \begin{tabbing}
      $\eee = \SELef{\STRUCTte{\tttS}{\ov{\eee_f}}}{\fff_i}$,
      $\eee' = \eee_{f_{i}}$, $\ttt = \ttt_{f_i}$\\
      $\FIELDSOFtD{\tttS} \EQ{} \ov{\FDECLft{\fff}{\ttt}}$ \` by
      inversion on $\REDUCE{}$\\
      $\JUDGEEXPR{\emptyset}{\STRUCTte{\tttS}{\ov{\eee_f}} }{\tttS}$ \` by
      inversion on typing\\
      $\FDECLft{\fff}{\ttt_f} \in \FIELDSOFtD{\tttS}$ \` by inversion
      on typing\\
      $\JUDGEEXPR{\emptyset}{\ov{\eee_f}}{\ov{\ttt_a}}$ \` by
      inversion on typing\\
      $\IMPLS{\ov{\ttt_a}}{\ov{\ttt_f}}$ \` by inversion on typing\\
      $\FIELDSOFtD{\tttS} = \ov{\FDECLft{\fff}{\ttt_f}}$ \` by
      inversion on typing\\
      $\JUDGEEXPR{\emptyset}{\eee_{f_i}}{\ttt_{a_i}}$ and
      $\IMPLS{\ttt_{a_i}}{\ttt_{f_i}}$ \\
    \end{tabbing}

\item[Case:] Rule $\RULENAME{r-call}$

    \begin{tabbing}
      $\eee =
      \MCALLeme{\STRUCTte{\ttt_S}{\ov{\eee}}}{\mmm}{\ov{\eee_a}}$,
      $\eee' =
      \eee_0\SUBS{\STRUCTte{\ttt_S}{\ov{\eee}}}{\xxx_0}\SUBS{\ov{\eee_a}}{\ov{\xxx}}$,
      $\ttt = \ttt_r$ \\
      $\MBODYmt{\mmm}{\tttS}
      \EQ{}
      \MBODYxxe{\xxx_0:\ttt_S}{ \ov{\xxx : \ttt_p} }{\eee_0}$
\` by inversion on $\REDUCE{}$\\

    $\JUDGEEXPR{\emptyset}{\STRUCTte{\ttt_S}{\ov{\eee}}}{\ttt_S}$ \` by inversion on
    typing\\
    $\mmm\MSIGpt{\ov{\pDECLxt{\xxx}{\ttt_p}}}{\ttt_r} \in
  \methodsof{\ttt_S}$ \` by inversion on typing\\
    $\JUDGEEXPR{\emptyset}{\ov{\eee_a}}{\ov{\ttt_a}}$ \` by
    inversion on typing\\
    $\IMPLS{\ov{\ttt_a}}{\ov{\ttt_p}}$ \` by inversion on typing\\
$\JUDGEEXPR{\xxx : \tttS, \ov{\xxx : \ttt_p}}{\eee_0}{\ttt_1}$ with
$\IMPLS{\ttt_1}{\ttt_r}$ \` by inversion on method typing\\
$\JUDGEEXPR{}{\eee_0\SUBS{\STRUCTte{\ttt_S}{\ov{\eee}}}{\xxx_0}\SUBS{\ov{\eee_a}}{\ov{\xxx}}}{\ttt_2}$,
for some $\IMPLS{\ttt_2}{\ttt_1}$ \` by Lemma~\ref{lem:subst}\\
$\IMPLS{\ttt_2}{\ttt_r}$ \` by transitivity\\
\end{tabbing}

\item[Case:] Rule \RULENAME{r-assert}

   \begin{tabbing}
     $\eee = \ASSRTet{\STRUCTte{\tttS}{\ov{\eee}}}{\ttt'}$,
     $\eee' = \STRUCTte{\tttS}{\ov{\eee}}$, $\ttt = \ttt'$ \\
     $\IMPLS{\tttS}{\ttt'}$ \` by inversion on $\REDUCE{}$\\
     $\JUDGEEXPR{\emptyset}{\STRUCTte{\tttS}{\ov{\eee}}}{\ttt'}$ \` by
     inversion on typing\\
   \end{tabbing}
\item[Case:] Rule \RULENAME{rc-recv}

  \begin{tabbing}
    $\eee = \MCALLeme{\eee_0}{\mmm}{\ov{\eee_a}}$,
    $\eee' = \MCALLeme{\eee_0'}{\mmm}{\ov{\eee_a}}$,
    $\ttt = \ttt_r$ \\
    $\eee_0 \REDUCE{} \eee_0'$ \` by inversion on $\REDUCE{}$\\
    $\JUDGEEXPR{\emptyset}{\eee_0}{\ttt_0}$   \` by inversion on typing\\
$\mmm\MSIGpt{\ov{\pDECLxt{\xxx}{\ttt_p}}}{\ttt_r} \in
  \methodsof{\ttt_0}$ \` by inversion on typing\\
    $\JUDGEEXPR{\emptyset}{\ov{\eee_a}}{\ov{\ttt_a}}$ \` by
    inversion on typing\\
    $\IMPLS{\ov{\ttt_a}}{\ov{\ttt_p}}$ \` by inversion on typing\\
$\JUDGEEXPR{\emptyset}{\eee_0}{\ttt_1}$, for some $\IMPLS{\ttt_1}{\ttt_0}$ \`
by i.h.\\
$\mmm\MSIGpt{\ov{\pDECLxt{\xxx}{\ttt_p}}}{\ttt_r} \in
  \methodsof{\ttt_1}$ \` by Lemma~\ref{lem:methsub}\\
$\JUDGEEXPR{\emptyset}{\MCALLeme{\eee_0'}{\mmm}{\ov{\eee_a}} }{\ttt_r}$ \`
by $\tyrulename{t-call}$\\
$\IMPLS{\ttt_r}{\ttt_r}$ \` by reflexivity
   \end{tabbing}
    
\item[Case:] Rule \RULENAME{rc-arg}

  \begin{tabbing}
    $\eee = \MCALLeme{\eee_0}{\mmm}{\ov{\eee_a}}$,
    $\eee' = \MCALLeme{\eee_0}{\mmm}{\ov{\eee'_a}}$,
    $\ttt = \ttt_r$ \\
    $\eee_{a_i} \REDUCE{} \eee'_{a_i} $ \` by inversion on $\REDUCE{}$\\
    $\JUDGEEXPR{\emptyset}{\eee_0}{\ttt_0}$   \` by inversion on typing\\
$\mmm\MSIGpt{\ov{\pDECLxt{\xxx}{\ttt_p}}}{\ttt_r} \in
  \methodsof{\ttt_0}$ \` by inversion on typing\\
    $\JUDGEEXPR{\emptyset}{\ov{\eee_a}}{\ov{\ttt_a}}$ \` by
    inversion on typing\\
    $\IMPLS{\ov{\ttt_a}}{\ov{\ttt_p}}$ \` by inversion on typing\\
$\JUDGEEXPR{\emptyset}{\eee'_{a_i}}{\ttt_1}$, for some $\IMPLS{\ttt_1}{\ttt_{a_i}}$ \`
by i.h.\\
$\IMPLS{\ttt_1}{\ttt_{p_i}}$ \` by transitivity\\
$\JUDGEEXPR{\emptyset}{\MCALLeme{\eee_0}{\mmm}{\ov{\eee'_a}} }{\ttt_r}$ \`
by $\tyrulename{t-call}$\\
$\IMPLS{\ttt_r}{\ttt_r}$ \` by reflexivity
\end{tabbing}

\item[Case:] Rule \RULENAME{rc-assert}

  \begin{tabbing}
    $\eee = \ASSRTet{\eee_0}{\ttt'}$,
    $\eee' = \ASSRTet{\eee'_0}{\ttt'}$, $\ttt = \ttt'$\\
    $\eee_0 \REDUCE{} \eee'_0$ \` by inversion on $\REDUCE{}$\\
    $\JUDGEWF{\ttt'}$ \` by inversion on typing\\
    {\bf Sub}\={\bf case:} $\JUDGEEXPR{\Ga}{\eee_0}{\tttI}$ \` by inversion on typing\\
    \> $\JUDGEEXPR{\Ga}{\eee'_0}{\ttt_0}$, for some
    $\IMPLS{\ttt_0}{\tttI}$ \` by i.h.\\
    \> {\bf Subsubcase:} $\ttt_0$ is an interface type and $\ttt'$ is
    a struct type\\
    \> $\ttt_0 = \tttI$ and $\IMPLS{\ttt'}{\tttI}$ \` by value restriction\\
    \> $\JUDGEEXPR{\Ga}{\ASSRTet{\eee'_0}{\ttt'}}{\ttt'}$ \` by
    \tyrulename{t-assert$_\SSS$}\\
    \> {\bf Subsubcase:} $\ttt_0$ is an interface type and $\ttt'$ is
    an interface type\\
     \> $\JUDGEEXPR{\Ga}{\ASSRTet{\eee'_0}{\ttt'}}{\ttt'}$ \` by
     \tyrulename{t-assert$_{\III}$} \\
     \> {\bf Subsubcase:} $\ttt_0$ is a struct. type\\
     \> $\JUDGEEXPR{\Ga}{\ASSRTet{\eee'_0}{\ttt'}}{\ttt'}$
     \` by    \tyrulename{t-stupid}\\
     
\end{tabbing}

  \end{description}
\end{proof}
  
\subsection{Progress}
\label{app:fg_prog}

The inductive definition of $\PANICe{\eee}$ is given
below:
\[
  \begin{array}{c}
    \inferrule[\tyrulename{p-struct}]
    {\PANICe{e_i}}
    {\PANICe{\STRUCTte{\tttS}{\ov{\eee}}}}
    \quad
    \inferrule[\tyrulename{p-sel}]
    {\PANICe{\eee}}
    {\PANICe{\SELef{\eee}{\fff}}}
    \quad
    \inferrule[\tyrulename{p-call-arg}]
    {\PANICe{\eee_{{i}}}} 
    {\PANICe{\MCALLeme{\vvv}{\mmm}{\ov{\eee}}}}
    \quad
    \inferrule[\tyrulename{p-call-body}]
    {\PANICe{\eee}}
    {\PANICe{\MCALLeme{\eee}{\mmm}{\ov{\eee}}}}
    \\[1em]
    \inferrule[\tyrulename{p-assert}]
    {\PANICe{\eee}}
    {\PANICe{\ASSRTet{\eee}{\ttt}}}
    \quad
    \inferrule[\tyrulename{p-cast}]
    {\NIMPLS{\tttS}{\ttt}}
    {\PANICe{\ASSRTet{\STRUCTte{\tttS}{\ov{\vvv}}}{\ttt}}}
    \end{array}
  \]

  \begin{lemma}[Canonical Forms]\label{lem:canforms}
If $\eee$ is a value and $\JUDGEEXPR{\emptyset}{\eee}{\ttt}$ then
$\ttt = \tttS$, for some $\tttS$ and $\eee =
\STRUCTte{\tttS}{\ov{\vvv}}$, for some $\ov{\vvv}$.
\end{lemma}
\begin{proof}
Straightforward induction on typing.
\end{proof}

\begin{restatable}[Progress]{theorem}{thmprogpanic}\label{thm:progpanic}
If $\JUDGEEXPR{\emptyset}{\eee}{\ttt}$ then either $\eee$ is a value,
$\PANICe{\eee}$, or $\eee \REDUCE \eee'$.
\end{restatable}

\begin{proof}

  By induction on typing. We show two illustrative but standard cases 
  and then all cases pertaining to type assertions.

  \begin{description}

 \item[Case:] \tyrulename{t-literal} 
    \begin{tabbing}
      $\JUDGEEXPR{\emptyset}{\STRUCTte{\tttS}{\ov{\eee_f}}}{\tttS}$ \`
      this case\\
      $\JUDGEEXPR{\emptyset}{\ov{\eee_f}}{\ov{\ttt}}$ \` by inversion\\
       $\FIELDSOFtD{\ttt_S} = \ov{\FDECLft{\fff}{\ttt_f}}$ \` by
       inversion\\
       $\impls{\ov{\ttt}}{\ov{\ttt_f}}$ \` by inversion\\
       $\ov{\eee_{f_i}}$ is a value, $\PANICe{\eee_{f_i}}$ or
       $\eee_{f_i} \REDUCE  \eee_{f_i}'$ \` by i.h., for all such
       $i$\\
       {\bf Subcase:} $\PANICe{\eee_{f_i}}$, for some $i$\\
       $\PANICe{\STRUCTte{\tttS}{\ov{\eee_f}}}$ \` by
       $\tyrulename{p-struct}$\\
       {\bf Subcase:} $\eee_{f_i}$ is a value, for all $i$\\
       $\STRUCTte{\tttS}{\ov{\eee_f}}$ is a value \` by definition\\
       {\bf Subcase:} $\eee_f$ = $\ov{\vvv}\cdot \eee_i\cdot
       \ov{\eee_{j}}$, where $\eee_i \REDUCE \eee_i'$\\
       $\STRUCTte{\tttS}{\ov{\eee_f}} \REDUCE
       \STRUCTte{\tttS}{\ov{\vvv}\cdot \eee_i'\cdot
         \ov{\eee_{j}}}$ \` by $\RULENAME{rc-literal}$
     \end{tabbing}

     \item[Case:] \tyrulename{t-field}

     \begin{tabbing}
       $\JUDGEEXPR{\emptyset}{\SELef{\eee}{\fff_i}}{\ttt}$ \` this case\\
       $\JUDGEEXPR{\emptyset}{\eee}{\tttS}$ \` by inversion\\
       $\FDECLft{\fff_i}{\ttt} \in \FIELDSOFtD{\tttS}$ \` by inversion\\
       $\eee$ is a value, $\PANICe{\eee}$ or $\eee \REDUCE \eee'$ \`
       by i.h.\\
       {\bf Subcase:} $\eee$ is a value\\
       $\eee = \STRUCTte{\tttS}{\ov{\vvv}}$ \` by
       Lemma~\ref{lem:canforms}\\
       $\SELef{\STRUCTte{\tttS}{\ov{\vvv}}}{\fff_i} \REDUCE
       \vvv_i$ \` by $\RULENAME{r-field}$\\
       {\bf Subcase:} $\PANICe{\eee}$ \\
       $\PANICe{\SELef{\eee}{\fff_i}}$ \` by \tyrulename{p-sel}\\
       {\bf Subcase:} $\eee \REDUCE \eee'$\\
       $\SELef{\eee}{\fff_i}
     \REDUCE
     \SELef{\eee'}{\fff_i}$ \` by $\RULENAME{rc-field}$
   \end{tabbing}

 \item[Case:] \tyrulename{t-call}

   \begin{tabbing}
$\JUDGEEXPR{\emptyset}{\MCALLeme{\eee}{\mmm}{\ov{\eee_a}}}{\ttt_r}$ \` this
case\\
$ \JUDGEEXPR{\emptyset}{\eee}{\ttt}$ \` by inversion\\
$\METHODmpt{\mmm}{\ov{\pDECLxt{\xxx}{\ttt_p}}}{\ttt_r} \in
\methodsof{\ttt}$ \` by inversion\\
$\JUDGEEXPR{\emptyset}{\ov{\eee_a}}{\ov{\ttt_a}}$ \` by inversion\\
$\impls{\ov{\ttt_a}}{\ov{\ttt_p}}$ \` by inversion\\
$\eee$ is a value, $\PANICe{\eee}$ or $\eee \REDUCE \eee'$ \` by
i.h.\\
{\bf Subcase:} $\PANICe{\eee}$\\
$\PANICe{\MCALLeme{\eee}{\mmm}{\ov{\eee_a}}}$ \` by
$\tyrulename{p-call-body}$\\
{\bf Subcase:} $\eee \REDUCE \eee'$\\
$\MCALLeme{\eee}{\mmm}{\ov{\eee_a}}
     \REDUCE
     \MCALLeme{\eee'}{\mmm}{\ov{\eee_a}}$ \` by $\RULENAME{rc-recv}$\\
     {\bf Subcase:} $\eee$ is a value\\
       $\eee = \STRUCTte{\tttS}{\ov{\vvv}}$ \` by
       Lemma~\ref{lem:canforms}, for some $\tttS$ and $\ov{\vvv}$\\
     $\ov{\eee_{a_i}}$ is a value, $\PANICe{\eee_{a_i}}$ or
       $\eee_{a_i} \REDUCE  \eee_{f_i}'$ \` by i.h., for all such
       $i$\\
       {\bf Subsubcase:} $\ov{\eee_{a_i}}$ is a value, for all $i$\\
       $\MBODYmt{\mmm}{\ttt_S}
  \EQ
  \MBODYxxe{\xxx}{\ov{\xxx}}{\eee_b}$, for some $\eee_b$  \`
  since $\METHODmpt{\mmm}{\ov{\pDECLxt{\xxx}{\ttt_p}}}{\ttt_r} \in
\methodsof{\tttS}$ \\
       $\MCALLeme{\eee}{\mmm}{\ov{\eee_a}}
    \REDUCE{}
    \eee_b\SUBS{\eee}{\xxx}\SUBS{\ov{\eee_a}}{\ov{\xxx}}$ \` by
    \RULENAME{r-call}\\
    {\bf Subsubcase:} $\PANICe{\eee_{a_i}}$, for some $i$\\
    $\PANICe{\MCALLeme{\eee}{\mmm}{\ov{\eee_a}}}$ \` by
    \tyrulename{p-call-arg}\\
    {\bf Subsubcase:} $\eee_a$ = $\ov{\vvv'}\cdot \eee_i\cdot
    \ov{\eee_{j}}$, where $\eee_i \REDUCE \eee_i'$\\
    $\MCALLeme{\STRUCTte{\tttS}{\ov{\vvv}}}{\mmm}{\ov{\vvv'}\cdot \eee_i\cdot \ov{\eee_j}}
     \REDUCE
     \MCALLeme{\STRUCTte{\tttS}{\ov{\vvv}}}{\mmm}{\ov{\vvv'}\cdot \eee_i' \cdot \ov{\eee_j}}$
     \` by \RULENAME{rc-arg}
     \end{tabbing}
     
     \item[Case:] \tyrulename{t-assert$_\III$}
\begin{tabbing}
       $\JUDGEEXPR{\emptyset}{\ASSRTet{\eee}{\tttI}}{\tttI}$ \` this
       case\\
       $\JUDGEEXPR{\emptyset}{\eee}{\tttI'}$ \` by inversion\\
       $\eee$ is a value, $\PANICe{\eee}$ or $\eee \REDUCE \eee'$ \`
       by i.h.\\
       {\bf Subcase:} $\eee$ is a value\\
       Impossible, deriving a contradiction with
       Lemma~\ref{lem:canforms}.\\
       {\bf Subcase:}$\PANICe{\eee}$ \\
       $\PANICe{\ASSRTet{\eee}{\tttI}}$\` by
       $\tyrulename{p-assert}$\\
       {\bf Subcase:} $\eee \REDUCE \eee'$\\
       $\ASSRTet{\eee}{\tttI} \REDUCE \ASSRTet{\eee'}{\tttI}$ \` by $\RULENAME{rc-assert}$
     \end{tabbing}
   \item[Case:] \tyrulename{t-assert$_\SSS$}   
     \begin{tabbing}
       $\JUDGEEXPR{\emptyset}{\ASSRTet{\eee}{\tttS}}{\tttS}$ \` this case\\
       $\JUDGEEXPR{\emptyset}{\eee}{\tttI}$ \` by inversion\\
       $\impls{\ttt_\SSS}{\ttt_\III}$ \` by inversion\\
       $\eee$ is a value, $\PANICe{\eee}$ or $\eee \REDUCE \eee'$ \`
       by i.h.\\
       {\bf Subcase:} $\eee$ is a value\\
       Impossible, deriving a contradiction with
        Lemma~\ref{lem:canforms}.\\
       {\bf Subcase:}$\PANICe{\eee}$ \\
       $\PANICe{\ASSRTet{\eee}{\tttS}}$\` by
       $\tyrulename{p-assert}$\\
       {\bf Subcase:} $\eee \REDUCE \eee'$\\
       $\ASSRTet{\eee}{\tttS} \REDUCE \ASSRTet{\eee'}{\tttS}$ \` by $\RULENAME{rc-assert}$
       
     \end{tabbing}

      \item[Case:] \tyrulename{t-stupid}

     \begin{tabbing}
       $\JUDGEEXPR{\emptyset}{\ASSRTet{\eee}{\ttt}}{\ttt}$ \` this case\\
       $\JUDGEEXPR{\emptyset}{\eee}{\tttS}$ \` by inversion\\
$\eee$ is a value, $\PANICe{\eee}$ or $\eee \REDUCE \eee'$ \`
       by i.h.\\
       {\bf Subcase:} $\eee$ is a value\\
       $\eee = \STRUCTte{\tttS}{\ov{\vvv}}$ \` by
        Lemma~\ref{lem:canforms} and $\ov{\vvv}$\\
       $\PANICe{\ASSRTet{\eee}{\ttt}}$ \` by
       $\tyrulename{p-cast}$\\
       {\bf Subcase:} $\PANICe{\eee}$\\
       $\PANICe{\ASSRTet{\eee}{\ttt}}$\` by
       $\tyrulename{p-assert}$\\
       {\bf Subcase:} $\eee \REDUCE \eee'$ \\
       $\ASSRTet{\eee}{\ttt} \REDUCE \ASSRTet{\eee'}{\ttt}$ \` by $\RULENAME{rc-assert}$
       \end{tabbing}
    
  \end{description}

 \end{proof}

 \newpage

\section{FGG Type Soundness}
\label{app:fggtsound}
This section develops FGG type soundness in the form of Type
Preservation by evaluation (Appendix~\ref{app:fgg_tpres} and Progress
(Appendix~\ref{app:fgg_prog}).

The various substitution properties are often stated more generally in
this appendix for convenience (i.e., those in the main matter are
special cases of the results developed here). We make use of
a reduction relation that replaces the contextual rule
$\tyrulename{r-context}$ with the equivalent set of rules that
identify each possible reduction explicitly, implementing a
left-to-right call-by-value semantics. Uses of these congruence rules
are labelled with the prefix {\sc rc}. The development of progress
makes use of an inductively defined $e\,\PANIC$ predicate, which holds
iff expression $e$ causes a runtime panic (i.e., contains an invalid
type assertion). This predicate is equivalent to the definition of {\it
  panics} from the main matter.

\subsection{Type Preservation}
\label{app:fgg_tpres}

Below we write $\Phi , \ov{\alpha : \tau_I}$ under the assumption that
$\distinct(\hat\Phi, \ov{\alpha})$ and treat typing contexts and
$\TENV$ and type formals $\Phi, \Psi$ interchangeably, as in the rules
of Section~\ref{sec:fgg}.

\begin{lemma}[Weakening]
  \label{lem:fgg_weak}
  Let $\JUDGEWFG{\TENV}{\ov{\tau_\III}}$ and $\JUDGEWFG{\TENV}{\ov{\tau_0}}$. Then:

  \begin{enumerate}
  \item $\JUDGEWFG{\TENV}{\tau}$ implies $\JUDGEWFG{\TENV , \ov{\alpha
        : \tau_\III}}{\tau}$;
  \item $\JUDGEWFG{\TENV}{\Psi}$ implies
    $\JUDGEWFG{\TENV, \ov{\alpha : \tau_\III}}{\Psi}$
    and $\JUDGEWFG{\TENV}{\Psi,\ov{\alpha : \tau_\III}}$;
  \item $\JUDGEFORMCTXT{\TENV}{\Psi}{\TENV'}$ implies
    $\JUDGEFORMCTXT{\TENV}{\Psi, \ov{\alpha : \tau_\III}}{(\TENV', \ov{\alpha : \tau_\III})}$;
  \item $\IMPLSG{\TENV}{\tau_1}{\tau_2}$ implies
    $\IMPLSG{\TENV , \ov{\alpha : \tau_\III} }{\tau_1}{\tau_2}$;
  \item $\JUDGEEXPRG{\TENV}{\Ga}{\eee}{\tau}$ implies
    $\JUDGEEXPRG{\TENV, \ov{\alpha : \tau_\III}}{\Ga}{\eee}{\tau}$
    and $\JUDGEEXPRG{\TENV}{\Ga , \ov{ x : \tau_0}}{\eee}{\tau}$.
  \end{enumerate}
\end{lemma}
\begin{proof}
Each statement above is proved by a straightforward induction on the
given derivation. 
\end{proof}

\begin{lemma}
  \label{lem:fgg_invsub}
If $\IMPLSG{\TENV}{\tau}{\tau_\SSS}$, for some $\JUDGEWFG{\TENV}{\tau_\SSS}$, then
$\tau = \tau_\SSS$.
\end{lemma}
\begin{proof}
Straightforward.
\end{proof}

\fggwftype*
\begin{proof}
Straightforward induction on typing.
\end{proof}

\begin{lemma}
  \label{lem:fgg_subequiv}
Subtyping in FGG is reflexive and transitive, that is,
(1) if $\JUDGEWFG{\TENV}{\tau}$ then $\IMPLSG{\TENV}{\tau}{\tau}$ and (2) if
$\IMPLSG{\TENV}{\tau_1}{\tau_2}$ and $\IMPLSG{\TENV}{\tau_2}{\tau_3}$
then $\IMPLSG{\TENV}{\tau_1}{\tau_3}$
\end{lemma}
\begin{proof}
Straightforward induction on the definition of $\IMPLOP$.
\end{proof}

\begin{lemma}
\label{lem:fgg_substbound}

      Assume  $\JUDGEWFG{\TENV_0 , \ov{\alpha : \tau_\III} ,
        \TENV_1}{\tau}$ and
      $\IMPLSG{\TENV_0}{\ov{\tau_0}}{\ov{\tau_\III}\SUBS{\ov{\tau_0}}{\ov{\alpha}}}$ with
      $\JUDGEWFG{\TENV_0}{\ov{\tau_0}}$. Let
      $\SUBSTBUILDER{\eta}{\ov{\alpha}}{\ov{\tau_0}}$, then 
      $\IMPLSG
      {\TENV_0 , \TENV_1[\eta]}
      {
\BOUNDTENVt{\TENV_0 , \TENV_1[\eta]}{\tau[\eta]}
      }{
(\BOUNDTENVt{\TENV_0 , \ov{\alpha : \tau_\III}, \TENV_1}{\tau})[\eta]
      }$

  \end{lemma}
\begin{proof}
By case analysis on the structure of $\tau$.
    \begin{description}
 
    \item[Case:] $\tau$ is a non-variable type

      Trivial by reflexivity.

      \item[Case:] $\tau$ is a type variable $\beta$ and $\beta \in
        \TENV_0$

        Immediate by reflexivity.

        \item[Case:]$\tau$ is a type variable $\beta$ and $\beta \in
          \TENV_1$
          \begin{tabbing}
$\BOUNDTENVt{\TENV_0 ,
  \TENV_1[\eta]}{\tau[\eta]} =
\BOUNDTENVt{\TENV_0 , \TENV_1[\eta]}{\beta}$ \` by
definition\\
$= \TENV_1(\beta)[\eta] = (\BOUNDTENVt{\TENV_0 , \ov{\alpha :
  \tau_\III}, \TENV_1}{\beta})[\eta]$ \` by
definition\\
$\IMPLSG{\TENV_0 , \TENV_1[\eta]}
{\TENV_1(\beta)[\eta]
}{\TENV_1(\beta)[\eta] }$ \` by reflexivity
            \end{tabbing}
          \item[Case:]$\tau$ is a type variable $\beta$ and $\beta \in
            \ov{\alpha}$
            \begin{tabbing}
$\BOUNDTENVt{\TENV_0 ,
  \TENV_1[\eta] }{\tau[\eta]} =
\BOUNDTENVt{\TENV_0 , \TENV_1[\eta]}{\tau_0} =
\tau_0$ \` by definition\\
$(\BOUNDTENVt{\TENV_0 , \ov{\alpha : \tau_\III} ,
  \TENV_1}{\alpha})[\eta] =
\tau_\III[\eta]$ \` by definition\\
$\IMPLSG{\TENV_0}{\tau_0}{\tau_\III[\eta]}$ \`
assumption\\
$\IMPLSG{\TENV_0,
  \TENV_1[\eta]}{\tau_0}{\tau_\III[\eta]}$ \`
by Lemma~\ref{lem:fgg_weak}
              \end{tabbing}

      \end{description}
    
  \end{proof}

\begin{lemma}[Type Substitution Preserves Subtyping]
  \label{lem:fgg_tsubstsub}

Assume $\IMPLSG{\TENV_0 , \ov{\alpha : \tau_\III} , \TENV_1}{\tau_1}{\tau_2}$
, $\JUDGEWFG{\TENV_0}{\ov{\tau_0}}$ and
$\IMPLSG{\TENV_0}{(\ov{\tau_0}}{\ov{\tau_\III)}\SUBS{\ov{\tau_0}}{\ov{\alpha}}}$
and let $\SUBSTBUILDER{\eta}{\ov{\alpha}}{\ov{\tau_0}}$. We have that 
$\IMPLSG{\TENV_0 , \TENV_1[\eta]}{\tau_1  [\eta]}{\tau_2 [\eta]}$.

\end{lemma}
 \begin{proof}
By induction on the derivation of $\IMPLSG{\TENV_0 , \ov{\alpha : 
    \tau_\III} , \TENV_1}{\tau_1}{\tau_2}$.

\begin{description}
\item[Case:] $\RULENAME{$\IMPLOP_\III$}$

  \begin{tabbing}
    $\tau_ 1 = \tau$ and $\tau_2 = \tau_\III'$ \` this case\\
        $ \methodsofD{\TENV_0 , \ov{\alpha : 
    \tau_\III} , \TENV_1}{\tau}
        \supseteq 
        \methodsofD{\TENV_0 , \ov{\alpha : 
    \tau_\III} , \TENV_1}{\tau_\III'}$ \` by inversion\\
        $\methodsofD{\TENV_0,\TENV_1[\eta]}{\tau[\eta]} \supseteq
        \methodsofD{\TENV_0,\TENV_1[\eta]}{\tau_\III'[\eta]}$ \` by definition\\
        $\IMPLSG{\TENV_1[\eta]}{\tau[\eta]}
        {\tau_\III'[\eta]}$ \` by rule $\RULENAME{$\IMPLOP_\III$}$    
      \end{tabbing}

       \item[Case:] $\RULENAME{$\IMPLOP$-param}$

     \begin{tabbing}
       $\tau_1 = \tau_2 = \beta \in (\TENV_0 , \ov{\alpha : 
  \tau_\III} , \TENV_1)$ \` this case\\
{\bf Subcase:} $\beta \in \ov{\alpha}$\\
$\JUDGEWFG{\TENV_0, \TENV_1[\eta]}{\tau_0}$
\` by Lemma~\ref{lem:fgg_weak}\\
$\IMPLSG{\TENV_0 ,
  \TENV_1[\eta]}{\tau_0}{\tau_0}$ \` by reflexivity \\

{\bf Subcase:} $\beta \notin \ov{\alpha}$\\
$\IMPLSG{\TENV_0 , \TENV_1[\eta] }
{\beta}{\beta}$ \` by rule $\RULENAME{$\IMPLOP$-param}$

\end{tabbing}

\item[Case:] $\RULENAME{$\IMPLOP_\SSS$}$
      \begin{tabbing}
        $\tau_1 = \tau_2 = \tau_\SSS$ \` this case\\
        $\IMPLSG{\TENV_0 , \TENV_1 [\eta]}
        {\tau_\SSS  [\eta]}{\tau_\SSS
          [\eta]}$ \` by rule $\RULENAME{$\IMPLOP_\SSS$}$
       \end{tabbing}
  
\end{description}

\end{proof}

   \begin{lemma}[Type Substitution Preserves Well-Formedness]
      \label{lem:fgg_tsubstwf}
      Let  $\JUDGEWFG{\TENV_0 , \ov{\alpha : \tau_\III}  , 
        \TENV_1}{\tau}$, $\JUDGEWFG{\TENV_0}{\ov{\tau_0}}$,
      $\SUBSTBUILDER{\eta}{\ov{\alpha}}{\ov{\tau_0}}$ and
      $\IMPLSG{\TENV_0}{(\ov{\tau_0}}{\ov{\tau_\III)}[\eta]}$
      then
      $\JUDGEWFG{\TENV_0,\TENV_1[\eta] }{\tau[\eta]}$.
     \end{lemma}

     \begin{proof}
By induction on the derivation of $\JUDGEWFG{\TENV_0 , \ov{\alpha :
  \tau_\III} , \TENV_1}{\tau}$.

\begin{description}
\item[Case:]$\RULENAME{t-named}$

  \begin{tabbing}
    $\tau = \TYPINSTtn{\ttt}{\ov{\tau'}}$ \` this case\\
    $\JUDGEWFG{\TENV_0 , \ov{\alpha : \tau_\III} ,
      \TENV_1}{\ov{\tau'}}$ \` by inversion\\
$ (\kw{type}~t(\kw{type}~\ov{\beta~\tau_\III})~T) \in \ov{D}$ \` by
inversion\\
$\eta' = (\ov{\beta~\tau_\III} \by_{\TENV_0 , \ov{\alpha : \tau_\III} ,
      \TENV_1} \ov{\tau'})$ \` by inversion\\
    $\IMPLSG{\TENV_0 , \ov{\alpha : \tau_\III},
      \TENV_1}{(\ov{\beta}}{\ov{\tau_\III}})\SUBS{\ov{\tau}}{\ov{\beta}}$
    \` by inversion\\
   $\JUDGEWFG{\TENV_0 ,
      \TENV_1[\eta]
    }{\ov{\tau'[\eta]}}$ \` by i.h.\\
    $\IMPLSG{\TENV_0 , \TENV_1[\eta]}
    {(\ov{\beta}}{\ov{\tau_\III}[\eta]})\SUBS{\ov{\tau}[\eta]}{\ov{\beta}}$
\` by Lemma~\ref{lem:fgg_tsubstsub}\\
    
    $\JUDGEWFG{\TENV_0 ,
      \TENV_1[\eta]}{\TYPINSTtn{\ttt}{\ov{\tau'[\eta]}}}$
    \` by rule $\RULENAME{t-named}$
  \end{tabbing}

\item[Case:] $\RULENAME{t-param}$
  \begin{tabbing}
    $\tau = \beta$ \` this case\\
    {\bf Subcase:} $\beta \in \ov{\alpha}$ \\
    $\tau[\eta] = \tau_0$ \` by definition\\
    $\JUDGEWFG{\TENV_0,\TENV_1[\eta]}{\tau_0}$ \` by
    Lemma~\ref{lem:fgg_weak}\\
    {\bf Subcase:} $\beta \not\in \ov{\alpha}$ \\
    $\tau[\eta] = \beta$ \` by definition\\
    $\JUDGEWFG{\TENV_0,\TENV_1[\eta]}{\beta}$ \`
    by rule $\RULENAME{t-param}$

  \end{tabbing}
  
 \end{description}

\end{proof}

  \begin{lemma}
    \label{lem:fgg_methodssub}
If $\JUDGEWFG{\TENV}{\tau}$ and 
 $\mmm  \
          \GPARAMSk{\ov{\PDECLkN{\kkk}{\tau_\III}}}
          \
          \MSIGpt{\ov{\pDECLxt{\yyy}{\tau_p}}}{\tau_r}
          \in
          \methodsofD{\TENV}{
          \tau
          }$
then for any $\tau_0$ such that 
$\IMPLSG{\TENV}{\tau_0}{\tau}$ and $\JUDGEWFG{\TENV}{\tau_0}$
we have that 
 $\mmm  \
          \GPARAMSk{\ov{\PDECLkN{\kkk}{\tau_\III}}}
          \
          \MSIGpt{\ov{\pDECLxt{\yyy}{\tau_p}}}{\tau_r}
          \in
          \methodsofD{\TENV}{
          \tau_0
          }$
        \end{lemma}
        \begin{proof}
          By inversion on $\IMPLSG{\TENV}{\tau_0}{\tau}$.
\begin{description}
\item[Case:] $\RULENAME{$\IMPLOP_\III$}$
  \begin{tabbing}
$\methodsofD{\TENV}{\tau_0} \supseteq \methodsofD{\TENV}{\tau}$ \` by inversion\\

    {\bf Subcase:} $\tau_0$ is some struct or interface type\\
$\mmm  \
          \GPARAMSk{\ov{\PDECLkN{\kkk}{\tau_\III}}}
          \
          \MSIGpt{\ov{\pDECLxt{\yyy}{\tau_p}}}{\tau_r}
          \in
          \methodsofD{\TENV}{
          \tau_0
        }$ \` by definition\\
        {\bf Subcase:} $\tau_0 = \alpha$\\
        $\methodsofD{\TENV}{\tau_0} = \methodsofD{\TENV}{\tau'_I}$, for
        some $\tau'_I$ such that $(\alpha : \tau'_I) \in \TENV$ \` by
        definition\\
$\mmm  \
          \GPARAMSk{\ov{\PDECLkN{\kkk}{\tau_\III}}}
          \
          \MSIGpt{\ov{\pDECLxt{\yyy}{\tau_p}}}{\tau_r}
          \in
          \methodsofD{\TENV}{
          \BOUNDTENVt{\TENV}{\tau_0}
        }$ \` by above reasoning
\end{tabbing}

\item[Case:] $\RULENAME{$\IMPLOP_\SSS$}$

  \begin{tabbing}
Trivial since $\tau_0 = \tau = \tau_\SSS$, for some $\tau_\SSS$.
\end{tabbing}

\item[Case:] $\RULENAME{$\IMPLOP$-param}$

  \begin{tabbing}
    Trivial since $\tau_0 = \tau = \alpha$.
  \end{tabbing}

   \end{description}
          
        \end{proof}

        \begin{lemma}
\label{lem:fgg_methodssubst}
Let $\mmm  \
          \GPARAMSk{\ov{\PDECLkN{\beta}{\tau_\III}}}
          \
          \MSIGpt{\ov{\pDECLxt{\yyy}{\tau_p}}}{\tau_r}
          \in
          \methodsofD{
          \TENV_0,\ov{\alpha : \tau_\III} ,
      \TENV_1}{\tau_\eee}
        $, $\JUDGEWFG{\TENV_0}{\ov{\tau_0}}$, 
$\eta = (\ov{\alpha} \by \ov{\tau_0})$ and
$\IMPLSG{\TENV_0}{(\ov{\tau_0}}{\ov{\tau_\III})[\eta]}$. It follows
that
 $\mmm  \
          \GPARAMSk{\ov{\PDECLkN{\beta}{\tau_\III[\eta]}}}
          \
          \MSIGpt{\ov{\pDECLxt{\yyy}{\tau_p[\eta]}}}{\tau_r[\eta]}
          \in
          \methodsofD{
          \TENV_0,\TENV_1 [\eta] }
  {\tau_\eee [\eta]}
        $
        \end{lemma}
        \begin{proof}
          By the definition of $\methods$, $(\eta = (\Psi \by_\Delta \psi)$ and Lemma~\ref{lem:fgg_tsubstsub} 
         \end{proof}

\begin{lemma}[Type Substitution Preserves Typing]
  \label{lem:fgg_tsubsttyp}
Let
$\JUDGEEXPRG{\TENV_0,\ov{\alpha : \tau_\III} ,
  \TENV_1}{\Ga}{\eee}{\tau}$, $\JUDGEWFG{\TENV_0}{\ov{\tau_0}}$, 
$\eta = (\ov{\alpha} \by \ov{\tau_0})$ and
$\IMPLSG{\TENV_0}{(\ov{\tau_0}}{\ov{\tau_\III})[\eta]}$.
It follows that 
$\JUDGEEXPRG{\TENV_0,\TENV_1[\eta]}
{\Ga[\eta]}{\eee[\eta]}{\tau[\eta]}$.
\end{lemma}

\begin{proof}
By induction on the derivation of $\JUDGEEXPRG{\TENV_0,\ov{\alpha :
  \tau_\III }, \TENV_1}{\Ga}{\eee}{\tau}$. 

\begin{description}
\item[Case:] $\RULENAME{t-var}$

\begin{tabbing}
$\JUDGEEXPRG{\TENV_0,\TENV_1[\eta]}{\Ga[\eta]}
{\xxx}{\Ga (\xxx) [\eta]}$ \` trivially by
$\RULENAME{t-var}$\\
\end{tabbing}

\item[Case:] $\RULENAME{t-literal}$
  \begin{tabbing}
    $\eee = \STRUCTte{\tau_\SSS'}{\ov{\eee}}$ and $\tau = \tau_\SSS'$ \` this case\\
   $\JUDGEWFG{\TENV_0,\ov{\alpha : \tau_\III} , \TENV_1}{\tau_\SSS'}$,
    $\JUDGEEXPRG{\TENV_0,\ov{\alpha : \tau_\III} , \TENV_1}{\Ga}{\ov{\eee}}{\ov{\tau}}$,
    $\FIELDSOFtD{\tau_\SSS'} = \ov{\FDECLft{\fff}{
       \tau_\fff
       }}$, $\IMPLSG{\TENV_0,\ov{\alpha : \tau_\III },
       \TENV_1}{\ov{\tau}}{\ov{\tau_\fff}}$ \\\` by inversion\\
     $\JUDGEWFG{\TENV_0 , \TENV_1\SUBS{\ov{\tau_0}}{\ov{\alpha}}}{\tau_\SSS'\SUBS{\ov{\tau_0}}{\ov{\alpha}}}$
     \` by Lemma~\ref{lem:fgg_tsubstwf}\\
     $\JUDGEEXPRG{\TENV_0,
       \TENV_1[\eta]}{\Ga[\eta]}{\ov{\eee}[\eta]}{\ov{\tau}[\eta]}$
    \` by i.h.\\
$\IMPLSG{\TENV_0,\TENV_1[\eta]}{\ov{\tau}[\eta]}{\ov{\tau_\fff}[\eta]}$
\` by Lemma~\ref{lem:fgg_tsubstsub}\\
$\FIELDSOFtD{\tau_\SSS'[\eta]} = \ov{\FDECLft{\fff}{
       \tau_\fff[\eta]
     }}$ \` by definition\\
   $\JUDGEEXPRG{\TENV_0,
     \TENV_1[\eta]}{\Ga[\eta]}{\STRUCTte{\tau_\SSS'[\eta]}{\ov{\eee}[\eta]}}
   {\tau_\SSS'[\eta]}$ \` by $\RULENAME{t-literal}$\\
 \end{tabbing}

 \item[Case:] $\RULENAME{t-field}$
  \begin{tabbing}
    $e = \SELef{\eee}{\fff}$ \` this case\\
    $\JUDGEEXPRG{\TENV_0,\ov{\alpha : \tau_\III },
      \TENV_1}{\Ga}{\eee}{\tau_\SSS}$ and $\FDECLft{\fff}{\tau} \in
    \FIELDSOFtD{\tau_\SSS}$\` by inversion\\
    $\JUDGEEXPRG{\TENV_0,\TENV_1[\eta]}{\Ga[\eta]}
    {\eee[\eta]}{\tau_\SSS[\eta]}$ \` by i.h.\\
    $\FDECLft{\fff}{\tau[\eta]} \in
    \FIELDSOFtD{\tau_\SSS[\eta]}$ \` by definition\\
    $\JUDGEEXPRG{\TENV_0,\TENV_1[\eta]}{\Ga[\eta]}
    {\SELef{\eee[\eta]}{\fff}}{\tau[\eta]}$
    \` by $\RULENAME{t-field}$\\
   \end{tabbing}

\item[Case:] $\RULENAME{t-call}$
  \begin{tabbing}
    $e = \MCALLemne{\eee}{\mmm}{\psi}{\ov{\eee}}$ and
    $\tau = \tau_r[\eta']$\` this case\\
    $\mmm  \
          \GPARAMSk{\ov{\PDECLkN{\beta}{\tau_\III}}}
          \
          \MSIGpt{\ov{\pDECLxt{\yyy}{\tau_p}}}{\tau_r}
          \in
          \methodsofD{
          \TENV_0,\ov{\alpha : \tau_\III} ,
      \TENV_1}{\tau_\eee}
        $ \` by inversion\\
        $\JUDGEEXPRG{\TENV_0,\ov{\alpha : \tau_\III} ,
      \TENV_1}{\Ga}{\eee}{\tau_\eee}$ \` by inversion\\
    $\SUBSTBUILDERCTXT{\eta'}{(\TENV_0,\ov{\alpha : \tau_\III} , \TENV_1)}
    {\GPARAMSk{\ov{\PDECLkN{\beta}{\tau_\III}}}}{\psi}$ and
    $\IMPLSG{\TENV_0,\ov{\alpha : \tau_\III} , \TENV_1}{(\ov{\beta}}{\ov{\tau_I})[\eta']}$
\` by inversion\\
        $\JUDGEEXPRG{\TENV_0,\ov{\alpha : \tau_\III },
      \TENV_1}{\Ga}{\ov{\eee}}{\ov{\tau_a}}$ \` by inversion\\
        $\IMPLSG{\TENV_0,\ov{\alpha : \tau_\III },
      \TENV_1}{(\ov{\tau_a}}{\ov{\tau_p})[\eta']
    }$ \` by inversion\\
    $\JUDGEEXPRG{\TENV_0,
      \TENV_1[\eta]}{\Ga[\eta]}{\eee[\eta]}
    {\tau_\eee[\eta]}$
   \` by i.h.\\
$\IMPLSG{\TENV_0,
      \TENV_1[\eta]}{\psi[\eta]}{\ov{\tau_\III}
          [\eta'][\eta] }$ \` by
        Lemma~\ref{lem:fgg_tsubstsub}\\
        $\JUDGEEXPRG{\TENV_0,
          \TENV_1[\eta]}{\Ga[\eta]}{\ov{\eee [\eta]}
        }{\ov{\tau_a [\eta]}}$
\` by i.h.\\
$\IMPLSG{\TENV_0,
      \TENV_1[\eta]}{\ov{\tau_a}[\eta'][\eta]}{\ov{\tau_p}
      [\eta'][\eta]}$ \` by
    Lemma~\ref{lem:fgg_tsubstsub}\\

    $\mmm  \
          \GPARAMSk{\ov{\PDECLkN{\beta}{\tau_\III[\eta]}}}
          \
          \MSIGpt{\ov{\pDECLxt{\yyy}{\tau_p[\eta]}}}{\tau_r[\eta]}
          \in
          \methodsofD{
          \TENV_0,\TENV_1 [\eta] }
  {\tau_\eee [\eta]}
  $ \` by  Lemma~\ref{lem:fgg_methodssubst}\\
  Let $\eta'' = (\eta \circ \eta')$\\
  $\SUBSTBUILDERCTXT{\eta''}{(\TENV_0, \TENV_1[\eta])}
    {\GPARAMSk{\ov{\PDECLkN{\beta}{\tau_\III[\eta'][\eta]}}}}{\psi[\eta'][\eta]}$
    \` by definition of $\SUBSTBUILDERCTXT{\eta}{\TENV}{\Psi}{\psi}$
    and above\\
$\JUDGEEXPRG{\TENV_0,
          \TENV_1[\eta]}{\Ga[\eta]}
        {\MCALLemne{\eee [\eta]}{\mmm}{\psi[\eta]}{\ov{\eee [\eta]}}}
        {\tau_r [\eta'][\eta]
        }$
\` by rule $\RULENAME{t-call}$
    \end{tabbing}

\item[Case:] $\RULENAME{t-assert$_\III$}$
\begin{tabbing}
$\eee= \ASSRTet{\eee}{\tau_\III'}$ and $\tau = \tau_\III'$ \` this
case\\
$\JUDGEWFG{\TENV_0, \ov{\alpha :\tau_\III }, \TENV_1 }{\tau_\III'}$ and 
$\JUDGEEXPRG{\TENV_0, \ov{\alpha : \tau_\III} , \TENV_1
}{\Ga}{\eee}{\tau_\III'}$ \` by inversion\\
$\JUDGEWFG{\TENV_0,\TENV_1[\eta] }{\tau_\III'[\eta]}$ \` by
Lemma~\ref{lem:fgg_tsubstwf} \\
$\JUDGEEXPRG{\TENV_0, \TENV_1 [\eta]
}{\Ga [\eta]}{\eee [\eta]}{\tau_\III' [\eta]}$
\` by i.h.\\
$\JUDGEEXPRG{\TENV_0, \TENV_1 [\eta]
}{\Ga [\eta]}
{\ASSRTet{\eee [\eta]}{\tau_\III' [\eta]}}
{ \tau_\III' [\eta]}$ \` by rule $\RULENAME{t-assert$_\III$}$

\end{tabbing}

\item[Case:] $\RULENAME{t-assert$_\SSS$}$
  
  \begin{tabbing}
$\eee = \ASSRTet{\eee}{\tau_\SSS}$ and $\tau = \tau_\SSS$ \` this
case\\
$\JUDGEWFG{\TENV_0, \ov{\alpha : \tau_\III} , \TENV_1 }{\tau_\SSS}$ \` by inversion\\

$ \JUDGEEXPRG{\TENV_0, \ov{\alpha : \tau_\III }, \TENV_1 }{\Ga}{\eee}{\tau_\III'}$ \` by inversion\\
$\IMPLSG{\TENV_0, \ov{\alpha : \tau_\III} , \TENV_1 }{\tau_\SSS}{
  \BOUNDTENVt{\TENV_0, \ov{\alpha : \tau_\III} , \TENV_1}{\tau_\III'}}$ \` by inversion\\

$\JUDGEWFG{\TENV_0, \TENV_1 [\eta] }{\tau_\SSS [\eta]}$ \` by
Lemma~\ref{lem:fgg_tsubstwf} \\
$\IMPLSG{\TENV_0, \TENV_1 [\eta] }{\tau_\SSS [\eta] }{
  \BOUNDTENVt{\TENV_0, \ov{\alpha : \tau_\III} , \TENV_1}{\tau_\III'}[\eta] }$ \` by
Lemma~\ref{lem:fgg_tsubstsub}\\
$\IMPLSG{\TENV_0, \TENV_1 [\eta] }
{\tau_\SSS [\eta]}
{\BOUNDTENVt{\TENV_0,\TENV_1[\eta]}{\tau_\III'[\eta]}}$
\` by Lemma~\ref{lem:fgg_substbound} and Lemma~\ref{lem:fgg_subequiv}\\
$ \JUDGEEXPRG{\TENV_0 , \TENV_1 [\eta] }{\Ga
  [\eta]}
{\eee [\eta]}{\tau_\III'
  [\eta]}$
\` by i.h.\\

  $ \JUDGEEXPRG{\TENV_0 , \TENV_1 [\eta] }{\Ga
  [\eta]}
{ \ASSRTet{\eee [\eta] }{\tau_\SSS
    [\eta]}}
{\tau_\SSS
  [\eta]}$ \` by rule $\RULENAME{t-assert$_\SSS$}$\\

    \end{tabbing}

    \item[Case:] $\tyrulename{t-stupid}$

    \begin{tabbing}
      $\eee = \ASSRTet{\eee}{\tau}$ and $\tau = \tau$ \` this case\\
      $ \JUDGEWFG{\TENV_0, \ov{\alpha : \tau_\III} , \TENV_1}{\tau_\SSS}$ \` by
inversion\\
$ \JUDGEEXPRG{\TENV_0, \alpha : \tau_\III ,
  \TENV_1}{\Ga}{\eee}{\tau_\SSS}$ \` by inversion\\
$\JUDGEWFG{\TENV_0, \TENV_1 [\eta]}{\tau_\SSS[\eta]}$ \` by Lemma~\ref{lem:fgg_tsubstwf} \\
   $ \JUDGEEXPRG{\TENV_0, 
        \TENV_1 [\eta]}{\Ga
        [\eta]}{\eee [\eta] }{\tau_\SSS [\eta]}$
\` by i.h.\\

$ \JUDGEEXPRG{\TENV_0, 
        \TENV_1 [\eta] }{\Ga
        [\eta]}
        {\ASSRTet{\eee [\eta] }{\tau
            [\eta]}}
        {\tau[\eta]}$ 
        \`  by rule
        $\tyrulename{t-stupid}$ \\
      \end{tabbing}

\end{description}
\end{proof}

\begin{lemma}[Value Substitution Preserves Typing]
\label{lem:fgg_subst}
If $\JUDGEEXPRG{\emptyset}{\Ga , \ov{\xxx : \tau}}{\eee}{\tau'}$ and 
$\JUDGEEXPRG{\emptyset}{\emptyset}{\ov{\vvv_0}}{\ov{\tau_0}}$ where
$\IMPLSG{}{\ov{\tau_0}}{\ov{\tau}}$ then 
$\JUDGEEXPRG{\emptyset}{\Ga}{\eee\SUBS{\ov{\vvv_0}}{\ov{\xxx}}}{\tau''}$, for some
$\tau''$ with $\IMPLSG{\emptyset}{\tau''}{\tau'}$.
\end{lemma}
\begin{proof}
By induction on the derivation of $\JUDGEEXPRG{\emptyset}{\Ga , \ov{x :
    \tau}}{\eee}{\tau'}$.

\begin{description}
\item[Case:] Rule $\tyrulename{t-var}$

  \begin{tabbing}
  $\eee = \yyy$ and $\tau' = \Ga(y)$ \\
  {\bf Sub}\={\bf case:} $\xxx \neq \yyy$ \\
  \>$\yyy\SUBS{\ov{\vvv_0}}{\ov{\xxx}} = \yyy$ \` by definition\\
  \>$\JUDGEEXPRG{\emptyset}{\Ga}{\yyy}{\tau}$ \` by $\tyrulename{t-var}$\\
  \>$\IMPLSG{\emptyset}{\tau}{\tau}$ \` by reflexivity\\
  {\bf Subcase:} $\xxx = \yyy$ \\
  \>$\xxx\SUBS{\ov{\vvv_0}}{\ov{\xxx}} = \vvv_0$ and $\Ga(x) = \tau'$ \\
 \>$\JUDGEEXPRG{\emptyset}{\Ga}{\vvv}{\tau_0}$ and $\IMPLS{\tau_0}{\tau'}$ \` assumption
\end{tabbing}

\item[Case:] Rule $\tyrulename{t-literal}$
   \begin{tabbing}
     $\eee =\STRUCTte{\tau_\SSS}{\ov{\eee_f}}$,$\tau' = \tttS$ and
     $\eee\SUBS{\ov{\vvv_0}}{\ov{\xxx}} = \STRUCTte{\tau_\SSS}{ \eee_f\SUBS{\ov{\vvv_0}}{\ov{\xxx}}}$ \\
  $\JUDGEEXPRG{\emptyset}{\Ga,\ov{\xxx : \tau}}{\ov{\eee_f}}{\ov{\ttt_a}}$ \` by inversion\\
  $\FIELDSOFtD{\tttS} = \ov{\FDECLft{\fff}{\ttt_f}}$ \` by inversion\\
  $\impls{\ov{\ttt_a}}{\ov{\ttt_f}}$ \` by inversion\\
  $\JUDGEWFG{\emptyset}{\tau_\SSS}$ \` by inversion\\
  $\JUDGEEXPRG{\emptyset}{\Ga}{\ov{\eee_f\SUBS{\ov{\vvv_0}}{\ov{\xxx}}}}{\ov{\ttt_c}}$, for
  some $\IMPLSG{\emptyset}{\ov{\ttt_c}}{\ov{\ttt_a}}$ \` by i.h.\\
  $\IMPLSG{\emptyset}{\ov{\ttt_c}}{\ov{\ttt_f}}$ \` by transitivity\\
  $\JUDGEEXPRG{\emptyset}{\Ga}{\STRUCTte{\tau_\SSS}{ \eee_f\SUBS{\ov{\vvv_0}}{\ov{\xxx}}}
  }{\tau_\SSS}$ \` by $\tyrulename{t-literal}$\\
 $\IMPLSG{\emptyset}{ \tau_\SSS}{\tau_\SSS}$ \` by reflexivity
\end{tabbing}

\item[Case:] Rule $\tyrulename{t-field}$

\begin{tabbing}
  $\eee = \SELef{\eee_s}{\fff}$, $\tau' = \ttt_f$ and
$\eee\SUBS{\ov{\vvv_0}}{\ov{\xxx}} = \SELef{\eee_s\SUBS{\ov{\vvv_0}}{\ov{\xxx}}}{\fff}$\\

$\JUDGEEXPRG{\emptyset}{\Ga, \ov{\xxx : \tau}}{\eee_s}{\tau_\SSS}$ \` by inversion\\
$\FDECLft{\fff}{\ttt_f} \in \FIELDSOFtD{\tau_\SSS}$ \` by inversion\\

$\JUDGEEXPRG{\emptyset}{\Ga}{\eee_s\SUBS{\ov{\vvv_0}}{\ov{\xxx}}}{\ttt_b}$, for some
$\IMPLS{\ttt_b}{\tau_\SSS}$ \` by i.h.\\
$\ttt_b = \tau_\SSS$ \` by Lemma~\ref{lem:fgg_invsub}\\
$\JUDGEEXPRG{\emptyset}{\Ga}{\eee_s\SUBS{\ov{\vvv_0}}{\ov{\xxx}}}{\tau_\SSS}$ \` substituting for
equals\\
$\JUDGEEXPR{\emptyset}{\Ga}{\SELef{\eee_s}{\fff}}{\ttt_f}$ \` by
$\tyrulename{t-field}$\\
$\IMPLSG{\emptyset}{\ttt_f}{\ttt_f}$ \` by reflexivity

\end{tabbing}

\item[Case:] Rule $\tyrulename{t-call}$
  \begin{tabbing}
$\mmm  \
          \GPARAMSk{\ov{\PDECLkN{\kkk}{\tau_\III}}}
          \
          \MSIGpt{\ov{\pDECLxt{\yyy}{\tau_p}}}{\tau_r}
          \in
          \methodsofD{
          \emptyset}{\tau_m}
        $ \` by inversion\\
  $\JUDGEEXPRG{\emptyset}{\Ga, \ov{\xxx : \tau}}{\eee}{\tau_m}$ \` by
  inversion\\
    $\JUDGEEXPRG{\emptyset}{\Ga ,\ov{\xxx : \tau}}{\ov{\eee}}{\ov{\tau_a}}$ \` by
    inversion\\

$\eta = (\GPARAMSk{\ov{\PDECLkN{\kkk}{\tau_\III}}} \by \psi)$ \` by inversion\\
      $\IMPLSG{\emptyset}{(\ov{\tau_a}}{\ov{\tau_p}
        )[\eta]}$ \` by inversion\\
        $\JUDGEEXPRG{\emptyset}{\Ga}{\eee\SUBS{\ov{\vvv_0}}{\ov{\xxx}}}{\tau_m'}$
        for some $\IMPLSG{\emptyset}{\tau_m'}{\tau_m}$ \` by
        i.h.\\
        $\JUDGEEXPRG{\emptyset}{\Ga
        }{\ov{\eee\SUBS{\ov{\vvv_0}}{\ov{\xxx}}}}{\ov{\tau_a'}}$ for
        some $\IMPLSG{\emptyset}{\ov{\tau_a'}}{\ov{\tau_a}}$ \` by
        i.h.\\ 
        $\IMPLSG{\emptyset}{\ov{\tau_a'}}{\ov{\tau_p}
          [\eta] }$ \` by transitivity\\
        $\mmm  \
          \GPARAMSk{\ov{\PDECLkN{\kkk}{\tau_\III}}}
          \
          \MSIGpt{\ov{\pDECLxt{\yyy}{\tau_p}}}{\tau_r}
          \in
          \methodsofD{
          \emptyset}{\tau_m'}
        $ \` by Lemma~\ref{lem:fgg_methodssub}\\
        $ \JUDGEEXPRG{\emptyset}{\Ga}
       {
       \MCALLemne{\eee\SUBS{\ov{\vvv_0}}{\ov{\xxx}}}{\mmm}{\psi}{\ov{\eee}\SUBS{\ov{\vvv_0}}{\ov{\xxx}}}
       }
       {\tau_r [\eta] }$ \` by rule $\tyrulename{t-call}$
    \end{tabbing}

\item[Case:] Rule $\tyrulename{t-assert$_\III$}$
    \begin{tabbing}
$  \JUDGEWFG{\emptyset}{\tau_\III}$ and
$\JUDGEEXPRG{\emptyset}{\Ga, \ov{\xxx : \tau}}{\eee}{\tau_\III'}$ \` by inversion\\
$\JUDGEEXPRG{\emptyset}{\Ga}{\eee\SUBS{\ov{\vvv_0}}{\ov{\xxx}}}{\tau''}$ with
$\IMPLSG{\emptyset}{\tau''}{\tau_\III'}$ \` by i.h.\\
{\bf Subcase:} $\tau''$ is an interface type\\
$\JUDGEEXPRG{\emptyset}{\Ga}{\ASSRTet{\eee\SUBS{\ov{\vvv_0}}{\ov{\xxx}}}{\tau_\III}}{\tau_\III}$
\` by rule $\tyrulename{t-assert$_\III$}$\\
{\bf Subcase:} $\tau''$ is some struct type $\tau_\SSS$\\
$\JUDGEEXPRG{\emptyset}{\Ga}{\ASSRTet{\eee\SUBS{\ov{\vvv_0}}{\ov{\xxx}}}{\tau_\III}}{\tau_\III}$
\` by rule $\tyrulename{t-stupid}$ 
      \end{tabbing}
    
\item[Case:] Rule $\tyrulename{t-assert$_\SSS$}$
  \begin{tabbing}
 $\JUDGEWFG{\emptyset}{\tau_\SSS}$ \` by inversion\\
   
    $\JUDGEEXPRG{\emptyset}{\Ga, \ov{\xxx : \tau} }{\eee}{\tau_\III}$ \` by
    inversion\\
    $\IMPLSG{\emptyset}{\tau_\SSS}{  \BOUNDTENVt{\emptyset}{\tau_\III}}$ \` by
    inversion\\
$\JUDGEEXPRG{\emptyset}{\Ga}{\eee\SUBS{\ov{\vvv_0}}{\ov{\xxx}}}{\tau''}$ with
$\IMPLSG{\emptyset}{\tau''}{\tau_\III}$  \` by i.h.\\
{\bf Subcase:} $\tau''$ is an interface type\\
$\tau'' = \tau_\III$ \` by value restriction\\
$\JUDGEEXPRG{\emptyset}{\Ga}{\ASSRTet{\eee\SUBS{\ov{\vvv_0}}{\ov{\xxx}}}{\tau_\SSS}}{\tau_\SSS}$
\` by $\tyrulename{t-assert$_\SSS$}$ \\
{\bf Subcase:} $\tau''$ is a struct type\\
$\JUDGEEXPRG{\emptyset}{\Ga}{\ASSRTet{\eee\SUBS{\ov{\vvv_0}}{\ov{\xxx}}}{\tau_\SSS}}{\tau_\SSS}$
\` by $\tyrulename{t-stupid}$ \\
\end{tabbing}

 \item[Case:] Rule $\tyrulename{t-stupid}$
   \begin{tabbing}
     $\JUDGEWFG{\emptyset}{\tau_\SSS}$ \` by inversion\\
   
      $\JUDGEEXPRG{\emptyset}{\Ga, \ov{\xxx : \tau}}{\eee}{\tau_\SSS}$ \` by inversion\\
$\JUDGEEXPRG{\emptyset}{\Ga}{\eee\SUBS{\ov{\vvv_0}}{\ov{\xxx}}}{\tau''}$ such
    that $\IMPLSG{\emptyset}{\tau''}{\tau_\SSS}$ \` by i.h.\\
    $\tau''$ is a struct type \` by Lemma~\ref{lem:fgg_invsub}\\
    $\JUDGEEXPRG{\emptyset}{\Ga}{\ASSRTet{\eee\SUBS{\ov{\vvv_0}}{\ov{\xxx}}}{\tau}}{\tau}$
    \` by rule $\tyrulename{t-stupid}$
    \end{tabbing}

  \end{description}

\end{proof}

 \begin{lemma}
   \label{lem:fgg_methtyp}
If $ \mmm  \
          \GPARAMSk{\ov{\PDECLkN{\kkk}{\tau_\III}}}
          \
          \MSIGpt{\ov{\pDECLxt{\yyy}{\tau_p}}}{\tau_r}
          \in
          \methodsof{\NTYPEtt{\ttt_\SSS}{\ov{\tau_\ttt}}}
          $ and
          $\MBODYmtN{\TENV}{\mmm}{\ov{\tau_\mmm}}{\NTYPEtt{\ttt_\SSS}{\ov{\tau_\ttt}}}
          = \MBODYxxe{\xxx}{\ov{\yyy}}{\eee_0}$
          where
          $\JUDGEWFG{\TENV}{\NTYPEtt{\ttt_\SSS}{\ov{\tau_\ttt}}}$,
          $\JUDGEWFG{\TENV}{\ov{\tau_\mmm}}$ and
          $\IMPLSG{\TENV}{\ov{\tau_\mmm}}{\ov{\tau_\III\SUBS{\ov{\tau_\mmm}}{\ov{\kkk}}}}$
          then there exists $\tau'$ such that
          $\JUDGEWFG{\TENV}{\tau'}$, 
          $\IMPLSG{\TENV}{\tau'}{\tau_r\SUBS{\ov{\tau_\mmm}}{\ov{\kkk}}}$
          and $\JUDGEEXPRG{\TENV}{\ov{\yyy} :
            \ov{\tau_p}\SUBS{\ov{\tau_\mmm}}{\ov{\kkk}} , \xxx :
            \NTYPEtt{\ttt_\SSS}{\ov{\tau_\ttt}}}{\eee_0}{\tau'}$.
  \end{lemma}

  \begin{proof}
    ~
    \begin{tabbing}
    $\MBODYxxe{\xxx}{\ov{\yyy}}{\eee_0} = \MBODYxxe{\xxx}{\ov{\yyy}}
  {\eee_0' \SUBS{\ov{\tau_\mmm}}{\ov{\kkk}}\SUBS{\ov{\tau_\ttt}}{\hat\Phi}}
  $ \` by inversion on $\mathit{body}$\\
 $\MDECLpmGe{\pDECLxt{\xxx}{\NTYPEtt{\ttt_\SSS}{{\Phi}}}}{\mmm}
    {
      \MSIGPpt{\GPARAMSk{\ov{\PDECLkN{\kkk}{{\tau_\III}}}}}
      {\ov{\pDECLxt{\yyy}{{\tau_p}}}}{{\tau_r}}}
    {\eee_0'}
    \in 
    \ov{D}$ \` by inversion on $\mathit{body}$\\
     $\JUDGEEXPRG{
       \Phi,
        \ov{\kkk: \tau_\III}
        }{\xxx : {\TYPINSTtn{\tttS}{\Phi}}, \ov{\yyy:
            \tau_p}}{\eee_0'}{\tau}$ \` by inversion on
        well-formedness\\
        $\IMPLSG{\Phi, \ov{\kkk:
            \tau_\III}}{\tau}{\tau_r}$ \` by inversion on
        well-formedness\\
        $\IMPLSG{\emptyset}{\tau\SUBS{\ov{\tau_\ttt}}{\ov{\kkk}}\SUBS{\ov{\tau_\mmm}}{\hat\Phi}}
        {\tau_r\SUBS{\ov{\tau_\ttt}}{\ov{\kkk}}\SUBS{\ov{\tau_\mmm}}{\hat\Phi}}$
        \` by Lemma~\ref{lem:fgg_tsubstsub}\\
   $\JUDGEEXPRG{\emptyset}{\ov{\yyy} :
            \ov{\tau_p}\SUBS{\ov{\tau_\mmm}}{\ov{\kkk}} , \xxx :
            \NTYPEtt{\ttt_\SSS}{\ov{\tau_\ttt}}}{\eee_0\SUBS{\ov{\tau_\ttt}}{\ov{\kkk}}\SUBS{\ov{\tau_\mmm}}{\hat\Phi}}{\tau\SUBS{\ov{\tau_\ttt}}{\ov{\kkk}}\SUBS{\ov{\tau_\mmm}}{\hat\Phi}}$

     \` by Lemma~\ref{lem:fgg_tsubsttyp} \\
\end{tabbing}
   \end{proof}

The proof of type preservation below relies on explicit congruence
rules for the reduction semantics.
   
\begin{theorem}[Type Preservation]
\label{thm:fgg_pres}

If $\JUDGEEXPRG{\emptyset}{\emptyset}{\eee}{\tau}$ and $\eee \REDUCE{} \eee'$
then
$\JUDGEEXPRG{\emptyset}{\emptyset}{\eee}{\tau'}$, for some $\tau'$ such that $\IMPLSG{\emptyset}{\tau'}{\tau}$.

\end{theorem}
\begin{proof}
  
By induction on the derivation of $\eee \REDUCE{} \eee'$.

\begin{description}
\item[Case:] $\RULENAME{r-field}$
\begin{tabbing}
$\FIELDSOFtD{\tau_\SSS} \EQ{}
\ov{\FDECLft{\fff}{{\tau}}}$ \` by inversion on the reduction
relation\\
$\JUDGEWFG{\emptyset}{\tau_\SSS}$,
$\JUDGEEXPRG{\emptyset}{\emptyset}{\ov{\vvv}}{\ov{\tau}}$ \` by inversion on
typing\\
$\JUDGEEXPRG{\emptyset}{\emptyset}{\SELef{\STRUCTte{\tau_\SSS}{\ov{\vvv}}}{\fff}}{\tau_\fff}$ \` by inversion
on typing\\
$\FIELDSOFtD{\tau_\SSS} = \ov{\FDECLft{\fff}{
       \tau_\fff}}$ and $\IMPLSG{\emptyset}{\ov{\tau}}{\ov{\tau_\fff}}$ \` by inversion on typing\\
$\JUDGEEXPRG{\emptyset}{\Ga}{\vvv_i}{\tau_i}$ with
$\IMPLSG{\emptyset}{\tau_i}{\tau_\fff}$ \` by the above\\
\end{tabbing}

\item[Case:] $\RULENAME{r-call}$

\begin{tabbing}
$	\MBODYmtN{\emptyset}{\mmm}{\ov{\tau}}{\vtype(\vvv)}
	=
	\MBODYxxe{\xxx}{\ov{\yyy}}{\eee}$ \` by inversion on the
        reduction relation\\
  $ \MCALLemne{\vvv}{\mmm}
   {\ov{\tau}}
   {\ov{\vvv}}
   \REDUCE{}
   \eee
   \SUBS{\vvv}{\xxx}
   \SUBS{\ov{\vvv}}{\ov{\yyy}}$ \` by inversion on the reduction
   relation\\
$\mmm  \
          \GPARAMSk{\ov{\PDECLkN{\kkk}{\tau_\III}}}
          \
          \MSIGpt{\ov{\pDECLxt{\yyy}{\tau_p}}}{\tau_r}
          \in
          \methodsofD{
          \emptyset}{\tau}
          $ \` by inversion on typing\\
$\JUDGEEXPRG{\emptyset}{\emptyset}{\vvv}{\tau}$ \` by inversion on typing\\
$\IMPLSG{\emptyset}{\ov{\tau}}{\ov{\tau_\III} \SUBS{\ov{\tau}}{\ov{\alpha}}
}$ \` by inversion on typing\\
 $\JUDGEEXPRG{\emptyset}{\emptyset}{\ov{\vvv}}{\ov{\tau_a}}$ \` by inversion on typing\\
$\IMPLSG{\emptyset}{\ov{\tau_a}}{\ov{\tau_p}
  \SUBS{\ov{\tau}}{\ov{\alpha}} }$
\` by inversion on typing\\
$\exists \tau'$ s. t.
          $\JUDGEWFG{\emptyset}{\tau'}$, 
          $\IMPLSG{\emptyset}{\tau'}{\tau_r\SUBS{\ov{\tau_\mmm}}{\ov{\kkk}}}$
          and $\JUDGEEXPRG{\emptyset}{\ov{\yyy} :
            \ov{\tau_p}\SUBS{\ov{\tau_\mmm}}{\ov{\kkk}} , \xxx :
            \NTYPEtt{\ttt_\SSS}{\ov{\tau_\ttt}}}{\eee}{\tau'}$\\
        \` by Lemma~\ref{lem:fgg_methtyp} \\
$\JUDGEEXPRG{\emptyset}{\emptyset}{\eee
   \SUBS{\vvv}{\xxx}
   \SUBS{\ov{\vvv}}{\ov{\yyy}}}{\tau''}$, for some $\tau''$ s.t. $\IMPLSG{\emptyset}{\tau''}{\tau'}$ \` by
 Lemma~\ref{lem:fgg_subst}\\
 $\IMPLSG{\emptyset}{\tau''}{\tau_r\SUBS{\ov{\tau_\mmm}}{\ov{\kkk}}}$
\` by transitivity\\
\end{tabbing}

\item[Case:] $\RULENAME{r-assert}$
  \begin{tabbing}
$\ASSRTet
   {\vvv}
   {\tau}
   \REDUCE
   \vvv$ \` this case\\
   $	\IMPLSG{\emptyset}{\vtype(\vvv)}{\tau}$ \` by inversion on
   the reduction relation\\
   $\JUDGEEXPRG{\emptyset}{\emptyset}{\vvv}{\vtype(\vvv)}$ \` by
   definition\\
   $\vvv = \STRUCTte{\tau_\SSS}{\ov{\vvv}}$, for some $\tau_\SSS$ and
   $\ov{\vvv}$ \` by definition of value\\
   $\JUDGEEXPRG{\emptyset}{\emptyset}{\vvv}{\tau_\SSS}$ \` by inversion on
   typing\\
  \end{tabbing}

\item[Case:] $\RULENAME{rc-recv}$
  \begin{tabbing}
$\MCALLemne{\eee}{\mmm}{\ov{\tau}}{\ov{\eee}}
     \REDUCE
     \MCALLemne{\eee'}{\mmm}{\ov{\tau}}{\ov{\eee}}$ \` this case\\
     $\eee \REDUCE \eee'$ \` by inversion on the reduction relation\\
     $\mmm  \
          \GPARAMSk{\ov{\PDECLkN{\ov{\alpha}}{\tau_\III}}}
          \
          \MSIGpt{\ov{\pDECLxt{\yyy}{\tau_p}}}{\tau_r}
          \in
          \methodsofD{
          \emptyset}{\tau}
        $ \` by inversion on typing\\
        $  \JUDGEEXPRG{\emptyset}{\emptyset}{\eee}{\tau}$ \` by inversion on
        typing\\
        $\IMPLSG{\emptyset}{\ov{\tau}}{\ov{\tau_\III}
          \SUBS{\ov{\tau}}{\ov{\alpha}} }$ \` by inversion on typing\\
         $\JUDGEEXPRG{\emptyset}{\emptyset}{\ov{\eee}}{\ov{\tau_a}}$ \` by
         inversion on typing\\
         $\IMPLSG{\emptyset}{\ov{\tau_a}}{\ov{\tau_p}
           \SUBS{\ov{\tau}}{\ov{\alpha}} }$ \` by inversion on typing\\
         $  \JUDGEEXPRG{\emptyset}{\emptyset}{\eee'}{\tau'}$, for some $\tau'$
         such that $\IMPLSG{\emptyset}{\tau'}{\tau}$ \` by i.h.\\
         $\mmm  \
          \GPARAMSk{\ov{\PDECLkN{\kkk}{\tau_\III}}}
          \
          \MSIGpt{\ov{\pDECLxt{\yyy}{\tau_p}}}{\tau_r}
          \in
          \methodsofD{
          \emptyset}{\tau'}
        $ \` by Lemma~\ref{lem:fgg_methodssub}\\
       $ \JUDGEEXPRG{\emptyset}{\emptyset}
       {
       \MCALLemne{\eee'}{\mmm}{\ov{\tau}}{\ov{\eee}}
       }
       {\tau_r \SUBS{\ov{\tau}}{\ov{\alpha}}}$ \` by rule  $\RULENAME{t-call}$
     \end{tabbing}

     \item[Case:] $\RULENAME{rc-field}$
  \begin{tabbing}
$\SELef{\eee}{\fff}
     \REDUCE
     \SELef{\eee'}{\fff}$ \` this case\\
     $\eee \REDUCE \eee'$ \` by inversion on the reduction relation\\
     $\JUDGEEXPRG{\emptyset}{\emptyset}{\eee}{\tau_\SSS}$
       and
       $\FDECLft{\fff}{\tau} \in \FIELDSOFtD{\tau_\SSS}$ \` by
       inversion on typing\\
       $\JUDGEEXPRG{\emptyset}{\emptyset}{\eee'}{\tau}$ with
       $\IMPLSG{\emptyset}{\tau}{\tau_\SSS}$ \` by i.h.\\
       $\tau = \tau_\SSS$ \` by Lemma~\ref{lem:fgg_invsub}\\
       $\JUDGEEXPRG{\emptyset}{\emptyset}{\SELef{\eee'}{\fff}}{\tau}$ \` by
       rule $\RULENAME{t-field}$

     \end{tabbing}

     \item[Case:] $\RULENAME{rc-literal}$
     \begin{tabbing}
$\STRUCTte{\tau_\SSS}{\ov{\vvv}\cdot \eee\cdot \ov{\eee}}
     \REDUCE
     \STRUCTte{\tau_\SSS}{\ov{\vvv}\cdot \eee'\cdot \ov{\eee}}$ \`
     this case\\
     $ \eee \REDUCE \eee'$ \` by inversion on the reduction relation\\
     $\JUDGEWFG{\emptyset}{\tau_\SSS}$ \` by inversion on typing\\
  
       $\JUDGEEXPRG{\emptyset}{\emptyset}{\ov{\vvv}\cdot \eee\cdot \ov{\eee}}{\ov{\tau}}$ \` by
       inversion on typing\\
  
       $\FIELDSOFtD{\tau_\SSS} = \ov{\FDECLft{\fff}{
       \tau_\fff
       }}$ \` by inversion on typing\\
      
       $\IMPLSG{\emptyset}{\ov{\tau}}{\ov{\tau_\fff}}$ \` by inversion
       on typing\\
       $\JUDGEEXPRG{\emptyset}{\emptyset}{\eee'}{\tau'}$ with
       $\IMPLSG{\emptyset}{\tau'}{\tau}$ \` by i.h.\\
       $\IMPLSG{\emptyset}{\tau'}{\tau_\fff}$ \` by transitivity\\
       $\JUDGEEXPRG{\emptyset}{\emptyset}{\STRUCTte{\tau_\SSS}{\ov{\vvv}\cdot \eee'\cdot \ov{\eee}}}{\tau_\SSS}$
       \` by rule $\RULENAME{t-literal}$
     \end{tabbing}

   \item[Case:] $\RULENAME{rc-assert}$

     \begin{tabbing}
$    \ASSRTet{\eee}{\tau}
     \REDUCE
     \ASSRTet{\eee'}{\tau}
   $ \` this case\\
   $ \eee \REDUCE \eee'$ \` by inversion on the reduction relation\\
   $\JUDGEWFG{\emptyset}{\tau}$ \` by inversion on typing\\
   {\bf Subcase:}  $\JUDGEEXPRG{\emptyset}{\emptyset}{\eee}{\tau_\III'}$ \\
   $\JUDGEEXPRG{\emptyset}{\emptyset}{\eee'}{\tau'}$ with
   $\IMPLSG{\emptyset}{\tau'}{\tau_\III'}$ \` by i.h.\\
   {\bf Subsubcase:} $\tau'$ is an interface type, 
   $\IMPLSG{\emptyset}{\tau}{\tau'}$ and $\tau$ is a
   struct type\\
   $\JUDGEEXPRG{\emptyset}{\emptyset}{\ASSRTet{\eee'}{\tau}}{\tau}$ \` by rule
   $\RULENAME{t-assert$_\SSS$}$\\
    {\bf Subsubcase:} $\tau'$ is an interface type and $\tau$ is an interface type\\
    $\JUDGEEXPRG{\emptyset}{\emptyset}{\ASSRTet{\eee'}{\tau}}{\tau}$ \` by rule
    $\RULENAME{t-assert$_\III$}$\\
    {\bf Subsubcase:} $\tau'$ is an interface type, 
   $\NIMPLSG{\emptyset}{\tau}{\tau'}$ and $\tau$ is a
   struct type\\
     Impossible\\
{\bf Subsubcase:} $\tau'$ is a struct type\\
    $\JUDGEEXPRG{\emptyset}{\emptyset}{\ASSRTet{\eee'}{\tau}}{\tau}$ \` by rule
    $\RULENAME{t-stupid}$\\
{\bf Subcase:}  $\JUDGEEXPRG{\emptyset}{\emptyset}{\eee}{\tau_\SSS}$ \\
        $\JUDGEEXPRG{\emptyset}{\emptyset}{\eee'}{\tau'}$ with
        $\IMPLSG{\emptyset}{\tau'}{\tau_\SSS}$ \` by i.h.\\
        $\tau' = \tau_\SSS$ \` by Lemma~\ref{lem:fgg_invsub}\\
$\JUDGEEXPRG{\emptyset}{\emptyset}{\ASSRTet{\eee'}{\tau}}{\tau}$
        \` by rule $\RULENAME{t-stupid}$\\
\end{tabbing}

\end{description}
\end{proof}

\subsection{Progress}
\label{app:fgg_prog}

We extend the definition of $\PANICe{\eee}$ straightforwardly from
FG to FGG.

\begin{lemma}[FGG Canonical Forms]\label{lem:fgg_canforms}
If $\eee$ is a value and $\JUDGEEXPRG{\emptyset}{\emptyset}{\eee}{\ttt}$ then
$\ttt = \tau_\SSS$, for some $\tau_\SSS$ and $\eee =
\STRUCTte{\tau_\SSS}{\ov{\vvv}}$, for some $\ov{\vvv}$.
\end{lemma}
\begin{proof}
Straightforward induction on typing.
\end{proof}

\begin{theorem}[FGG Progress]\label{thm:fgg_prog}
  If $\JUDGEEXPRG{\emptyset}{\emptyset}{\eee}{\tau}$ then either
  $\eee$ is a value, $\PANICe{\eee}$ or $\eee \REDUCE{} \eee'$.
  \end{theorem}
  \begin{proof}
The proof follows the same lines of Theorem~\ref{thm:progpanic}.
  \end{proof}

 \newpage
\subsection{Monomorphisability}
\label{app:nomono}

\begin{lemma}\label{lem:expr-omega-finite}
  If $\Delta \stoup \Gamma \vdash e : \tau$ and
  $\Delta \stoup \Gamma \vdash e \yields \omega$, then $\omega$ is
  finite.
\end{lemma}
\begin{proof}
  By straightforward induction on the rules from Figure~\ref{fig:fgg-omega-new}.
\end{proof}

Hereafter, we let $\eo$ range over elements of $\Omega$ and $\omega$, i.e.,
$\eo$ is of the form $\tau$ or $\tau.m(\psi)$.
Also, we assume that all formal type parameters are pairwise
distinct, without loss of generality.

Given $\eta_1 = (\ov{\alpha \by \tau})$, we write
$[\eta_1 \cdot \eta_2]$ for $(\ov{\alpha \by \tau [\eta_2]})$.
Also, we write $[\eta_1 \cdot \eta_2 \cdot \eta_3]$ for $[\eta_1 \cdot
[\eta_2 \cdot \eta_3]]$. 
Next, we extend naturally the occurs check from
Figure~\ref{fig:fgg-polyrec-check} as follows:
\begin{mathpar}
  \inferrule
  {
    \phi = \ov{\tau}
    \and
    \alpha_i \prec \tau_i 
  }
  { (\type~\ov{\alpha~\tau_I}) \prec \phi}

  \inferrule
  {
    \tau \neq \alpha
    \and
    \alpha \in \FV{\tau}
  }
  { \alpha \prec \tau } 

  \inferrule
  {
    \alpha_i \prec \tau_i 
  }
  { \ov{\alpha} \prec \ov{\tau}} 
\end{mathpar}

We say that a substitution
$(\ov{\alpha \by \tau})$ is
\goodsub\ if $\neg ( \ov{\alpha}
\prec \ov{\tau})$.

\begin{lemma}\label{lem:types-subs}
  Given $P = \ov{D} \prog d$, if $P \ok$ holds
  and $P \notmonomorphisable$ does not hold, then
  for each declaration
  $\func~(x~t_S(\Phi))~m(\Psi)N~\br{\return~e} \in \ov{D}$,
  posing $\Delta = 
  {\Phi, \Psi}$,
for all $n$ and
  $\eo \in G_{\Delta}^{n}(\{ t_S(\hat\Phi),
  t_S(\hat\Phi).m(\hat{\Psi}) \})$,
there are
  $\eta_1 \cdots \eta_n$ such that
$\eo = \eo_0 [\eta_1 \cdots \eta_n]$, such that for all
  $1 \leq i \leq n$, $[\eta_1 \cdots \eta_i]$ is a \goodsub\
  substitution.
\end{lemma}
\begin{proof}
By induction on $n$.
Pose $\omega = \{ t_S(\hat\Phi),
  t_S(\hat\Phi).m(\hat{\Psi}) \}$.

  Let $n =1$. 
If $\eo \in \omega$ then we have the result with the empty
  substitution.
If $\eo \in G(\omega) \setminus \omega$, then it has the form
  $\eo_0 [\eta] $ where $\eo_0$ occurs in the
  body of method $t_S(\hat\Phi).m(\hat{\Psi})$ with
$\eta = (\Phi, \Psi \by \Phi, \Psi)$, which is trivially
  \goodsub.

  Assume the result holds for $n$, let us show it holds for $n+1$.
By definition of $G$, we have 
  \[
    \begin{array}{l}
      G_\Delta^{n+1}(\omega) = 
      \\
      \quad
      G_\Delta^{n}(\omega) \cup 
      \FExtensionD{G_\Delta^{n} (\omega)}{\Delta} \cup
      \MExtensionD{G_\Delta^{n}(\omega) }{\Delta} \cup
      \IExtensionD{G_\Delta^{n} (\omega)}{\Delta} \cup
      \SExtensionD{G_\Delta^{n}(\omega) }{\Delta}
    \end{array}
  \] 
  By induction, we have that for all $\eo \in G_\Delta^{n}(\omega)$,
$\eo = \eo_0 [\eta_1 \cdots \eta_n]$ such that for all $1 \leq i
  \leq n$,  substitution $[\eta_1 \cdots \eta_i]$ is \goodsub.
  \begin{itemize}
  \item If $\eo$ is generated by $\Fclo$, then we have $\eo = \eo_0
    [\eta \cdot \eta_1 \cdots \eta_n ]$ where $\eo_0$ is a field of some $\tau \in
    G_\Delta^{n}(\omega)$.
The fact that all $[\eta \cdot \eta_1 \cdots \eta_i]$ are
    \goodsub\ follows by induction hypothesis and the assumption that
    structures are not recursive.
\item If $\eo$ is generated by $\Mclo$, then it is a type that
    occurs in the signature of some
    $\tau.m(\psi) \in G_\Delta^{n}(\omega)$, with
    $\tau.m(\psi) = \eo' [\eta_1 \cdots \eta_n ]$ by induction
    hypothesis, hence we have $\eo = \eo_0 [\eta_1 \cdots \eta_n ]$.
\item If $\eo$ is generated by $\Iclo$, then we have that $\tau'_I$
    and $\tau_I.m'(\psi)$ are in $G_\Delta^{n}(\omega)$.
By induction hypothesis, we have 
    $\tau'_I = \rho [ \eta_1 \cdots \eta_n]$
    and
    $\tau_I.m'(\psi) = \rho_2 [ \eta'_1 \cdots \eta'_n]$,
    thus we have $\eo = \rho_3 [(\eta_1, \eta'_1) \cdots
    (\eta_n,\eta'_n)]$.
All $[(\eta_1, \eta'_1) \cdots (\eta_i,\eta'_i)]$ are good by
    induction hypothesis and assumptions that type parameters are
    pairwise distinct.
\item If $\eo$ is generated by $\Sclo$, then $\eo = \eo_0 (\Phi',
    \Psi' \by \ov{ \tau [\eta_1 \cdots \eta_n]})$, with $\eo_0$ an
    instance occurring in the body of a method $t'_S(\Phi').m'(\Psi')$ and 
    $ \tau [\eta_1 \cdots \eta_n] \in G_\Delta^{n}(\omega)$.

    We show that for all $1 \leq i \leq n$, each
    $(\Phi', \Psi' \by \ov{ \tau [\eta_1 \cdots \eta_i]})$ is a \goodsub\
    substitution, by contradiction.
If the substitution is not \goodsub, then one of the type
    parameter in $\Phi', \Psi'$ must occur in  $\ov{ \tau [\eta_1
      \cdots \eta_i]}$. Since all type parameters are distinct, it
    means that we have visited $t'_S(\Phi').m'(\Psi')$ before. 
Hence all substitutions must have occurred in
    $G_{\Delta'}^{i}(\{ t'_S(\hat\Phi'), t'_S(\hat\Phi').m(\hat{\Psi'})
    \})$, with $\Delta' = \Psi', \Phi'$.
We know by induction ($i \leq n$) that all substitutions
    in this set are \goodsub, thus we have a contradiction.
\end{itemize}
\end{proof}

Give $\eta = (\ov{\alpha \by \tau})$ and
$\eta' = (\ov{\alpha \by \sigma})$, we write $[\eta] \sorder [\eta']$,
iff the number of solved variables in $\eta'$ is strictly greater than
the number of solved variables in $\eta$.
Variable $\alpha$ is \emph{solved} in $(\ov{\alpha \by \tau})$ if
$\alpha \in \ov{\alpha}$ and $\alpha \notin \FV{\ov{\tau}}$.
We say that $\eta$ is solved if all its variables are solved,
otherwise it is unsolved.

\begin{lemma}\label{lem:finite-gen}
  Given $P = \ov{D} \prog d$, if $P \ok$ holds
  and $P \notmonomorphisable$ does not hold, then
  for each declaration
  $\func~(x~t_S(\Phi))~m(\Psi)N~\br{\return~e} \in \ov{D}$,
  posing $\Delta = 
  {\Phi, \Psi}$,
there is
  $n$ (finite) such that 
  $G_{\Delta}^{n}(\{ t_S(\hat\Phi),
  t_S(\hat\Phi).m(\hat{\Psi}) \})
  =
  G_{\Delta}^{n+1}(\{ t_S(\hat\Phi),
  t_S(\hat\Phi).m(\hat{\Psi}) \})$.
\end{lemma}
\begin{proof}
  Clearly, we have
  $
  G_{\Delta}^{n}(\{ t_S(\hat\Phi),
  t_S(\hat\Phi).m(\hat{\Psi}) \})
  \subseteq
  G_{\Delta}^{n+1}(\{ t_S(\hat\Phi),
  t_S(\hat\Phi).m(\hat{\Psi}) \})$ for all $n\geq 1$.

  Using Lemma~\ref{lem:types-subs}, we know that every element in
  these sets (where $k$ is $n$ or $n+1$) is of the form
$ \eo [ \eta_1 \cdots \eta_k] $ such that $1\leq i \leq k$,
  $[ \eta_1 \cdots \eta_i]$ is good.
There are finitely many instance $\eo$ since they consists of a type
  name with finitely many type parameters, or a pair of type name with
  finitely many parameters and a method name, with finitely many type
  parameters.
The number of distinct $\eo$ is bounded by the number of
  declarations (methods and types) and the number of method signatures
  in interfaces.

  Because each $\eta_i$ is extracted from a syntactical occurrence of a
  method call or type instantiation, there are also finitely many
  substitutions $\eta_i$. The number of distinct substitutions
  $\eta_i$ is bounded by the size of the syntax of $P$.

  By Lemma~\ref{lem:types-subs}, the sequences of substitutions grow
  on the left. We show that there is a well-founded ordering on these
  sequences, when they are not permutations.
Note that they are only finitely many substitutions that are
  permutations.
First, every sequence of substitutions $[\eta_1 \cdots \eta_i \cdots
  \eta_n]$ with $\eta_i$ solved, can replaced by  $[\eta_1 \cdots
  \eta_i]$. 
  
Next, we show that for
  $\eo [\eta'] \in G_{\Delta}^{n}(\{ t_S(\hat\Phi),
  t_S(\hat\Phi).m(\hat{\Psi}) \}) $ and
  $\eo [\eta \cdot \eta'] \in G_{\Delta}^{n+1}(\{ t_S(\hat\Phi),
  t_S(\hat\Phi).m(\hat{\Psi}) \}) $, if
  $\eo [\eta'] \neq \eo [\eta \cdot \eta']$, then we have
  $[\eta'] \sorder [\eta \cdot \eta']$ (assuming we have shortened the
  sequences with solved substitutions as above).

  Clearly the number of solved variables in $[ \eta \cdot \eta']$ is
  at least the number of solved variables in $[\eta' ]$, indeed each
  ${\alpha_i}$ occurring in $\ov{\tau}$ is replaced by $\sigma_i$.
There must be $\alpha_i$ such that it $\FV{\sigma} \cap \ov{\alpha}
  = \emptyset$ otherwise there would be a cycle which would lead to a
  bad substitution.
Applying all substitutions in $\eta'$ then replacing $\alpha_i$ will
  solve it in $[\eta \cdot \eta']$ hence the number of solved variable
  is strictly increasing.

Hence, there is well-founded ordering over the elements generated by
  the limit of $G$, hence there exists a finite fixpoint.
\end{proof}

\begin{theorem}[Decidability]If\/ $P \ok$ then it is decidable whether or not\/
  $P \notmonomorphisable$ holds.
\end{theorem}
\begin{proof}
We construct a dovetailing algorithm that decides whether or not
  $P \notmonomorphisable$ holds as follows.
Given $P = \ov{D} \prog d$, we simultaneously check whether $D_i
  \notmonomorphisable$ holds for all $i$ iteratively (starting with
  $n=1$).
The algorithm terminates
  either ($i$) when
  $n$ is found such that 
  $G_{\Delta}^{n}(\omega \cup \{
  t_S(\hat\Phi) \})
  = 
  G_{\Delta}^{n{+}1}(\omega \cup \{
  t_S(\hat\Phi) \})
  $, i.e., a fixpoint has been found, for \emph{each} method declaration, or
  ($ii$)
  when there is a 
  method such that 
  there is $n$ 
with $(t_S(\phi).m(\psi)) \in G_{\Delta}^{n}(\omega \cup \{
  t_S(\hat\Phi) \})
~\text{s.t.}~
  \Phi \prec  \phi  ~\text{or}~ \Psi \prec \psi$ (i.e., the occurs
  check fails).

  This algorithm terminates if ($i$) all declaration checks reach
  a fixpoint or ($ii$) if at least one declaration fails the occurs check.
By Lemma~\ref{lem:finite-gen} we know that either ($i$)
  or ($ii$) will eventually be satisfied, hence the algorithm
  always terminates.
\end{proof}

\begin{theorem}[Monomorphisability]If\/ $P \ok$
  and\/ $P \notmonomorphisable$ doesn't hold
  then\/ $P \yields \Omega$
  with\/ $\Omega$ finite.
\end{theorem}
\begin{proof}
  Direct consequence from Lemma~\ref{lem:finite-gen}, considering the
  main function as a special case of a method declaration (with no
  type parameter). 
\end{proof}

\newpage
\newcommand{\lemref}[1]{Lemma~\ref{#1}}
\newcommand{\NYRULENAME}[1]{[\textsc{#1}]}
\newcommand{\thmref}[1]{Theorem~\ref{#1}}
\newcommand{\defref}[1]{Definition~\ref{#1}}
\newcommand{\secref}[1]{Section~\ref{#1}}
\newcommand{\appref}[1]{Appendix~\ref{#1}}

\newcommand{\COMPAT}[3]{\ensuremath{#1 \asymp_{#3} #2 }}

\section{Proof of Theorem~\ref{thm:main:monotsound}}
\begin{lemma}\label{lem:recv-meth-omega}
  If $\Delta \stoup \Gamma \vdash e : \tau$,
  $\Delta \stoup \Gamma \vdash e \yields \omega$, and $\sigma.m(\phi)
  \in \omega$, then $\sigma \in \omega$.
\end{lemma}
\begin{proof}
  By straightforward induction on the rules from Figure~\ref{fig:fgg-omega-new}.
\end{proof}

\begin{lemma}
  \label{lem:app_inomegaexp}
  Let $\emptyset ; \emptyset \vdash e : \tau$,
  $\emptyset \vdash e \yields \omega$ and $\Omega = \lim_{n
    \rightarrow \infty}  G_{\emptyset}^{n}(\omega)$. Then:
  \begin{itemize}
   \item $\tau \in \Omega$;
   \item If $\tau_S\br{\ov{e'}}$ is a subexpression of $e$ then $\tau_S \in
    \Omega$ and  $\ov{\tau} \in \Omega$, with $(\ov{f~\tau}) =
    \fields(\tau_S)$;
  \item If $e'.(\sigma)$ is a subexpression of $e$ then $\sigma \in \Omega$;
  \item If $e'.m(\psi)(\ov{e})$ is a subexpression of $e$ with
    $\emptyset\vdash e' : \sigma$ then
    $\{\sigma, \sigma.m(\psi)\} \subset \Omega$.
  \end{itemize}
\end{lemma}
\begin{proof}
Straightforward by the definitions for computation of instance sets
and $G$.
\end{proof}

\begin{lemma}
  \label{lem:app_exptypeinomega}
Let $\Delta ; x : t_S(\hat{\Phi}) \comma \ov{x : \tau} \vdash e : \tau$ for some
$\func~(x~t_S(\Phi))~m(\Psi)(\ov{x~\tau})~\sigma~\br{\return~e} \ok$,
$\Phi \stoup \Psi \ok~\Delta $ and $\Delta \vdash \tau \IMPLOP
\sigma$. If $\Delta ; x : t_S(\hat{\Phi}) \comma \ov{x : \tau} \vdash
e : \tau$ if $\theta = (\Delta \by_\Delta \phi)$, for some
$\phi$, $\emptyset ; x : t_S(\hat{\Phi})[\theta] \comma \ov{x : \tau[\theta]} \vdash
e \yields \omega$, and $\Omega = \lim_{n
    \rightarrow \infty}  G_{\emptyset}^{n}(\omega \cup
  \{t_S(\hat{\Phi})[\theta] , \ov{\tau[\theta]} , t_S(\hat{\Phi})[\theta].m(\hat{\Phi}[\theta])\})$ then $\tau[\theta] \in \Omega$.
\end{lemma}
\begin{proof}
By induction definitions for computation of instance sets,
$G$, $\Fclo$ and $\Mclo$.
\end{proof}

\begin{lemma}
   \label{lem:app_exptypeinomega2}
If $\emptyset ; \ov{x :\sigma} \vdash e : \tau$, $\emptyset;\ov{x :\sigma} \vdash e \yields
\omega$ and $\Omega = \lim_{n \rightarrow \infty}
G_{\emptyset}^{n}(\omega' )$, with $\cup \{\ov{\sigma}\} \subset
\omega'$ and $\omega \subset \omega'$ then $\tau \in \Omega$.
\end{lemma}
\begin{proof}
By induction definitions for computation of instance sets,
$G$, $\Fclo$ and $\Mclo$.
\end{proof}

\begin{lemma}
  \label{lem:app_substcommute}
  If $\theta \vdash \tau \mapsto t^\dagger$ then $\vdash \tau[\theta] \mapsto
  t^\dagger$. If $\theta \vdash e \mapsto e^\dagger$ then $\vdash
  e[\theta] \mapsto e^\dagger$.
\end{lemma}
\begin{proof}
Straightforward induction on the definition of monomorphisation of
expressions and type names.
\end{proof}

\begin{lemma}[Monomorphisation preserves subtyping]
\label{lem:mono_subt}
Let $P\ok$, $P \yields \Omega$,
$\Delta \vdash \tau,\sigma \ok$ and $\theta = (\Delta \by_\Delta \phi
)$, for some $\phi$,  $\sigma[\theta],\tau[\theta]\in \Omega$
and $\theta \vdash \tau \mapsto t^\dagger$ 
If $\Delta \vdash \tau \IMPLOP \sigma$ then $t^\dagger \IMPLOP u^\dagger$, with
$\theta \vdash \sigma \mapsto u^\dagger$.
\end{lemma}
\begin{proof}
  We proceed by case analysis on $\Delta \vdash \tau \IMPLOP \sigma$. If
  the derivation holds by rule {\sc $\IMPLOP$-param}, then $\tau =
  \sigma = \alpha$ and we have that $\alpha[\theta] = \tau'$, for some
  $\tau'$. Since $\tau' \in \Omega$ and $\theta \vdash \tau \mapsto
  t^\dagger$ we have that $\theta \vdash \tau' \mapsto t^\dagger$ and $t^\dagger\ok$.
  We conclude by reflexivity of the FG implements relation.

  If the derivation holds by rule {\sc $\IMPLOP_S$}, then $\tau =
  \sigma = \tau_S$. Since $\tau_S[\theta] \in \Omega$ and $\theta
  \vdash \tau_S \mapsto t^\dagger$ we have that $t^\dagger\ok$ and
  conclude by the FG implements  {\sc $\IMPLOP_S$} rule.

  If the derivation holds by rule {\sc $\IMPLOP_I$}, then
  $\sigma = \tau_I$, for some $\tau_I$, with
  $\methods_\Delta(\tau) \supseteq \methods_\Delta(\tau_I)$.  Since
  $\theta \vdash \tau \mapsto t^\dagger$ and $\tau[\theta]\in\Omega$
  then $t^\dagger\ok$. Moreover, since since $\tau_I[\theta]\in\Omega$
  then there exists some $u^\dagger$ such that
  $\theta \vdash \tau_I \mapsto u^\dagger$ and $u^\dagger\ok$.

  We now
  proceed by a case analysis on the set of methods
  $M = \{\tau_I[\theta].m(\psi) \mid \tau_I[\theta].m(\psi) \in \Omega\}$. If this set is empty, then in the
  monomorphisation of $P$ ($P^\dagger)$, the type declaration for
  $u^\dagger$ contains only the methods generated by the second
  premise of rule {\sc m-spec}. Since $\tau[\theta]\in\Omega$, either
  $\tau$ is a struct type, and then we know that all its methods have
  a corresponding ``dummy'' analogue in $P^\dagger$ and since
  $\methods_\Delta(\tau) \supseteq \methods_\Delta(\tau_I)$ then it
  must be the case that
  $\methods(t^\dagger)\supseteq \methods(u^\dagger)$ and we conclude
  by the FG implements  {\sc $\IMPLOP_I$} rule. If $\tau$ is an
  interface type,  the type declaration for
  $t^\dagger$ contains at least the methods generated by the second
  premise of rule {\sc m-spec}, and since $\methods_\Delta(\tau)
  \supseteq \methods_\Delta(\tau_I)$ it follows that
  $\methods(t^\dagger)\supseteq \methods(u^\dagger)$  and we conclude
  by the FG implements  {\sc $\IMPLOP_I$} rule.
  
  If the set $M$ is not empty, then the declaration for $u^\dagger$
  contains the dummy version of all methods of $\tau_I[\theta]$ and the
  monomorphisations of those in $M$ (by rules {\sc m-type}, {\sc
    m-interface}, and {\sc m-spec}). If $\tau$ is an interface type,
  since $\tau[\theta]\in\Omega$, $M\subseteq \Omega$ and
  $\emptyset \vdash \tau[\theta] \IMPLOP \tau_I[\theta]$ by
  (Lemma~\ref{lem:fgg_tsubstsub}) then by $\Iclo$ each $\tau_I[\theta].m(\psi) \in
  \Omega$ has a corresponding $\tau_I[\theta]'.m(\psi) \in \Omega$ and
  so the type declaration for $t^\dagger$ contains both the methods generated by the second
  premise of rule {\sc m-spec} and all those corresponding to the
  monomorphisations of the methods of $\tau_I[\theta]$ in $M$, and
  thus $\methods(t^\dagger)\supseteq \methods(u^\dagger)$ and we conclude
  by the FG implements  {\sc $\IMPLOP_I$} rule.
  Finally, if $\tau$ is a struct type, we know that all its methods have
  a corresponding ``dummy'' analogue in $P^\dagger$ for
  $t^\dagger$. Moreover, Since $\emptyset \vdash \tau[\theta] \IMPLOP \tau_I[\theta]$ by
  (Lemma~\ref{lem:fgg_tsubstsub}), $\tau[\theta] \in \Omega$ and
  $M\subset \Omega$, by $\Sclo$ each $\tau_I[\theta].m(\psi) \in
  \Omega$ has a corresponding $\tau[\theta].m(\psi) \in \Omega$ and
  so, since  $\methods_\Delta(\tau) \supseteq \methods_\Delta(\tau_I)$
  we have that by rule {\sc m-func}, $\methods(t^\dagger)\supseteq
  \methods(u^\dagger)$ and we conclude
  by the FG implements  {\sc $\IMPLOP_I$} rule.

\end{proof}

\begin{lemma}
\label{lem:app_monodeclexists}
If $P\ok$ with $P = \ov{D} \prog d$ and $\vdash P \mapsto P^\dagger$
  with $P^\dagger = \ov{D^\dagger} \prog d^\dagger$, $P\yields
  \Omega$, $\emptyset \vdash \tau\ok$ and $\tau\in\Omega$ then if $\vdash \tau \mapsto t^\dagger$
  then $\type~t^\dagger~T^\dagger \in \ov{D^\dagger}$, for some $T^\dagger$.
\end{lemma}
\begin{proof}
Straightforward by definition of monomorphisation.
\end{proof}

\begin{lemma}[Monomorphisation preserves well-formedness of type declarations]
  \label{lem:app_monotdecl}
  If $P\ok$ with $P = \ov{D} \prog d$ and $\vdash P \mapsto P^\dagger$
  with $P^\dagger = \ov{D^\dagger} \prog d^\dagger$ then:
If $\type~t^\dagger~T^\dagger \in \ov{D^\dagger}$ then $\type
  ~t^\dagger~T^\dagger\ok$.

\end{lemma}

\begin{proof}
  By inversion on monomorphisation we have that $P \yields \Omega$,
  $\Omega \vdash \ov{D\mapsto \mathcal{D}}$,
    $\ov{D^\dagger} = \{ \type~\dummytype~\struct~\br{} \}  \cup  \bigcup\ov{\calD}$
      and 
      $\emptyset \vdash d \mapsto d^\dagger$.

If $\type~t^\dagger~T^\dagger
 =\type~\dummytype~\struct~\br{}$ then
 $\type~\dummytype~\struct~\br{}\ok$ is immediate. Otherwise, we have
 that $\type~t^\dagger~T^\dagger$ is such that
$\type~t(\Phi)~T \in D$ and there exists some $t(\phi) \in \Omega$ where $\eta ; \mu \vdash T
\mapsto T^\dagger$, with $\eta = (\Phi \by \phi)$ and
$\mu = \set{ m(\psi) \mid t(\phi).m(\psi) \in \Omega }$.

If $T^\dagger$ is of the form $\struct\br{\ov{f~t^\dagger}}$, by
        inversion we have that $\eta \vdash \ov{\tau \mapsto
          t^\dagger}$. Since $t(\phi) \in \Omega$, by $\Fclo$ we have
        that $\ov{\tau[\eta]} \in \Omega$. We show that 
        can thus establish that $\type~t^\dagger~T^\dagger\ok$ by showing
        $\distinct(\ov{f})$ (since $\type~t(\Phi)~T \in D$ and $P\ok$)
        and by showing that $\ov{t^\dagger\ok}$ (since
        $\ov{\tau[\eta]} \in \Omega$, $\eta \vdash \ov{\tau \mapsto
          t^\dagger}$ and $\Omega \vdash \ov{D\mapsto \mathcal{D}}$).

 If $T^\dagger$ is of the form $\interface\br{\bigcup \ov{\calS}}$, we
 know that $T = \interface\br{S}$ and $\eta ; \mu \vdash \ov{S\mapsto
   \mathcal{S}}$. For method signatures in $\bigcup \ov{\calS}$
 arising from rule {\sc m-id}, well-formedness is immediate. For
 the rest, distinctness follows from $P\ok$ and well-formedness
 follows from $\Mclo$, $t(\phi).m(\psi) \in \Omega$ and $\Omega \vdash
 \ov{D\mapsto \mathcal{D}}$. 
 
\end{proof}

\begin{lemma}[Monomorphisation preserves typing of expressions]
  \label{lem:app_monoprestypingopenexpr}
  If $P \ok$ and $P \yields \Omega$ and 
  $\emptyset ; \ov{x : \tau} \vdash e : \sigma$ and
  $\emptyset ; \ov{x : \tau} \vdash e \yields \omega$ and
  $\omega \subset \Omega$,
  $\ov{\tau} \in \Omega$ and $\vdash \ov{\tau \mapsto t^\dagger}$
  and $\vdash e \mapsto e^\dagger$ and $\vdash \sigma \mapsto u^\dagger$
  then  $\emptyset ; \ov{x : t^\dagger} \vdash e : u^\dagger$
\end{lemma}
\begin{proof}
  By induction on the derivation of
  $\emptyset ; \ov{x : \tau} \vdash e : \sigma$ with case analysis on
  the last rule and a further case analysis on monomorphisation and instance
  generation. In the sequel we write $\Gamma^\dagger$ for the FG typing
  context $\ov{x : t^\dagger}$ and $\Gamma$ for the corresponding FGG
  typing context.

  \begin{description}

 \item[Case:] $\RULENAME{t-var}$

Since $(x:\sigma) \in \Gamma$ then $(x:u^\dagger)$ in $\Gamma^\dagger$
and thus we conclude by FG rule $\RULENAME{t-var}$.

\item[Case:] $\RULENAME{t-call}$ ($e$ is $e_0.m(\psi)(\ov{e})$) 

  By inversion we have that:
  $\emptyset ; \Gamma \vdash e_0 : \tau_0$, $\emptyset ; \Gamma
  \vdash \ov{e : \tau'}$, $(m(\Psi)(\ov{x~\sigma})~\sigma) \in
  \methods_\emptyset(\tau_0)$ and $\eta = (\Psi \by_\emptyset \psi)$
  $\emptyset \vdash (\ov{\tau' \imp \sigma})[\eta]$.
 Since $\omega \subset \Omega$ and $\ov{\tau} \in \Omega$ it follows
by Lemmas~\ref{lem:app_exptypeinomega} and~\ref{lem:app_inomegaexp}
that $\tau_0.m(\psi) \in \Omega$ and $\tau_0 \in \Omega$. Thus, since
$(m(\Psi)(\ov{x~\sigma})~\sigma) \in \methods_\emptyset(\tau_0)$ then
by $\Mclo$ it follows that $\ov{\sigma [\eta]} \in \Omega$ and $\sigma[\eta] \in \Omega$.
Let
$\emptyset \vdash \tau_0 \mapsto t_0^\dagger$, 
$\emptyset \vdash \ov{\tau' \mapsto t'^\dagger}$,
$\vdash e_0 \mapsto e_0^\dagger$ and $\vdash \ov{e\mapsto e^\dagger}$,
for some $t_0^\dagger, t^\dagger, e_0^\dagger$ and $\ov{e^\dagger}$.
By i.h. we have that
$\Gamma^\dagger \vdash e_0^\dagger : t_0^\dagger$ and $\Gamma^\dagger
\vdash \ov{e^\dagger : t'^\dagger}$. Since $\tau_0.m(\psi) \in \Omega$
then $(m^\dagger(\ov{x~u'^\dagger})~u^\dagger) \in
\methods(t_0^\dagger)$,
where $\emptyset\vdash \ov{\sigma[\eta] \mapsto u'^\dagger}$ and $\emptyset\vdash
\sigma[\eta] \mapsto u^\dagger$ and $\vdash m \mapsto m^\dagger$.
From Lemma~\ref{lem:app_exptypeinomega2} it follows that $\ov{\tau'}
\in \Omega$ and so by Lemma~\ref{lem:mono_subt}, we know that
$\ov{t' \imp u'^\dagger}$.
By FG typing rule {\sc t-call} it follows that $\Gamma^\dagger \vdash
e_0^\dagger.m^\dagger(\ov{e^\dagger}) : u^\dagger$ and so we conclude
this case.

\item[Case:] $\RULENAME{t-literal}$ ($e$ is $\tau_S\br{\ov{e}}$)

  Let $\vdash \tau_S \mapsto t^\dagger$
By inversion we have $\emptyset \vdash \tau_S \ok$, 
$\emptyset \stoup \Gamma \vdash \ov{e : \tau}$,
 $(\ov{f~\sigma}) = \fields(\tau_S)$
and 
$\emptyset \vdash \ov{\tau \imp \sigma}$.
Let $\vdash \ov{e\mapsto e^\dagger}$ and $\vdash \ov{\tau \mapsto
  t'^\dagger}$, for some $e^\dagger$ and $t'^\dagger$. By i.h. it
follows
that $\Gamma^\dagger \vdash \ov{e^\dagger : t'^\dagger}$. By
Lemma~\ref{lem:app_exptypeinomega2} we have that $\tau_S \in \Omega$,
$\ov{\sigma} \in \Omega$ and $\ov{\tau}\in\Omega$.
Let $\vdash \ov{\sigma \mapsto u^\dagger}$, for some
$\ov{u^\dagger}$. We have that $(\ov{f~u^\dagger}) =
\fields(t^\dagger)$ by Lemmas~\ref{lem:app_monodeclexists}
and~\ref{lem:app_monotdecl}. 
By Lemma~\ref{lem:mono_subt} it follows that $\ov{t'^\dagger \imp
  u^\dagger}$ and so by FG typing rule $\RULENAME{t-literal}$ we
conclude this case.

\item[Case:] $\RULENAME{t-field}$ ($e$ is $e_0.f_i$)

  Let $\vdash \tau_i \mapsto t_i^\dagger$. By inversion we have
  $\Gamma \vdash e : \tau_S$ and $(\ov{f~\tau}) = \fields(\tau_S)$.
  By Lemma~\ref{lem:app_exptypeinomega2} we have that $\tau_S \in
  \Omega$ and by $\Fclo$ we have that $\ov{\tau}\in\Omega$.
  Let $\vdash \tau_S \mapsto t_S^\dagger$, $\vdash \ov{\tau \mapsto
    t^\dagger}$ and $e \mapsto e^\dagger$ for some $t_S^\dagger,
  \ov{t^\dagger}$ and $e^\dagger$. By i.h. it follows that
  $\Gamma^\dagger \vdash e^\dagger : t_S^\dagger$. Since $\tau_S \in
  \Omega$ and by Lemmas~\ref{lem:app_monodeclexists}
  and~\ref{lem:app_monotdecl} we know that $(\ov{f~t^\dagger}) = \fields(t_S)$ and we
  conclude by FG typing rule $\RULENAME{t-field}$.

\item[Case:] $\RULENAME{t-assert}_I$ ($e$ is $e_0.(\tau_J)$)

  Let $\vdash \tau_J \mapsto t_J^\dagger$ and $e_0 \mapsto
  e_0^\dagger$. By inversion we have that
  $\emptyset \vdash \tau_J\ok$ and $\Gamma \vdash e : \sigma_J$ 
We know that $\tau_J \in \Omega$ and by
Lemma~\ref{lem:app_exptypeinomega2} we have that $\sigma_J \in
\Omega$.  Let $\vdash \sigma_j \mapsto u^\dagger$, by i.h. we have
that $\Gamma^\dagger\vdash e_0^\dagger : u^\dagger$. Noting $\tau_J \in
\Omega$, by Lemmas~\ref{lem:app_monodeclexists} and~\ref{lem:app_monotdecl} we know $t_J^\dagger\ok$ and thus by FG typing rule $\RULENAME{t-assert}_I$
we have $\Gamma^\dagger \vdash e_0^\dagger.(t_J^\dagger) :
t_J^\dagger$, which concludes this case.

\item[Case:] $\RULENAME{t-assert}_S$ ($e$ is $e_0.(\tau_S)$)

  Let $\vdash \tau_S \mapsto t_S^\dagger$ and
  $\vdash e_0 \mapsto e_0^\dagger$. By inversion we know that
  $\emptyset \vdash \tau_S\ok$,
  $\emptyset ; \Gamma \vdash e_0 : \sigma_J$ and
  $\emptyset \vdash \tau_S \imp \bounds_\emptyset(\sigma_J)$.  We know that
  $\tau_S \in \Omega$ and by Lemma~\ref{lem:app_exptypeinomega2} we
  have that $\sigma_J \in \Omega$. Let
  $\vdash \sigma_J \mapsto u^\dagger$, by i.h. we have that
  $\Gamma^\dagger\vdash e_0^\dagger : u^\dagger$.
  Since $\tau_S\in \Omega$ and $\emptyset \vdash \tau_S\ok$, by
  Lemmas~\ref{lem:app_monodeclexists} and~\ref{lem:app_monotdecl}
  we have that $t_S^\dagger\ok$. Since $\sigma_J$ is closed, by
  definition $\bounds_\emptyset(\sigma_J) = \sigma_J$ and so
  $\emptyset \vdash \tau_S
  \imp \sigma_J$.  By Lemma~\ref{lem:mono_subt}, $t_S \imp u^\dagger$
  and by FG typing rule $\RULENAME{t-assert}_S$ we conclude
  that $\Gamma^\dagger \vdash e_0^\dagger.(t_S^\dagger) :
  t_S^\dagger$, which concludes this case.

 \item[Case:] $\RULENAME{t-stupid}$ ($e$ is $e_0.(\tau)$)

Let $\vdash \tau \mapsto t^\dagger$ and $\vdash e_0 \mapsto e_0^\dagger$. By inversion we have that
$\emptyset \vdash \tau \ok$ and $\emptyset;\Gamma \vdash e_0 :
\sigma_S$. We know that $\tau \in \Omega$ and by
Lemma~\ref{lem:app_exptypeinomega2} we have that $\sigma_S \in
\Omega$. Let $\vdash \sigma_S \mapsto u^\dagger$. By i.h. it follows that
$\Gamma^\dagger \vdash e_0^\dagger : u^\dagger$. Since  $\emptyset
\vdash \tau \ok$ and $\tau \in \Omega$ then
Lemmas~\ref{lem:app_monodeclexists} and~\ref{lem:app_monotdecl} give us $t^\dagger\ok$ and so by
FG typing rule $\RULENAME{t-stupid}$ we have that $\Gamma^\dagger
\vdash e_0^\dagger.(t^\dagger) : t^\dagger$, which concludes this
case. 
  
  \end{description}

\end{proof}

\begin{lemma}[Monomorphisation preserves well-formedness of method declarations]
  \label{lem:app_monomethdecls}
  If $P\ok$ with $P = \ov{D} \prog d$ and $\vdash P \mapsto P^\dagger$
  with $P^\dagger = \ov{D^\dagger} \prog d^\dagger$ then:
If $\func~(x~t_S^\dagger)~m^\dagger
  N^\dagger~\br{\return~e^\dagger} \in \ov{D^\dagger}$ then
  $\func~(x~t_S^\dagger)~m^\dagger N^\dagger~\br{\return~e^\dagger} \ok$. 
\end{lemma}
\begin{proof}
  By inversion on monomorphisation we have that $P \yields \Omega$,
  $\Omega \vdash \ov{D\mapsto \mathcal{D}}$,
    $\ov{D^\dagger} = \{ \type~\dummytype~\struct~\br{} \}  \cup  \bigcup\ov{\calD}$
      and 
      $\emptyset \vdash d \mapsto d^\dagger$.

   Since $\func~(x~t_S^\dagger)~m^\dagger
  N^\dagger~\br{\return~e^\dagger} \in \ov{D^\dagger}$ then by
  inversion we know that the method definition arises from the first
  premise of {\sc m-func} or the second. If the latter, then
  $m^\dagger N^\dagger = S^\dagger$ for some $S^\dagger$ and
  $e^\dagger = \dummytype\br{}$, with
  $\func~(x~t_S(\Phi))~m(\Psi)N~\br{\return~e} \in D$, 
  $t_S(\phi) \in \Omega$, $\eta \vdash t_S(\Phi) \mapsto t_S^\dagger$
  and $\eta \vdash m(\Psi)N \mapsto S^\dagger$, with $\eta = (\Phi
  \by \phi)$. We thus have that $S^\dagger = m^\ast()~{\dummytype}$ by
  inversion on $\eta \vdash m(\Psi)N \mapsto S^\dagger$. Since
  $t_S(\phi) \in \Omega$ then $t_S^\dagger$ is declared in $D^\dagger$
  and well-formedness follows immediately.

  If the former then we know that:
  $\func~(x~t_S(\Phi))~m(\Psi)N~\br{\return~e} \in D$, 
$t_S(\phi).m(\psi) \in \Omega$, $\theta \vdash t_S(\Phi) \mapsto
t_S^\dagger$, $\theta \vdash m(\Psi) \mapsto m^\dagger$,
$\theta \vdash N \mapsto N^\dagger$ and $\theta \vdash  e \mapsto
e^\dagger$, with $\theta = (\Phi \by \phi, \Psi \by \psi)$. We know
that $N = (\ov{x~\tau})~\sigma$, $N^\dagger = (\ov{x~t^\dagger})~u^\dagger$, for some
$\ov{t^\dagger}$ and $u^\dagger$, with $\theta \vdash \ov{\tau \mapsto
  t^\dagger}$ and $\theta \vdash \sigma \mapsto u^\dagger$.

We must show that $t_S^\dagger\ok$, $\ov{t^\dagger}\ok$,
$u^\dagger\ok$, $\distinct(x,\ov{x})$ and
$x : t_S^\dagger \comma \ov{x : t^\dagger} \vdash e^\dagger : t$ with $t \imp u^\dagger$.
Since $P\ok$ and $\func~(x~t_S(\Phi))~m(\Psi)N~\br{\return~e} \in D$
then we have that $\distinct(x,\ov{x})$. Since $t_S(\phi).m(\psi) \in
\Omega$ then by Lemma~\ref{lem:recv-meth-omega}, $t_S(\phi) \in
\Omega$, and so by Lemma~\ref{lem:app_monotdecl} it follows that $t_S\ok$.
 Similarly, by $\Mclo$ we know that $\ov{\tau[\theta]}\in\Omega$ and
 $\sigma[\theta] \in\Omega$ and thus by Lemma~\ref{lem:app_monotdecl}  
 we have that $\ov{t^\dagger}\ok$ and $u^\dagger\ok$.
 
 Since $\func~(x~t_S(\Phi))~m(\Psi)N~\br{\return~e} \in D$ we have
 that $\Delta; x{:}t_S(\hat{\Phi}), \ov{x : \tau} \vdash e : \tau$
 with $\Delta \vdash \tau \IMPLOP \sigma$. By
 Lemma~\ref{lem:fgg_tsubstsub} we know that $\emptyset \vdash
 \tau[\theta] \IMPLOP \sigma[\theta]$.
By Lemma~\ref{lem:fgg_tsubsttyp} we have that 
 $\emptyset ; x{:}t_S(\phi), \ov{x : \tau[\theta]} \vdash e[\theta] :
 \tau[\theta]$. Since $\theta \vdash e \mapsto e^\dagger$ then $\vdash
 e[\theta] \mapsto e^\dagger$. Similarly, $\theta \vdash t_S(\Phi)
 \mapsto t_S^\dagger$ implies that $\vdash t_S(\phi) \mapsto
 t_S^\dagger$; $\theta \vdash \ov{\tau \mapsto
  t^\dagger}$ implies $\vdash \ov{\tau[\theta] \mapsto
  t^\dagger}$; and $\theta \vdash \sigma \mapsto u^\dagger$ implies
$\vdash \sigma[\theta] \mapsto u^\dagger$
(Lemma~\ref{lem:app_substcommute}). By
Lemma~\ref{lem:app_exptypeinomega} we know that $\tau[\theta] \in
\Omega$. Let $\vdash \tau[\theta]\mapsto t$, for some $t$, then
 by Lemma~\ref{lem:app_monoprestypingopenexpr} we know that 
 $x : t_S^\dagger \comma \ov{x : t^\dagger} \vdash e^\dagger : t$ and
 by Lemma~\ref{lem:mono_subt} it follows that $t\IMPLOP u^\dagger$,
 which concludes the proof.
\end{proof}

\begin{theorem}[Monomorphisation preserves program well-formedness]
  If $P\ok$ and $\vdash P \mapsto P^\dagger$
then $P^\dagger\ok$.
\end{theorem}
\begin{proof}
  Since $P = \ov{D} \prog d$ we have that $P^\dagger = \ov{D^\dagger} \prog d^\dagger$, with $P
  \yields \Omega$, for some $\Omega$. Since $P\ok$ we know that
  $\distinct(\tdecls(\ov{D}))$ and $\distinct(\mdecls(\ov{D}))$.
  By Lemmas~\ref{lem:app_monomethdecls} and~\ref{lem:app_monotdecl} it
  follows that $\ov{D^\dagger\ok}$. We note that
that   $\distinct(\tdecls(\ov{D}^\dagger))$ and
$\distinct(\mdecls(\ov{D}^\dagger))$ follows straightforwardly from
the definitions of monomorphisation for names.
  Since $\emptyset ;\emptyset \vdash d : \tau$,
  for some $\tau$, then by Lemma~\ref{lem:app_monoprestypingopenexpr}
  we have that $\emptyset ;\emptyset \vdash d^\dagger : t$, with
  $\vdash \tau \mapsto t$, and so $P^\dagger\ok$.
 \end{proof}

\section{Proofs of Theorem \ref{thm:main:correspondence}}
\label{app:op}
\begin{theorem}[Monomorphisation reflects subtyping]
  \label{thm:mono_refl_subt}
  Let $P\ok$ with $P \yields \Omega$, $\Delta \vdash \tau_1,\tau_2 \ok$,
  $t^\dagger_1 \ok$, and $t^\dagger_2 \ok$. Let $\eta = (\Delta
  \by_\Delta\psi)$, for some $\psi$, with $\tau_1[\eta] \in \Omega$
  and $\tau_2[\eta] \in \Omega$.
If $\IMPLS{t_1^\dagger}{t_2^\dagger}$ with
$\eta \vdash \tau_1 \mapsto t_1^\dagger$ and
$\eta \vdash \tau_2 \mapsto t_2^\dagger$ 
then
$\IMPLS{\tau_1}{\tau_2}$.
\end{theorem}
\begin{proof} 
We proceed by cases on the derivation of
$\IMPLS{t_1^\dagger}{t_2^\dagger}$. We use the term dummy methods to
refer to methods whose signature is generated by rule {\sc m-id}.

If the derivation holds from rule
$\IMPLOP_S$ then  $t_1^\dagger = t_2^\dagger = t_S$ for some $t_S$. By
definition of monomorphisation, since $\eta \vdash \tau_1 \mapsto t_S$
and $\eta \vdash \tau_2 \mapsto t_S$ we have that $\tau_1 = \tau_2 =
\tau_S$, for some $\tau_S$, and we conclude by the FGG implements
$\IMPLOP_S$ rule.

If the derivation holds from rule $\IMPLOP_I$ then we have that
$t_2^\dagger = t_I$, for some $t_I$, and
$\methods(t_1^\dagger) \supseteq \methods(t_I)$. By the definition of
monomorphisation, we know that $\tau_2 = \tau_I$, for some
$\tau_I$. Since 
$\tau_1[\eta] \in \Omega$ and $\tau_2[\eta] \in \Omega$, and  
$\methods(t_I)$ contains at least as many (potentially dummy) elements as
$\methods_\Delta(\tau_I)$, which are also in $\methods(t_1^\dagger)$.
Since $\tau_2[\eta] \in \Omega$ then rules {\sc m-type}, {\sc
  m-interface} and {\sc m-spec} are such that all dummy methods of
$t_I$ arise from all methods of $\tau_2$. If $\tau_1$ is an interface type,
by a similar reasoning we have that all dummy methods of $t_1$, which
contain all those of $t_2$, arise from methods of $\tau_1$, and thus
$\methods_\Delta(\tau_1) \supseteq \methods_\Delta(\tau_2)$, and so
$\Delta \vdash \tau_1 \IMPLOP \tau_2$, by FGG implements rule $\IMPLOP_I$. If
$\tau_1$ is a struct type, then by rule {\sc m-func}, since
$\tau_1[\eta] \in \Omega$, we know that all dummy method
implementations of $t_1$ (which include all methods declared for $t_2$) map from
all the method declarations of $\tau_1$ and so
$\methods_\Delta(\tau_1) \supseteq \methods_\Delta(\tau_2)$, and so
$\Delta \vdash \tau_1 \IMPLOP \tau_2$, by FGG implements rule
$\IMPLOP_I$. 

\end{proof}


\noindent{\bf\emph{Note:}}  
In this section, we often omit $\emptyset$, e.g., we write
$\tau <: \tau'$ for $\emptyset \vdash \tau <: \tau'$.

\begin{lemma}
\label{lem:opcorr_methbodies}
Suppose $\eta \vdash m(\psi) \mapsto m^\dagger$,
$\mbody(\tau.m(\psi)) = (x,\ov{x}).e$,
$m^\dagger = \monoid{m(\psi[\eta])}$, 
$\eta \vdash e \mapsto e^\dagger$,
and  $\eta \vdash \tau \mapsto t^\dagger$ with $t^\dagger=\monoid{\tau[\eta]}$. 
Then we have $\mbody(t^\dagger.m^\dagger) = (x,\ov{x}).e^\dagger$
\end{lemma}
\begin{proof}
Straightforward by induction on the derivations of the monomorphisation definition.
\end{proof}

\begin{lemma}[Compositionality]
\label{lem:mono_opcorr_comp}
Suppose $\JUDGEEXPRG{\emptyset}{\ov{x : \tau}}{e}{\tau}$, 
$\emptyset ; \emptyset \vdash {\ov{v:\tau'}}$ and 
$\ov{\IMPLS{\tau'}{\tau}}$. 

Assume 
$\eta \vdash e \mapsto e^\dagger$ and
$\eta \vdash \ov{v \mapsto v^\dagger}$.
Then we have 
$\eta \vdash e\ov{\SUBS{v}{x}} \mapsto
e^\dagger\ov{\SUBS{v^\dagger}{x}}$.
\end{lemma}
\begin{proof}
Mechanical by investigating the type derivation of
$\JUDGEEXPRG{\emptyset}{\ov{x : \tau}}{e}{\tau}$ with 
\thmref{thm:main:monotsound} and \lemref{lem:mono_subt}. 
\end{proof}

\begin{lemma}
 \label{lem:opcorr_methbodies_completeness}
Suppose
$P = \ov{D} \prog d$,
such that
$P \ok$ and $\vdash P \mapsto
\ov{D^\dagger} \prog d^\dagger$ and
$d^\dagger= (x:t_S,\ov{x:t}).e^\dagger$.
Assume $\eta \vdash m(\psi) \mapsto m^\dagger$,
  $m^\dagger = \monoid{m(\psi[\eta])}$,
    $\eta \vdash \tau \mapsto t^\dagger$,
  $\eta \vdash e \mapsto e^\dagger$, 
  with $\eta \vdash \tau_S \mapsto t_S^\dagger$ and  
  $\eta \vdash \ov{\tau \mapsto t^\dagger}$. 
  Then
  $\mathit{body}(t^\dagger.m^\dagger) = (x:t_S^\dagger,\ov{x:t^\dagger}).e^\dagger$
  implies 
  $\mathit{body}(\tau.m(\psi)) = (x:\tau_S,\ov{x:\tau}).e$
\end{lemma}
\begin{proof}
By investigating the derivation of $\mbody$ in
Figure~\ref{fig:fg-aux},
\[
\mbody(t^\dagger.m^\dagger) = (x:t_S^\dagger,\ov{x:t^\dagger}).e^\dagger
\]
is derived from
\[
\func~(x~t_S^\dagger)~m^\dagger(\ov{x~t^\dagger})~t^\dagger~\br{\return~e^\dagger}\in \ov{D}^\dagger
\]
Then by investigating the derivations of \NYRULENAME{t-func}, 
we have:
\[
x : t_S^\dagger \comma \ov{x : t^\dagger} \vdash e : u^\dagger
\quad \mbox{and} \quad u^\dagger \imp t^\dagger
\]
By applying Theorem~\ref{thm:mono_refl_subt},
we have 
\[
\emptyset\vdash \tau_3 \imp \tau
\]
with
$\eta \vdash \tau_3 \mapsto u^\dagger$.
Applying the inductive hypothesis, we have
\[
\func~(x~\tau_S)~m(\psi)(\ov{x~\tau})~\tau~\br{\return~e}\in \ov{D}
\]
Then applying the $\mbody$ rule in Figure~\ref{fig:fgg-aux},  
we have
\[ \mathit{body}(\tau.m(\psi)) = (x:\tau_S,\ov{x:\tau}).e\]
as required. 
\end{proof}

\begin{theorem}[Monomorphisation preserves and reflects reductions]
 \label{thm:op:comp}
Assume
  $ P = \ov{D} \prog d$,
  such that
  $P \ok$ and $\vdash P \mapsto
  \ov{D^\dagger} \prog d^\dagger $.
  
  Then: (a)
  if $d \becomes e$ then there exists $e^\dagger$
  such that $d^\dagger \becomes e^\dagger$ and
  $\emptyset \vdash e \mapsto e^\dagger$;
(b)
  if $d^\dagger \becomes e'$ then there exists $e$ such that
    $d\becomes e$, 
  $\emptyset \vdash  e \mapsto e^\dagger$ and $e'=e^\dagger$. 
\end{theorem}

\begin{proof}
{\bf Proof of (a):} 
By induction on the derivation of $d \becomes e$ with a case analysis
on the last reduction rule used.

\begin{description}
\item[{\em Case $\NYRULENAME{r-field}$}]
\quad $d=\tau_S\{\ov{v}\}.f_i$
\quad $(\ov{f~\tau}) = \mathit{fields}(\tau_S)$
\quad $e= v_i$\\

Then by rule $\NYRULENAME{m-select}$ and
$\NYRULENAME{m-value}$, 
$\eta \vdash \tau_S\{\ov{v}\}.f_i \mapsto t_S^\dagger\{\ov{v^\dagger}\}.f_i$ with
       $\eta \vdash \tau_S \mapsto t_S^\dagger$
       and
       $\eta \vdash \ov{v \mapsto v^\dagger}$. On the other hand, 
       by
Lemmas~\ref{lem:app_monodeclexists}
  and~\ref{lem:app_monotdecl},  
we have that
$(\ov{f~\tau^\dagger}) = \mathit{fields}(t_S^\dagger)$.
Hence we have\\

\begin{quote}
         \begin{tabular}{ll}
           $d^\dagger=t_S^\dagger\{\ov{v^\dagger}\}.f_i$\\
           $e^\dagger=v_i^\dagger$\\
         \end{tabular}\\
\end{quote}

         \smallskip

Then by applying $\NYRULENAME{r-field}$,
we have $d^\dagger \becomes e^\dagger$, as required.\\

\item[{\em Case $\NYRULENAME{r-call}$}]
  \quad $d=v.m(\psi)(\ov{v})$
  \quad $(x:\tau_S,\ov{x:\tau}).e_0 = \mbody(\mathit{type}(v).m(\psi))$
  \quad $e=e_0[x \by v, \ov{x \by v}]$

\smallskip 
  
Then by \NYRULENAME{m-value}, we have
  $\eta \vdash v \mapsto v^\dagger$ and 
  $\eta \vdash \ov{v \mapsto v^\dagger}$, and 
by \NYRULENAME{m-call} and \NYRULENAME{m-method}, we have 
$\eta \vdash m(\psi) \mapsto m^\dagger$ with
$m^\dagger = \monoid{m(\psi[\eta])}$. 
By \lemref{lem:opcorr_methbodies}, 
we have that
 $\mbody(\mathit{type}(v^\dagger).m^\dagger) =
  (x:t_S^\dagger,\ov{x:t^\dagger}).e_0^\dagger$
  with $\eta \vdash \tau_S \mapsto t_S^\dagger$ and
  $\eta \vdash \ov{\tau \mapsto t^\dagger}$. 
  Thus, by \lemref{lem:mono_opcorr_comp}, 
  we have\\

\begin{quote}
    \begin{tabular}{ll}
      $d^\dagger  = v^\dagger.m^\dagger(\ov{v^\dagger})$\\
      $e^\dagger  = e_0^\dagger [x \by v^\dagger, \ov{x \by v^\dagger}]$
    \end{tabular}
\end{quote}

\smallskip 

Applying $\NYRULENAME{r-call}$, 
we have: $v^\dagger.m^\dagger(\ov{v^\dagger}) \becomes 
e_0^\dagger [x \by v^\dagger, \ov{x \by v^\dagger}]$. Hence we have
$d^\dagger \becomes e^\dagger$.\\
    
\item[{\em Case $\NYRULENAME{r-assert}$}]       
\quad $d=v.(\tau)$ \quad $e=v$ \quad $\vtype(v) \imp \tau$\\

Then by \NYRULENAME{m-assert}, we have:\\

\begin{quote}
\begin{tabular}{ll}
 $\eta \vdash v \mapsto v^\dagger$\\
 $\eta \vdash \tau \mapsto t^\dagger$\\
 $\eta \vdash v.(\tau) \mapsto v^\dagger.(t^\dagger)$     \\
\end{tabular}
\end{quote}

\smallskip 

By 
\lemref{lem:app_monoprestypingopenexpr} and
\lemref{lem:mono_subt}, we have that 
$\IMPLS{\mathit{type}(v^\dagger)}{t^\dagger}$.
Hence we have\\

\begin{quote}
\begin{tabular}{ll}
  $d^\dagger= v^\dagger.(t^\dagger)$     \\
  $e^\dagger = v^\dagger$\\
\end{tabular}
\end{quote}

\smallskip

Applying $\NYRULENAME{r-assert}$,
we have  $v^\dagger.(t^\dagger) \becomes v^\dagger$. Hence 
we have $d^\dagger \becomes e^\dagger$, as required. \\

\smallskip

       \item[{\em Case $\NYRULENAME{r-context}$}]
         \quad $d=E[e_1]$ \quad $e=E[e_2]$ \quad $e_1 \becomes e_2$

         \smallskip

         where $E$ is an evaluation context.
         Then by the inductive hypothesis, we have 
         $e_1^\dagger \becomes e_2^\dagger$.
         Hence there exists\\ 
         
    \begin{quote}
         \begin{tabular}{ll}
           $d^\dagger=E_0[e_1^\dagger]$\\
           $e_1^\dagger \becomes e_2^\dagger$\\
           $e^\dagger=E_0[e_2^\dagger]$\\
         \end{tabular}
    \end{quote}
         
    \smallskip

    Thus we must prove\\

    \begin{quote}
     $E_0[e_1^\dagger]=E[e_1]^\dagger$ and
     $E_0[e_2^\dagger]=E[e_2]^\dagger$  \quad ($\star$)\\
    \end{quote}

\smallskip

From ($\star$), by applying $\NYRULENAME{r-context}$, we have\\

         \begin{quote}
           \begin{tabular}{ll}
           $d^\dagger=E[e_1]^\dagger$\\
           $d^\dagger \becomes e^\dagger$\\
           $e^\dagger=E[e_2]^\dagger$\\
           \end{tabular}\\
         \end{quote}

\smallskip 
         
There are five subcases for proving ($\star$):\\

\smallskip
    
\begin{description}
    \item[{\em Subcase \NYRULENAME{rc-structure}}] 
     \quad $E[e_1] = \tau_S\br{\ov{v},e_1,\ov{e}}$
     \quad $e_1 \becomes e_2$
     \quad $E[e_2] = \tau_S\br{\ov{v},e_2,\ov{e}}$

      \smallskip

      Then we have:\\

      \smallskip

      \begin{quote}
      $E[e_1]^\dagger = t_S^\dagger\br{\ov{v}^\dagger,e_1^\dagger,\ov{e}^\dagger}$
      with $\eta \vdash \tau_S \mapsto t_S^\dagger$\\
      $e_1^\dagger \becomes e_2^\dagger$\\
       $E[e_2]^\dagger = t_S^\dagger\br{\ov{v}^\dagger,e_2^\dagger,\ov{e}^\dagger}$
      \end{quote}        

      Hence by letting $E_0[\ ] =
      t_S^\dagger\br{\ov{v}^\dagger,e_2^\dagger,\ov{e}^\dagger}$, 
      we have proved ($\star$), as
     required. \\
      
      \smallskip

    \item[{\em Subcase \NYRULENAME{rc-select}}]
     \quad $E[e_1] = e_1.f$
     \quad $e_1 \becomes e_2$
     \quad $E[e_2] = e_2.f$

      \smallskip

      Then we have:\\

      \smallskip

      \begin{quote}
      $E[e_1]^\dagger = e_1^\dagger.f$\\
      $e_1^\dagger \becomes e_2^\dagger$\\
       $E[e_2]^\dagger = e_2^\dagger.f$\\
      \end{quote}        

      Hence by letting $E_0[\ ] = [\ ].f$, 
      we have proved ($\star$), as
     required. \\
      
      \smallskip
      
    \item[{\em Subcase \NYRULENAME{rc-receive}}]
      \quad $E[e_1] = e_1.m(\psi)(\ov{e})$
      \quad $e_1 \becomes e_2$
      \quad $E[e_2] = e_2.m(\psi)(\ov{e})$\\

      \smallskip

      Then we have:\\

      \smallskip

      \begin{quote}
      $E[e_1]^\dagger = e_1^\dagger.m^\dagger(\ov{e}^\dagger)$\\
      $e_1^\dagger \becomes e_2^\dagger$\\
      $E[e_2]^\dagger = e_2^\dagger.m^\dagger(\ov{e}^\dagger)$\\
      \end{quote}        

Hence by letting $E_0[\ ] = [\ ].m^\dagger(\ov{e}^\dagger)$, we have proved ($\star$), as
     required. \\
      
      \smallskip

    \item[{\em Subcase \NYRULENAME{rc-argument}}]
     \quad $E[e_1] = v.m(\psi)(\ov{v},e_1,\ov{e})$
     \quad $e_1 \becomes e_2$
     \quad $E[e_1] = v.m(\psi)(\ov{v},e_2,\ov{e})$\\ 

           \smallskip

      Then we have:\\

      \smallskip

      \begin{quote}
      $E[e_1]^\dagger = v^\dagger.m^\dagger(\ov{v}^\dagger,e_1^\dagger,\ov{e}^\dagger)$\\
      $e_1^\dagger \becomes e_2^\dagger$\\
      $E[e_2]^\dagger = v^\dagger.m^\dagger(\ov{v}^\dagger,e_2^\dagger,\ov{e}^\dagger)$\\
      \end{quote}        

      \smallskip

      Hence by letting $E_0[\ ] =
      v^\dagger.m^\dagger(\ov{v}^\dagger,[\ ],\ov{e}^\dagger)$, 
      we have proved ($\star$), as required. \\

\smallskip 
      
    \item[{\em Subcase \NYRULENAME{rc-assert}}] 
     \quad $E[e_1] = e_1(\tau)$
     \quad $e_1 \becomes e_2$
     \quad $E[e_2] = e_2(\tau)$\\

           \smallskip

      Then we have:\\

      \smallskip

      \begin{quote}
       $E[e_1]^\dagger = e_1^\dagger(t^\dagger)$
       with $\eta \vdash \tau \mapsto t^\dagger$\\
       $e_1^\dagger \becomes e_2^\dagger$\\
      $E[e_2]^\dagger = e_2^\dagger(t^\dagger)$\\
      \end{quote}        

      \smallskip

      Hence by letting $E_0[\ ] =  [ \ ](t^\dagger)$ with
      $\eta \vdash \tau \mapsto t^\dagger$, so that 
      we have proved ($\star$), as required. \\

\smallskip

    \end{description}

         \smallskip

\end{description}
{\bf Proof of (b):} 
By induction on the derivation of $d^\dagger \becomes e^\dagger$
with a case analysis on the last reduction rule used and
inspecting the last typing rule applied for $d^\dagger$. 

\begin{description}
\item[{\em Case $\NYRULENAME{r-field}$}]
  \quad
  $d^\dagger = t_S^\dagger\{\ov{v}^\dagger\}.f_i$
  \quad 
  $(\ov{f~\tau}) = \mathit{fields}(t_S^\dagger) $
  \quad
  $d^\dagger \becomes e'$

  \smallskip

  Then we have $e' = v_i^\dagger$ by $\NYRULENAME{r-field}$.
  On the other hand, by inspecting the derivations from $\NYRULENAME{m-value}$ and  $\NYRULENAME{m-select}$, we have\\

\begin{quote}
    \begin{tabular}{ll}
      $\eta \vdash v \mapsto v^\dagger$\\ 
      $\eta \vdash \tau_S \mapsto t_S^\dagger$\\
      $\eta \vdash \ov{v \mapsto v^\dagger}$\\
      $\eta \vdash \tau_S\{\ov{v}\}.f_i \mapsto
       t_S^\dagger\{\ov{v^\dagger}\}.f_i$\\
       \end{tabular}
\end{quote}
  
\smallskip 
       
  By inspecting the derivation of $t_S^\dagger\{\ov{v}^\dagger\}.f_i$ by 
  $\NYRULENAME{t-field}$, 
  we have $(\ov{f~\tau}) = \mathit{fields}(\tau_S)$. 
  Applying $\NYRULENAME{r-field}$,
  we have $\tau_S\{\ov{v}\}.f_i \becomes v_i$, as required. 

\smallskip  
  
\item[\emph{Case $\NYRULENAME{r-call}$}]
  \quad 
  $d^\dagger = v^\dagger.m^\dagger(\ov{v^\dagger})$
  \quad
  $d^\dagger \becomes e'$\\

  By inspecting $\NYRULENAME{r-call}$, we have\\

  \smallskip 

\begin{quote}
  \begin{tabular}{ll}
  $v^\dagger.m^\dagger(\ov{v^\dagger}) \becomes
  e_0^\dagger[x \by v^\dagger, \ov{x \by v^\dagger}]$\\
  $e'= e_0^\dagger[x \by v^\dagger, \ov{x \by v^\dagger}]$\\
  $(x:t_S^\dagger,\ov{x:t^\dagger}).e_0^\dagger = \mathit{body}(\mathit{type}(v).m^\dagger)$\\
  \end{tabular}
\end{quote}

  \smallskip 

  By inspecting the derivation of
  $v^\dagger.m^\dagger(\ov{v^\dagger})$ and $\NYRULENAME{m-call}$,
  we have \\

  \smallskip

\begin{quote}
\begin{tabular}{ll}
  $\eta \vdash v \mapsto v^\dagger$\\
  $\eta \vdash m(\psi) \mapsto m^\dagger$\\
  $\eta \vdash v.m(\psi)(\ov{v}) \mapsto v^\dagger.m^\dagger(\ov{v^\dagger})$\\
\end{tabular}
\end{quote}

\smallskip

By \lemref{lem:opcorr_methbodies_completeness},
we have $\mathit{body}(\tau.m(\psi)) = (x:\tau_S.\ov{x:\tau}).e_0$
with $\eta \vdash e_0 \mapsto e_0^\dagger$,
$\eta \vdash \tau_S \mapsto t_S^\dagger$,
and 
$\eta \vdash \ov{\tau \mapsto t^\dagger}$. 
Applying $\NYRULENAME{r-call}$, we have  $v.m(\psi)(\ov{v}) \becomes
e_0 [x \by v, \ov{x \by v}]$.
Then by \lemref{lem:mono_opcorr_comp}, $\eta
  \vdash e_0 [x \by v, \ov{x \by v}] \mapsto
  e_0^\dagger[x \by v^\dagger, \ov{x \by v^\dagger}]$.
  Hence $d \becomes e$ as required. 

  \smallskip

\item[{\em Case $\NYRULENAME{r-assert}$}]
  \quad 
  $d^\dagger = v^\dagger.(t^\dagger)$
  \quad 
  $d^\dagger \becomes e'$
  \quad
  $\IMPLS{\mathit{type}(v^\dagger)}{t^\dagger}$\\

  Then we have $e'=v^\dagger$. Then by inspecting the derivation
  of $v^\dagger.(t^\dagger)$ by $\NYRULENAME{m-asssert}$,
  we have\\ 

  \centerline{
    \begin{tabular}{ll}
      $\eta \vdash v \mapsto v^\dagger$\\ 
      $\eta \vdash \tau \mapsto t^\dagger$\\
      $\eta \vdash v.(\tau) \mapsto v^\dagger.(t^\dagger)$\\
  \end{tabular}}

  \smallskip
    
  Then applying \lemref{lem:app_monoprestypingopenexpr}, 
  we have that $\eta\vdash \mathit{type}(v) \mapsto \mathit{type}(v^\dagger)$.
  Then by Theorem~\ref{thm:mono_refl_subt}, we have
  $\IMPLS{\mathit{type}(v)}{t}$. Applying 
  $\NYRULENAME{r-assert}$,  we have $v.(t) \becomes v$, as desired. 

\smallskip
  
\item[{\em Case $\NYRULENAME{r-context}$}]
 \quad $d^\dagger = E[e_1]^\dagger$
 \quad $E[e_1]^\dagger = E_0[e_1^\dagger]$
 \quad $d^\dagger \becomes e'$\\

  By investigating the derivation from $\NYRULENAME{r-context}$, we have\\

  \centerline{
  \begin{tabular}{ll}
  $e_1^\dagger \becomes e_2'$\\
  $e' = E_0[e_2']$
  \end{tabular}}

  \smallskip
  We have to prove:\\

  \begin{quote}
    $E_0[e_2'] = E_0[e_2^\dagger] = E[e_2]^\dagger$
    such that $e_1 \becomes e_2$ 
    \quad ($\star$)\\
  \end{quote}
  From ($\star$), we have $E[e_1]\becomes E[e_2]$ by
  $\NYRULENAME{r-context}$. \\

  \smallskip 
  
  There are five subcases for proving ($\star$).

  \begin{description}
  \item[{\em Subcase \NYRULENAME{rc-structure}}]
   \quad $E[e_1]^\dagger = t_S^\dagger\br{\ov{v}^\dagger,e_1^\dagger,\ov{e}^\dagger}$   \quad $E[e_1]^\dagger = E_0[e_1^\dagger]$
   \quad $E_0[e_1^\dagger] \becomes e'$ 

    \smallskip

    with $\eta \vdash \ov{v\mapsto v}^\dagger$, 
    $\eta \vdash \ov{e}\mapsto \ov{e}^\dagger$, and 
    $\eta \vdash \tau_S \mapsto t_S^\dagger$.
    Then by inductive hypothesis on $e_1^\dagger$,
    we have $e_1^\dagger \becomes e_2' = e_2^\dagger$.  
    Hence $e' = t_S^\dagger\br{\ov{v}^\dagger,e_2^\dagger,\ov{e}^\dagger} = 
    E_0[e_2^\dagger] = E[e_2]^\dagger$ and
    $e_1 \becomes e_2$, as required.\\

   \smallskip
    
   \item[{\em Subcase \NYRULENAME{rc-select}}]
     \quad $E[e_1]^\dagger = e_1^\dagger.f$
     \quad $E[e_1]^\dagger= E_0[e_1^\dagger]$
     \quad $e_1^\dagger \becomes e_2'$

     \smallskip 

     Then by inductive hypothesis on $e_1^\dagger$, 
     we have $e_1^\dagger \becomes e_2' = e_2^\dagger$. 
     Hence $e' = e_2^\dagger.f = E_0[e_2^\dagger] = E[e_2]^\dagger$ and
    $e_1 \becomes e_2$, as required.\\

     \smallskip

   \item[{\em Subcase \NYRULENAME{rc-receive}}]
     \quad $E[e_1]^\dagger = e_1^\dagger.m^\dagger(\ov{e}^\dagger)$
     \quad $E[e_1]^\dagger = E_0[e_1^\dagger]$
     \quad  $e_1^\dagger \becomes e_2'$\\

      \smallskip

      with $\eta \vdash m(\psi) \mapsto m^\dagger$,
      $\eta \vdash e_1 \mapsto e_1^\dagger$ and  
      $\eta \vdash \ov{e \mapsto e^\dagger}$.  
    Then by inductive hypothesis on $e_1^\dagger$,
    we have $e_1^\dagger \becomes e_2' = e_2^\dagger$.  
      Then we have $E[e_1]^\dagger =
      E_0[e_1^\dagger]=
      e_1^\dagger.m^\dagger(\ov{e}^\dagger) \becomes
      e_2^\dagger.m^\dagger(\ov{e}^\dagger) =
      E_0[e_2^\dagger]= E[e_2]^\dagger$ with 
      $e_1\becomes  e_2$, as required. \\
      
   \smallskip
     
   \item[{\em Subcase \NYRULENAME{rc-argument}}]
     \quad $E[e_1]^\dagger =
     v^\dagger.m^\dagger(\ov{v}^\dagger,e_1^\dagger,\ov{e}^\dagger)$
     $E[e_1]^\dagger = E_0[e_1^\dagger]$
     \quad $e_1^\dagger \becomes e_2'$

     \smallskip
     
      with $\eta \vdash m(\psi) \mapsto m^\dagger$,
      $\eta \vdash e_1 \mapsto e_1^\dagger$ and  
      $\eta \vdash \ov{e \mapsto e^\dagger}$. 
      Then by inductive hypothesis on
      $e_1$, we have
      $e_1^\dagger \becomes e_2'$ with $e_2'=e_2^\dagger$.
      Then we have
      $E[e_1]^\dagger =
      v^\dagger.m^\dagger(\ov{v}^\dagger,e_1^\dagger,\ov{e}^\dagger)
      \becomes
      v^\dagger.m^\dagger(\ov{v}^\dagger,e_2^\dagger,\ov{e}^\dagger)
      =E_0[e_2^\dagger] = E[e_2]^\dagger$ with 
      $e_1\becomes  e_2$, as required. \\

      \smallskip

   \item[{\em Subcase \NYRULENAME{rc-assert}}]  
     \quad $E[e_1]^\dagger = e_1^\dagger(t^\dagger)$
     \quad $E_0[e_1^\dagger]= e_1^\dagger(t^\dagger)$
     \quad $e_1 \becomes e_2'$
     
     \smallskip

      Then by inductive hypothesis on
      $e_1$, we have
      $e_1^\dagger \becomes e_2'$ with $e_2'=e_2^\dagger$.
     Hence we have 
     $E[e_1]^\dagger = e_1^\dagger(t^\dagger) \becomes  e_2^\dagger(t^\dagger)
     = E_0[e_2^\dagger]= E[e_2]^\dagger$ with
     $e_1\becomes  e_2$, as required.
\end{description}
\end{description}
\end{proof}

\end{document}
